\documentclass[acmsmall,screen]{acmart}

\usepackage{booktabs}                        
\usepackage{subcaption}   
\usepackage{color}
\usepackage{amsmath}
\usepackage{graphicx}
\usepackage{amsthm} 
\usepackage[all]{xy}
\usepackage{multirow}
\usepackage{wrapfig}
\usepackage{listings,lstcoq}
\definecolor{ltblue}{rgb}{0,0.4,0.4}
\definecolor{dkblue}{rgb}{0,0.1,0.6}
\definecolor{dkgreen}{rgb}{0,0.35,0}
\definecolor{dkviolet}{rgb}{0.3,0,0.5}
\definecolor{dkred}{rgb}{0.5,0,0}
\usepackage{qcircuit}
 
\usepackage{tikzit}
\input{zx.tikzdefs}

\tikzstyle{box}=[shape=rectangle, text height=1.5ex, text depth=0.25ex, yshift=0.5mm, fill=white, draw=black, minimum height=5mm, yshift=-0.5mm, minimum width=5mm, font={\small}]

\tikzstyle{Z dot}=[inner sep=0mm, minimum size=2mm, shape=circle, draw=black, fill={rgb,255: red,216; green,248; blue,216}]
\tikzstyle{Z phase dot}=[minimum size=5mm, font={\footnotesize\boldmath}, shape=rectangle, rounded corners=2mm, inner sep=0.2mm, outer sep=-2mm, scale=0.8, tikzit shape=circle, draw=black, fill={rgb,255: red,216; green,248; blue,216}, tikzit draw=blue]
\tikzstyle{X dot}=[Z dot, shape=circle, draw=black, fill={rgb,255: red,255; green,136; blue,136}]
\tikzstyle{X phase dot}=[Z phase dot, tikzit shape=circle, tikzit draw=blue, fill={rgb,255: red,255; green,136; blue,136}, font={\footnotesize\color{black}\boldmath}]

\tikzstyle{hadamard}=[fill=yellow, draw=black, shape=rectangle, inner sep=0.6mm, minimum height=1.5mm, minimum width=1.5mm]
\tikzstyle{vertex}=[inner sep=0mm, minimum size=1mm, shape=circle, draw=black, fill=black]
\tikzstyle{vertex set}=[inner sep=0mm, minimum size=1mm, shape=circle, draw=black, fill=white, font={\footnotesize\boldmath}]

\tikzstyle{hadamard edge}=[-, color=blue, dashed, dash pattern=on 3pt off 1.5pt, thick]
\tikzstyle{brace edge}=[-, tikzit draw=blue, decorate, decoration={brace,amplitude=1mm,raise=-1mm}]
\tikzstyle{diredge}=[->]

\def\>{\ensuremath{\rangle}}
\def\<{\ensuremath{\langle}}

\newcommand {\E } {{\mathcal{E}}}

\newcommand{\hs}{\mathcal{H}}
\newcommand {\tr} {{\mathit{tr}}}

\newcommand {\sem}[1] {{\llbracket #1 \rrbracket}}
\newcommand{\smeet}{{\:\&\:}}
\newcommand{\QVar}{{\mathit{QV}}}
\newcommand{\qvar}{{\mathit{qv}}}
\newcommand{\type}{{\mathit{type}}}
\newcommand{\gateX}{{\mathsf{X}}}

\newcommand{\gateZ}{{\mathsf{Z}}}
\newcommand{\gateH}{{\mathsf{H}}}

\newcommand{\gateSWAP}{{\mathsf{SWAP}}}
\newcommand{\gateCNOT}{{\mathsf{CNOT}}}
\newcommand{\gateDeutsch}{{\mathsf{Deutsch}}}
\newcommand{\gateToffoli}{{\mathsf{Toffoli}}}
\newcommand{\gateFredkin}{{\mathsf{Fredkin}}}
\newcommand{\reason}[1]{{({\textnormal{By #1}})}}
\newcommand{\qbar}{{\overline{q}}}
\newcommand{\rbar}{{\overline{r}}}
\newcommand{\iskip}{{\mathbf{skip}}}
\newcommand{\iabort}{{\mathbf{abort}}}
\newcommand{\kwhile}{{\mathbf{while}}}
\newcommand{\kif}{{\mathbf{if}}}

\newcommand{\iinit}[2]{%
  \if\relax\detokenize{#2}\relax
    {{#1} := |0\>}
  \else
    {{#1} := {#2}}
  \fi
}

\newcommand{\iqif}[3]{%
  \if\relax\detokenize{#3}\relax
    {\mathbf{qif}\,[{#1}]\left({#2}\right)\mathbf{fiq}}
  \else
    {\mathbf{qif}\,[{#1}]\left({#2}\rightarrow{#3}\right)\mathbf{fiq}}
  \fi
}

\newcommand{\iif}[4]{%
  \if\relax\detokenize{#1}\relax
    \if\relax\detokenize{#4}\relax
        {\mathbf{if} {#2}\left({#3}\right)\mathbf{fi}}
    \else
        {\mathbf{if}\,[{#2}]\left(\square{#3}\rightarrow{#4}\right)\mathbf{fi}}
    \fi
  \else
    \if\relax\detokenize{#4}\relax
        {\mathbf{if}\,{#1}[{#2}]\left({#3}\right)\mathbf{fi}}
    \else
        {\mathbf{if}\,{#1}[{#2}]\left(\square{#3}\rightarrow{#4}\right)\mathbf{fi}}
    \fi
  \fi
}

\newcommand{\iwhile}[3]{
    {\mathbf{while}\,{#1}[{#2}]=1\,\mathbf{do}\ {#3}\ \mathbf{od}}
}

\newcommand{\sqif}[5]{%
    \if\relax\detokenize{#1}\relax
    \else
        \if\relax\detokenize{#4}\relax
            {#1}{\:\leftarrow\:}
        \else
            {#1}\stackrel{#4}{\leftarrow}
        \fi
    \fi
    {#2}
    \if\relax\detokenize{#3}\relax
    \else
        \if\relax\detokenize{#5}\relax
            {\:\rightarrow\:}{#3}
        \else
            \stackrel{#5}{\rightarrow}{#3}
        \fi
    \fi
}

\newcommand{\sif}[3]{
    \if\relax\detokenize{#1}\relax
    \else
        {#1}{\:\vartriangleleft\:}
    \fi
    {#2}
    \if\relax\detokenize{#3}\relax
    \else
        {\:\vartriangleright\:}{#3}
    \fi
}

\newcommand{\swhile}[2]{{#1}{\:\ast\:}{#2}}

\newtheorem{thm}{Theorem}[section]

\newtheorem{lem}{Lemma}[section]
\newtheorem{defn}{Definition}[section]
\newtheorem{prop}{Proposition}[section]
\newtheorem{exam}{Example}[section]

\newcommand{\bt}[1]{{{\color{blue}#1}}}
\newcommand{\rt}[1]{{{\color{red}#1}}}

\newcommand{\qbricks}{{\textsc{Qbricks}}}
\newcommand{\qwire}{{$\mathcal{Q}$\textsc{wire}}}
\newcommand{\reqwire}{{$Re\mathcal{Q}$\textsc{wire}}}

\acmJournal{TOSEM}

\setcopyright{none}

\begin{document}
 
\title{Laws of Quantum Programming}

\author{Mingsheng Ying}
\authornote{The first two authors contribute equally.}
\orcid{0000-0003-4847-702X}
 \affiliation{    
  \institution{Centre for Quantum Software and Information, University of Technology Sydney}          
    \city{Ultimo}
   \postcode{}
  \country{Australia}
}
\email{Mingsheng.Ying@uts.edu.au}       
 
\author{Li Zhou}
\authornotemark[1]
\orcid{0000-0002-9868-8477}
 \affiliation{
  \department{Key Laboratory of System Software (Chinese Academy of Sciences) and State Key Laboratory of Computer Science} 
  \institution{Institute of Software, Chinese Academy of Sciences}   
  \city{Beijing}       
  \country{China}
}
\email{zhouli@ios.ac.cn}
\email{zhou31416@gmail.com}

\author{Gilles Barthe}
\orcid{0000-0002-3853-1777}
 \affiliation{
  \institution{Max Planck Institute for Security and Privacy}   
  \city{Bochum}       
  \country{Germany}
}
\affiliation{
  \institution{IMDEA Software Institute}            
  \city{Madrid}
  \country{Spain}
}
\email{gilles.barthe@mpi-sp.org}


\begin{abstract} In this paper, we investigate the fundamental laws of quantum programming. We extend a comprehensive set of Hoare et al.'s basic laws of classical programming to the quantum setting. These laws characterise the algebraic properties of quantum programs, such as the distributivity of sequential composition over (quantum) $\kif$-statements and the unfolding of nested (quantum) $\kif$-statements. At the same time, we clarify some subtle differences between certain laws of classical programming and their quantum counterparts. Additionally, we derive a fixpoint characterisation of quantum $\kwhile$-loops and a loop-based realisation of tail recursion in quantum programming. Furthermore, we establish two normal form theorems: one for quantum circuits and one for finite quantum programs. The theory in which these laws are established is formalised in the Coq proof assistant, and all of these laws are mechanically verified. As an application case of our laws, we present a formal derivation of the principle of deferred measurements in dynamic quantum circuits. 

We expect that these laws can be utilised in correctness-preserving transformation, compilation, and automatic code optimisation in quantum programming. In particular, because these laws are formally verified in Coq, they can be confidently applied in quantum program development.
\end{abstract}



\maketitle
 
\section{Introduction}
It is widely accepted that scientists and mathematicians uncover the fundamental laws of nature and mathematical structures, respectively, while engineers leverage these discoveries to address practical problems. 
This fundamental concept spurred Hoare et al.~\cite{Hoare87, Hoare88, Hoare93}  to pose a crucial question: \begin{itemize}
    \item \textit{What are the laws governing computer programming that establish the formal foundation for the field of software engineering?}
\end{itemize}
Their argument, rooted in the perspective that computer programs are essentially mathematical expressions, led them to establish a series of algebraic laws for programming, ranging from sequential to concurrent processes. 
The foundational laws articulated in~\cite{Hoare87, Hoare88, Hoare93} serve as a cornerstone in the establishment of the science of computer programming. 
Simultaneously, they prove invaluable in the realm of software engineering, encompassing tasks such as compilation, optimisation, and program transformation. 
For instance, during program development, a programmer may create a program aligning with end-user requirements. 
Utilising these laws, an optimising compiler can then transform it into an efficient program that effectively exploits the underlying architecture of the target machine (refer to~\cite{Hoare93} for a detailed example). 
This strategic approach has the advantage of ensuring correctness by construction during compilation. 
Each transformation within the process is mathematically proven based on these laws, providing a robust framework for the development and optimisation of computer programs.

The past decade has witnessed remarkable advancements in quantum computer hardware, leading to a surge in research on quantum programming languages and platforms~\cite{Qiskit, Cirq, ProjectQ, Omer05, LanQ, GKMW10, GP13, Qsharp, Quipper, Scaffold, TKet, isQ}. 
The research on quantum programming also spans from theoretical aspects, such as formal semantics, program logic, static analysis, and verification, to more practical considerations like compilation, optimisation, testing, and debugging. 
Despite this progress, the knowledge base of quantum programming as a scientific discipline remains fragmented and disjointed. 
Consequently, two fundamental questions naturally arise: 
\begin{itemize}
    \item \textit{What are the foundational laws that underpin the methodologies and techniques of quantum programming?}
    \item \textit{Are there discernible differences between these laws and those governing classical programming, as established in works like~\cite{Hoare87,Hoare88,Hoare93}?}
\end{itemize}
Addressing these questions is crucial for consolidating the understanding of quantum programming as a coherent and well-defined field, bridging the gap between theoretical insights and practical applications.
\begin{figure}
    \centering
    $\begin{array}{ccccc}\mbox{Classical-Quantum\ Hybrid\ Programs}\\ \Uparrow\\ {\rm Purely\ Quantum\ Programs}\\ \Uparrow\\ {\rm Quantum\ Circuits}
    \end{array}
    $
\caption{Three-Layer Framework for Quantum Programming.}
\label{fig 0}
\end{figure} 

The objective of this paper is to delve into fundamental inquiries and identify key laws that can establish an indispensable scientific foundation for quantum programming. To accomplish this, we present a comprehensive three-layer framework for quantum programming in Figure~\ref{fig 0}, 
with each layer intricately embedded into the one above: 
\begin{itemize}
    \item \textit{Quantum Circuits Description Layer}: 
    At the bottom is a description language for quantum circuits. 
    Assuming a set of basic gates, any quantum circuit can be defined using sequential composition and a quantum $\kif$-statement~\cite{Alt05, Ying12, Yuan, Bich, Voi23}.
    \item \textit{Purely Quantum Programs Layer}: 
    The middle layer introduces a language for writing purely quantum programs (without classical variables). 
    Building upon quantum circuits, an $\kif$-statement with the outcome of a quantum measurement as its condition is introduced. 
    This layer also defines the loop and recursion of quantum programs (with classical control flow). 
    Nondeterministic choice is added for the specification and refinement of quantum programs. 
    \item \textit{Classical-Quantum Hybrid Programming Layer}: 
    At the top is the classical-quantum hybrid programming layer, where purely quantum programs are embedded into a classical programming language. 
    This integration enables the programming of classical and quantum computation together, a necessity for most practical quantum algorithms.
\end{itemize}
While the components of this framework are not novel and have been extensively studied in previous literature, we organise them into a unified framework that clearly delineates the application of each law of quantum programming.

Within this three-layer framework, we establish a comprehensive set of basic laws of quantum programming. 
This includes two subsets of laws:
\begin{itemize}
    \item The first subset is proven for quantum circuits at the bottom layer of Figure~\ref{fig 0}. 
    Among them, several laws about the distributivity of quantum $\kif$-statements are particularly useful. 
    \item The second subset is established for purely quantum programs at the middle layer of Figure~\ref{fig 0}. 
    A congruence property is introduced, facilitating the lifting of laws for quantum circuits to quantum programs. 
    Then a diverse range of laws governing $\kif$-statements, loops, recursion, and nondeterministic choice are derived.   
\end{itemize}
These quantum programming laws can be further integrated with existing laws for classical (and probabilistic) programs~\cite{Hoare87, Hoare88, Hoare93, He97}. This integration enhances the framework's utility for reasoning about classical-quantum hybrid programs at the top layer of Figure~\ref{fig 0}.

Noteworthy technical contributions of this paper encompass:
\begin{itemize}
    \item \textit{\textbf{Normal forms}}: 
    We prove that every quantum circuit is equivalent to a normal form featuring a sequential composition of flat (but large) quantum $\kif$-statements. 
    A similar normal form (with an $\kif$-statement) is presented for finite quantum programs. 
    The significance of the normal forms lies in their applications, including providing a compilation strategy of quantum programs as a generalisation of that given in~\cite{Hoare93} and offering a method for checking the equivalence of two quantum programs by reducing them to a normal form.
    \item \textit{\textbf{Tail Recursion}}: 
    It is well-known in classical programming that tail recursion can be realised using loops. 
    We prove that this result holds true for recursive quantum programs with classical control flow. 
    A practical implication of this finding is that quantum programs with tail recursion can be easily implemented and optimised during compilation.
    \item \textbf{\textit{Verification in Coq}}: 
    We formalise the theory in which our laws are established in the Coq proof assistant, by leveraging CoqQ~\cite{Zhou23}, a general-purpose project for quantum program verification. 
    The formalisation from the first principle ensures the soundness of the laws presented in this paper. 
    We expect that they will be confidently and frequently used in quantum development. As observed in~\cite{BGZ10, Barb21}, algebraic laws like these are also very helpful in the security verification of quantum cryptography.
\end{itemize}

\subsection{An Illustrative Example}\label{illu-exam}

The laws proposed in this paper are anticipated to play a role in quantum programming, akin to that played by the laws presented in~\cite{Hoare87, Hoare88, Hoare93} in classical programming. 
To explore this prospect, this subsection demonstrates a simple example applying our proposed laws.

Quantum error correction (QEC) techniques are the key to large-scale and practical quantum computing. 
Usually, QEC codes can be conveniently written as quantum programs. 
Let us consider the simplest QEC code---the \textit{three qubit bit-flip code} (see~\cite{NC00}, Section 10.1.1). 
The following example shows how the correctness of this code can be derived directly using our laws.  

\begin{exam}[QEC Code]
    Suppose that we want to send a qubit $q$ in state $|\psi\>=\alpha|0\>+\beta|1\>$ through a channel that flips from $|0\>$ to $|1\>$ (and vice versa) with a certain probability, called a \textit{bit flip channel}. The basic idea of QEC is: 
    \begin{enumerate}
        \item (\textbf{Encoding}) Encode the state $|\psi\>$ in three qubits $q,q_1$ and $q_2$ as $\alpha|000\>+\beta|111\>$. This can be done by initialising $q_1,q_2$ both in basis state $|0\>$, and then performing the controlled-\textsf{NOT} gate twice, both with $q$ as the control qubit, but the first with $q_1$ and the second with $q_2$ as the target qubit.
        \item (\textbf{Sending via noisy channel}) Each of the three qubits is sent through an independent bit flip channel. Assume that a bit flip happens only on one or none of the qubits.
        \item (\textbf{Decoding}) Perform the controlled-\textsf{NOT} gate twice, both with $q$ as the control qubit, but the first with $q_2$ and the second with $q_1$ as the target qubit.
        \item (\textbf{Error correction})  Measure qubit $q_2$ in the computational basis: if the outcome is $0$, then the procedure terminates. Otherwise, we measure qubit $q_1$ also in the computational basis: if the outcome is $0$, then terminate; otherwise, perform the \textsf{NOT} gate on qubit $q$. Additional \textsf{NOT} gates are applied to $q_1$ and $q_2$ conditioned on the measurement outcomes to recover them to their initial state.
    \end{enumerate}
    As suggested in~\cite{Feng23}, the above entire process can be appropriately modelled as a nondeterministic quantum program:    
    \begin{align*}
        \mathit{QEC} \triangleq \ 
        {\color{gray}
        \begin{aligned}
            1 \\
            2 \\
            3
        \end{aligned}
        \left|\ \
        {\color{black}
        \begin{aligned}
            &\iinit{q_1}{}; \iinit{q_2}{}; \gateCNOT[q,q_1]; \gateCNOT[q,q_2]; \\
            &(\iskip\sqcup \gateX[q]\sqcup \gateX[q_1]\sqcup \gateX[q_2]); \\
            &\gateCNOT[q,q_2]; \gateCNOT[q,q_1]; 
            \sif{}{q_2}{(\sif{}{q_1}{(\gateX[q]; \gateX[q_1])}; \gateX[q_2])}.
        \end{aligned}
        }
        \right.
        }
    \end{align*} 
    Line 1 represents the encoding part where \(\gateCNOT\) denotes the controlled-\textsf{NOT} gate. 
    Line 2 models the noise where $\gateX$ denotes the \textsf{NOT} gate and $\sqcup$ stands for nondeterministic choice; thus, a bit flip occurs in at most one of the three channels. 
    Line 3 corresponds to the decoding and error correction part, 
    where $\sif{}{r}{P}$ denotes the conditional statement that if the outcome of the computational measurement on a qubit $r$ is $1$, then execute $P$.
\end{exam}
It was verified in~\cite{Feng23} using quantum Hoare logic that the program $\mathit{QEC}$ is correct: if $q$ in state $|\psi\>$ is inputted to it, then $q$ will be in the same state $|\psi\>$ at the output of $\mathit{QEC}$ (hence, error was corrected). 
Here, we prove the correctness of $\mathit{QEC}$ by \textit{program transformation} using only the laws established in this paper. 

Let $\sqif{C_0}{q}{C_1}{}{}$ be a quantum $\kif$-statement (see Eq. (\ref{q-control}) for its definition). 
First, we have:
\begin{equation}
\label{illu-exam-1}
    \begin{aligned}
        &\gateCNOT[q,q_2]; \gateX[q]; \gateCNOT[q,q_2]\\
        =\ &\sqif{\iskip}{q}{\gateX[q_2]}{}{}; \gateX[q]; \sqif{\iskip}{q}{\gateX[q_2]}{}{}\hspace{-3cm} & \\
        \equiv\ &\gateX[q]; \sqif{\gateX[q_2]}{q}{\iskip}{}{}; \sqif{\iskip}{q}{\gateX[q_2]}{}{} 
            & \reason{Prop.~\ref{choice-laws}(\ref{choice-laws-symmetry})} \\
        \equiv\ &\gateX[q]; \sqif{(\gateX[q_2]; \iskip)}{q}{(\iskip; \gateX[q_2])}{}{}
            & \reason{Prop.~\ref{circuit-seq-law}(\ref{circuit-seq-law-sequentiality})}\\
        \equiv\ &\gateX[q]; \sqif{\gateX[q_2]}{q}{\gateX[q_2]}{}{} 
            & \reason{Prop.~\ref{circuit-seq-law}(\ref{circuit-seq-law-unit})}\\
        \equiv\ &\gateX[q]; \gateX[q_2]. 
            & \reason{Prop.~\ref{qif-law}(\ref{qif-law-idempotence})}
    \end{aligned}
\end{equation}
Then we obtain: 
\begin{equation}
\label{illu-exam-2}
    \begin{aligned}
        &\gateCNOT[q,q_1]; \gateCNOT[q,q_2]; \gateX[q]; \gateCNOT[q,q_2]; \gateCNOT[q,q_1]\hspace{-3cm} &\\ 
        \equiv\ &\gateCNOT[q,q_1]; \gateX[q]; \gateX[q_2]; \gateCNOT[q,q_1]&\\
        \equiv\ &\gateCNOT[q,q_1]; \gateX[q]; \gateCNOT[q,q_1]; \gateX[q_2]
            & \reason{Prop.~\ref{circuit-seq-law}(\ref{circuit-seq-law-commutativity})}\\
        \equiv\ &\gateX[q]; \gateX[q_1]; \gateX[q_2].
            & \reason{an argument similar to (\ref{illu-exam-1})}
    \end{aligned}
\end{equation}
Likewise, we can prove:
\begin{equation}
\label{illu-exam-3}
    \begin{cases}
        \gateCNOT[q,q_1]; \gateCNOT[q,q_2]; \iskip; \gateCNOT[q,q_2]; \gateCNOT[q,q_1] \equiv\iskip,\\
        \gateCNOT[q,q_1]; \gateCNOT[q,q_2]; \gateX[q_1]; \gateCNOT[q,q_2]; \gateCNOT[q,q_1]\equiv \gateX[q_1],\\
        \gateCNOT[q,q_1]; \gateCNOT[q,q_2]; \gateX[q_2]; \gateCNOT[q,q_2]; \gateCNOT[q,q_1]\equiv \gateX[q_2].
    \end{cases}
\end{equation}
On the other hand, we have:
\begin{equation}
\label{illu-exam4}
    \begin{aligned}
        & \iinit{q_1}{}; \iinit{q_2}{}; \gateX[q]; \gateX[q_1]; \gateX[q_2]; \sif{}{q_2}{(\sif{}{q_1}{(\gateX[q]; \gateX[q_1])}; \gateX[q_2])}\hspace{-3cm} &\\ 
        \equiv\ &\gateX[q]; (\iinit{q_1}{}; \gateX[q_1]); (\iinit{q_2}{}; \gateX[q_2]); &\\
            & \sif{}{q_2}{(\sif{}{q_1}{(\gateX[q]; \gateX[q_1])}; \gateX[q_2])} 
            & \reason{Prop.~\ref{sequential-laws}(\ref{sequential-laws-commutativity})}\\
        \equiv\ &\gateX[q]; \iinit{q_1}{|1\>}; \iinit{q_2}{|1\>}; \sif{}{q_2}{(\sif{}{q_1}{(\gateX[q]; \gateX[q_1])}; \gateX[q_2])} 
            & \reason{Prop.~\ref{init-laws}(\ref{init-laws-unitary-elimination})}\\
        \equiv\ &\gateX[q]; \iinit{q_1}{|1\>}; \iinit{q_2}{|1\>}; (\sif{}{q_1}{(\gateX[q]; \gateX[q_1])}; \gateX[q_2])
            & \reason{Prop.~\ref{init-laws}(\ref{init-laws-if-elimination})}\\
        \equiv\ &\gateX[q]; \iinit{q_2}{|1\>}; \iinit{q_1}{|1\>}; \sif{}{q_1}{(\gateX[q]; \gateX[q_1])}; \gateX[q_2]
            & \reason{Prop.~\ref{sequential-laws}(\ref{sequential-laws-commutativity})}\\
        \equiv\ &\gateX[q]; \iinit{q_2}{|1\>}; \iinit{q_1}{|1\>}; (\gateX[q]; \gateX[q_1]); \gateX[q_2]
            & \reason{Prop.~\ref{init-laws}(\ref{init-laws-if-elimination})}\\
        \equiv\ &(\iinit{q_1}{|1\>}; \gateX[q_1]); (\iinit{q_2}{|1\>}; \gateX[q_2]); (\gateX[q]; \gateX[q])
            & \reason{Prop.~\ref{sequential-laws}(\ref{sequential-laws-commutativity})}\\
        \equiv\ & \iinit{q_1}{}; \iinit{q_2}{}.
            & \hspace{-3cm} \reason{Props.~\ref{sequential-laws}(\ref{sequential-laws-commutativity}),~\ref{circuit-seq-law}(\ref{circuit-seq-law-unit}--\ref{circuit-seq-law-composition})}
    \end{aligned}
\end{equation}
Similarly, we obtain: 
\begin{equation}
\label{illu-exam-5}
    \!\:
    \begin{cases}
        \iinit{q_1}{}; \iinit{q_2}{}; \iskip; \sif{}{q_2}{(\sif{}{q_1}{(\gateX[q]; \gateX[q_1])}; \gateX[q_2])}\equiv \iinit{q_1}{}; \iinit{q_2}{},\\
        \iinit{q_1}{}; \iinit{q_2}{}; \gateX[q_1]; \sif{}{q_2}{(\sif{}{q_1}{(\gateX[q]; \gateX[q_1])}; \gateX[q_2])}\equiv \iinit{q_1}{}; \iinit{q_2}{},\\
        \iinit{q_1}{}; \iinit{q_2}{}; \gateX[q_2]; \sif{}{q_2}{(\sif{}{q_1}{(\gateX[q]; \gateX[q_1])}; \gateX[q_2])}\equiv \iinit{q_1}{}; \iinit{q_2}{}.
    \end{cases}
\end{equation}
Finally, by combining Eqs. (\ref{illu-exam-2}) to (\ref{illu-exam-5}) and using Proposition~\ref{laws-nd}(\ref{laws-nd-distributivity-seq}), we assert:
\begin{equation}
\label{illu-exam-6}
    \begin{split}
        \mathit{QEC}
        &\equiv \iinit{q_1}{}; \iinit{q_2}{}; \big(\iskip\sqcup (\gateX[q]; \gateX[q_1]; \gateX[q_2])\sqcup \gateX[q_1]\sqcup \gateX[q_2]\big); \\ 
        &\quad \sif{}{q_2}{(\sif{}{q_1}{(\gateX[q]; \gateX[q_1])}; \gateX[q_2])}\\ 
        &\equiv (\iinit{q_1}{}; \iinit{q_2}{})\sqcup (\iinit{q_1}{}; \iinit{q_2}{}) \\
        &\quad \sqcup (\iinit{q_1}{}; \iinit{q_2}{})\sqcup (\iinit{q_1}{}; \iinit{q_2}{}) \\
        &\equiv \iinit{q_1}{}; \iinit{q_2}{}.
    \end{split}
\end{equation} 
It is clear from the last line of Eq. (\ref{illu-exam-6}) that $\mathit{QEC}$ does not change the state of $q$, and thus $q$ is still in state $|\psi\>=\alpha|0\>+\beta|1\>$ at the output of $\mathit{QEC}$. 

\subsection{Organisation of the Paper} 
The paper is organised as follows. 
For convenience of comparing the laws of quantum programming with their classical counterparts, we briefly review some basic laws of classical programming in Section~\ref{sec-classical-laws}. 
The three-layer framework of quantum programming is described in Section~\ref{sec-framework}. 
Our laws of quantum programming are developed from Sections~\ref{sec-laws-circ} through~\ref{sec-nondet}. 
In Section~\ref{sec-laws-circ}, various algebraic laws for quantum circuits at the bottom layer of our quantum programming framework are presented. 
The normal form of quantum circuits is also shown there. 
The laws for purely quantum programs at the second layer of our programming framework are introduced in Section~\ref{sec-laws-prog}. 
In particular, the normal form of finite quantum programs is proved there. 
In Section~\ref{sec-recur}, several laws for quantum loops and recursive quantum programs are established, including the fixpoint characterisation of quantum loops and the loop-based realisation of tail recursion in quantum programming. 
Some laws for deterministic quantum programs given in Sections~\ref{sec-laws-prog} are generalised to nondeterministic quantum programs in Section~\ref{sec-nondet}. 
Several new laws about nondeterministic quantum programs are also given there. In Section~\ref{Sec-Refine}, we present a series of laws for refinement of quantum programs. 
To show the effect of the laws developed in this paper, a more sophisticated example is given in Section~\ref{sec-app} where the principle of deferred measurements is formally derived using our laws. 
The formalisation of all these laws in the Coq proof assistant is presented in Section~\ref{sec-coq}. The related work is discussed in Section~\ref{sec-related}. 
A brief conclusion is drawn in Section~\ref{sec-concl}, where we also point out some topics for future research.

\section{Laws of Classical Programming}\label{sec-classical-laws}

In this section, we briefly recall the laws of classical programming established in~\cite{Hoare87, Hoare88, Hoare93}. 
On the one hand, these laws motivate our laws of quantum programming. 
On the other hand, a careful comparison between these laws and our laws of quantum programming presented in the subsequent sections will help us to understand better some basic (but subtle) differences between classical and quantum program constructs. 

The language employed in~\cite{Hoare87} to formulate the laws is Dijkstra's nondeterministic sequential programming language~\cite{Dij76}, which has the $\iskip$, $\iabort$ commands, assignment $x:=E$, sequential composition $P; Q$, conditional $\sif{P}{b}{Q}$ (standing for $\mathbf{if}\ b\ \mathbf{then}\ P\ \mathbf{else}\ Q$), nondeterministic choice $P\sqcup Q$, iteration $\swhile{b}{P}$ (standing for $\mathbf{while}\ b\ \mathbf{do}\ P$), and recursion $\mu X.F(X)$. 

The main laws established in~\cite{Hoare87, Hoare88, Hoare93} are presented in Figure~\ref{fig 1}. 
Most of these laws will be generalised in this paper to both the layers of quantum circuits and purely quantum programs (see Figure~\ref{fig 0}). 
The connections between these laws and our laws of quantum programming will be discussed whenever the latter are presented. 
At the third layer in Figure~\ref{fig 0}, the classical laws in Figure~\ref{fig 1} and our quantum laws can be combined together to reason about classical-quantum hybrid programs. 
\begin{figure}\centering
    \begin{align*}
        &(\mbox{If-1})\ \sif{P}{\mathit{true}}{Q} \equiv P \qquad 
         (\mbox{If-2})\ \sif{P}{\mathit{false}}{Q} \equiv Q \qquad 
         (\mbox{If-3})\ \sif{P}{b}{P} \equiv P\\
        &(\mbox{If-4})\ \sif{P}{b}{(\sif{Q}{b}{R})}\equiv \sif{(\sif{P}{b}{Q})}{b}{R}\\ 
        &(\mbox{If-5})\ \sif{P}{b}{Q} \equiv \sif{Q}{\neg b}{P}\\ 
        &(\mbox{If-6})\ \sif{P}{(\sif{c}{b}{d})}{Q} \equiv \sif{(\sif{P}{c}{Q})}{b}{(\sif{P}{d}{Q})}\\
        &(\mbox{If-7})\ \sif{P}{b}{(\sif{Q}{b}{R})} \equiv \sif{P}{b}{R}\\ 
        &(\mbox{If-8})\ \sif{(\sif{P}{b}{Q})}{c}{R} \equiv \sif{(\sif{P}{c}{R})}{b}{(\sif{Q}{c}{R})}\\
        &(\mbox{Sq-1})\ \iskip; P\equiv P; \iskip\equiv P\qquad\ \ \ 
         (\mbox{Sq-2})\ \iabort; P\equiv P; \iabort\equiv \iabort\\ 
        &(\mbox{Sq-3})\ P; (Q; R)\equiv (P; Q); R\qquad\ \ \ \ 
         (\mbox{Sq-4})\ (\sif{P}{b}{Q}); R\equiv \sif{(P; R)}{b}{(Q; R)} \\
        &(\mbox{Nd-1})\ P\sqcup Q\equiv Q\sqcup P\qquad\qquad\ \ \ \ 
         (\mbox{Nd-2})\ P\sqcup(Q\sqcup R)\equiv (P\sqcup Q)\sqcup R\\ 
        &(\mbox{Nd-3})\ \sif{(P\sqcup Q)}{b}{R} \equiv (\sif{P}{b}{R})\sqcup (\sif{Q}{b}{R})\\
        &(\mbox{Nd-4})\ \sif{R}{b}{(P\sqcup Q)} \equiv (\sif{R}{b}{P})\sqcup (\sif{R}{b}{Q})\\ 
        &(\mbox{Nd-5})\ (\sif{P}{b}{Q})\sqcup R \equiv \sif{(P\sqcup R)}{b}{(Q\sqcup R)}\\ 
        &(\mbox{Nd-6})\ (P\sqcup Q); R\equiv (P; R)\sqcup (Q; R)\qquad 
         (\mbox{Nd-7})\ R; (P\sqcup Q)\equiv (R; P)\sqcup (R; Q)\\
        &(\mbox{Re-1})\ \swhile{b}{P} \equiv \mu X.(\sif{(P; X)}{b}{\iskip})\\
        &(\mbox{Re-2})\ \mu X.F(X)\equiv F(\mu X.F(X))\qquad 
         (\mbox{Re-3})\ Y=F(Y)\Rightarrow \mu X.F(X)\subseteq Y
    \end{align*}
\caption{Laws of Classical Programming. 
    In law (If-6), $\sif{c}{b}{d}$ is a boolean expression that assumes value $c$ if $b$ is true and value $d$ if $b$ is false. 
    In law (Re-3), the inclusion relation $\subseteq$ represents the refinement of programs.}
\label{fig 1}
\end{figure}

\section{A Framework of Quantum Programming}\label{sec-framework}

In this section, we define a three-layer framework of quantum programming in which our laws can be properly accommodated. 
As pointed out in the Introduction, all ingredients of this framework are not new and have been widely used by the quantum programming community. 

We assume a set $\QVar$ of quantum variables and a set $\mathcal{U}$ of \textit{unitary matrix constants}. 
The unitary matrix constants in $\mathcal{U}$ will be instantiated in applications. 
For convenience, it is always assumed that the $2\times 2$ \textit{unit matrix} is in $\mathcal{U}$. We fix the following notations:
\begin{itemize}
    \item[-] Each quantum variable $q$ stands for a quantum system with a Hilbert space denoted by $\hs_q$ as its state space.
    \item[-] A sequence $\qbar=q_1,...,q_n$ of distinct quantum variables is called a quantum register. 
    It denotes a composite quantum system consisting of subsystems $q_1,...,q_n$. 
    Its state space, or type, denoted by $\type(\qbar)$, is given by the tensor product $\bigotimes_{i=1}^n\hs_{q_i}$.
    We write $\qbar_1,\qbar_2$ for the concatenation of two distinct registers $\qbar_1$ and $\qbar_2$.
    For simplicity of the presentation, we often identify $\qbar$ with the set $\{q_1,...,q_n\}$ of quantum variables occurring in $\qbar$. 
    \item[-] Each unitary matrix constant $U\in\mathcal{U}$ assumes a type of the form $\type(U)=\hs_1\otimes ...\otimes\hs_n$. 
    This means that the unitary transformation denoted by $U$ can be performed on a composite quantum system consisting of $n$ subsystems with state spaces $\hs_1,...,\hs_n$, respectively. 
    Thus, if $\qbar=q_1,...,q_n$ is a quantum register with $\hs_{q_i}=\hs_i$ for $i=1,...,n$ (that is, the types of $U$ matches that of the register $\qbar$), then $U[\qbar]$ can be thought of as a basic quantum gate with quantum wires $q_1,...,q_n$.  
\end{itemize}

\subsection{Quantum circuits}
The first layer of our framework (see Figure~\ref{fig 0}) is a description language of quantum circuits. 
Formally, we have: 
\begin{defn}
\label{def-circ}
    Quantum circuits $C\in\mathbf{QC}$ are defined by the syntax:
    \begin{align}
    \label{def-QC}
        C::= \iskip \mid U[\qbar] \mid C_1; C_2 \mid 
        \iqif{\qbar}{\square_{i=1}^d |\psi_i\>}{C_i}
    \end{align}
    More precisely, quantum circuits $C$ are inductively defined by the following clauses, where $\qvar(C)$ denotes the quantum variables in $C$:
    \begin{enumerate} 
        \item\label{circ-skip} \textbf{Skip}: $\qvar(\iskip) = \emptyset$;
        \item\label{circ-gate} \textbf{Basic gates}: If $U\in\mathcal{U}$ is a unitary matrix constant and $\qbar$ is a quantum register such that their types match, then quantum gate $U[\qbar]$ is a circuit, and $\qvar(U[\qbar])=\qbar$;  
        \item\label{circ-seq} \textbf{Sequential composition}: If $C_1$ and $C_2$ are circuits, then $C\equiv C_1; C_2$ is a circuit too, and $\qvar(C)=\qvar(C_1)\cup\qvar(C_2)$;
        \item\label{circ-qif} \textbf{Quantum $\kif$-statement}: If $\qbar$ is a quantum register, $\left\{|\psi_i\>\right\}_{i=1}^d$ is an orthonormal basis of the Hilbert space $\type(\qbar)$, and $C_i$ $(i=1,...,d)$ are circuits with
        \begin{equation}
        \label{coin-cond}
            \qbar\cap\left(\bigcup_{i=1}^d\qvar(C_i)\right)=\emptyset,
        \end{equation}
        then
        \begin{equation}
        \label{q-mux}
            C\equiv \iqif{\qbar}{\square_{i=1}^d |\psi_i\>}{C_i}
        \end{equation}
        is a circuit, and $\qvar(C)=\qbar\cup\left(\bigcup_{i=1}^d\qvar(C_i)\right)$. $\qbar$ is called the guard register, and $\{|\psi_i\>\}$ is called the guard basis of the quantum $\kif$-statement (\ref{q-mux}). 
    \end{enumerate}
\end{defn}

Intuitively, each $C\in\mathbf{QC}$ represents a circuit with quantum wires $\qvar(C)$. 
In other words, the type of $C$ is $\type(C)=\bigotimes_{q\in\qvar(C)}\hs_q$, and as will be seen in Definition~\ref{def-c-sem}, the semantics of $C$ is a unitary operator on $\type(C)$. 
The circuit constructs introduced in the above definition are explained as follows. 
\begin{enumerate}
    \item[(i)] The circuit $C_1; C_2$ in Clause (\ref{circ-seq}) stands for the sequential composition of circuits $C_1$ and $C_2$. 
    Indeed, if $C_1$ and $C_2$ do not share quantum variables; that is, $\qvar(C_1)\cap \qvar(C_2)=\emptyset$, we can also define their parallel composition $C_1\otimes C_2$. But $C_1; C_2$ and $C_1\otimes C_2$ are semantically equivalent whenever $\qvar(C_1)\cap \qvar(C_2)=\emptyset$. 
    So, the parallel composition is not included in the above definition. 
    \item [(ii)] The quantum $\kif$-statement defined in Eq. (\ref{q-mux}) has been extensively discussed in the previous literature~\cite{Alt05, Ying12, Sab18, Yuan, Bich, Voi23}. 
    The condition (\ref{coin-cond}) means that $\qbar$ is a system external to all $C_i$ $(i=1,...,d)$. Semantically, a quantum $\kif$-statement (\ref{q-mux}) is a quantum multiplexor (i.e., a multi-way generalisation of conditional)~\cite{Markov}. 
    A multiplexer can be understood as a switch that passes one of its data inputs through to the output, as a function of a set of select inputs. 
    Here, $\qbar$ is the select register, and if $\qbar$ is in state $|\psi_i\>$, then a quantum datum $|\varphi\>$ is inputted to the corresponding circuit $C_i$ and an output $|\varphi_i\>$ is obtained at the end of $C_i$. 
    A basic difference between classical and quantum multiplexors is that the quantum select register $\qbar$ can be in a superposition of $|\psi_i\>$ $(i=1,...,d)$, say $\sum_{i=1}^d\alpha_i|\psi_i\>$. 
    In this case, the output is then $\sum_{i=1}^d\alpha_i|\varphi_i\>$, a superposition of the outputs of different circuits $C_i$ $(i=1,...,d)$. Therefore, the control flow of a quantum $\kif$-statement is inherently quantum. 
\end{enumerate}

For simplicity of presentation, let us introduce a syntactic sugar for quantum $\kif$-statement when some of the branches are $\iskip$: 
\[
    \iqif{\qbar}{\square_{i\in I} |\psi_i\>}{U_i}
    \triangleq
    \iqif{\qbar}{\left(\square_{i\in I} |\psi_i\>\rightarrow U_i\right)\ \square\ \left(\square_{i\notin I} |\psi_i\>\rightarrow \iskip\right)}{}
\]

The following example illustrates how some frequently used (multi-qubit) quantum gates can be defined in terms of quantum $\kif$-statements.  
\begin{exam}
    Assume that the single-qubit \textsf{NOT} matrix $\gateX$, and $\rbar_x(\theta)=iR_x(2\theta)$ are unitary matrix constants in $\mathcal{U}$, where $R_x(\theta)$ is the rotation about the $X$ axis:  
    \[
        R_x(\theta)=\left(
        \begin{array}{cc}
            \cos\frac{\theta}{2} & -i\sin\frac{\theta}{2}\\ 
            -i\sin\frac{\theta}{2} &\cos\frac{\theta}{2}
        \end{array}
        \right).
    \]
    Then:
    \begin{enumerate}
        \item The \(\gateCNOT\) (controlled-\textsf{NOT}) gate with $q_1$ as its control qubit can be defined by
        \begin{align*}
            \gateCNOT[q_1,q_2] \triangleq \iqif{q_1}{|1\>}{\gateX[q_2]}.
        \end{align*} 
        \item The \(\gateToffoli\) gate with $q_1,q_2$ as its control qubits is defined by 
        \begin{align*}
            \gateToffoli[q_1,q_2,q_3] \triangleq \iqif{q_1,q_2}{|11\>}{\gateX[q_3]}.
        \end{align*} 
        \item The \(\gateDeutsch\) gate with $q_1,q_2$ as its control qubits is defined by 
        \begin{align*}
            \gateDeutsch(\theta)[q_1,q_2,q_3] \triangleq \iqif{q_1,q_2}{|11\>}{\rbar_x(\theta)[q_3]}.
        \end{align*}
        Note that $\gateDeutsch(\frac{\pi}{2})=\gateToffoli$. 
        \item The \(\gateFredkin\) gate with $q_1$ as its control qubit is defined by 
        \begin{align*}
            \gateFredkin[q_1,q_2,q_3] \triangleq \iqif{q_1}{|1\>}{\gateSWAP[q_2,q_3]}
        \end{align*}
        where the $\gateSWAP$ gate is defined by
        $$\gateSWAP[q_2,q_3] \triangleq \gateCNOT[q_2,q_3]; \gateCNOT[q_3,q_2]; \gateCNOT[q_2,q_3].$$
    \end{enumerate}
\end{exam}

\subsection{Quantum $\kwhile$-programs}\label{sec-qprog}

Upon the layer of quantum circuits described in the previous subsection, the second layer consists of purely quantum programs, namely, quantum programs without classical variables. 
They are written in the quantum $\kwhile$-language as follows.
\begin{defn}
\label{def-prog}
    Quantum programs $P\in\mathbf{QProg}$ are defined by the syntax\footnote{Note that we use the same notations $\iskip$ and $; $ (sequential composition) at both the circuit and program layers. This does not cause ambiguity in semantics.}:
    \begin{equation}
    \label{prog-syntax}
        \begin{split}
            P::=&\ \iskip \mid \iabort \mid \iinit{\qbar}{|\psi\>} \mid C \mid P_1; P_2\\ 
            & \mid \iif{M}{\qbar}{m}{P_m} \mid \iwhile{M}{\qbar}{P}.
        \end{split}
    \end{equation}
\end{defn}

Here, $\iinit{\qbar}{|\psi\>}$ means that the quantum register (i.e., a sequence of quantum variables) $\qbar$ is initialised in state $|\psi\>$. 
$C$ stands for an arbitrary quantum circuit defined in Eq. (\ref{def-QC}). $\iskip$, $\iabort$ and sequential composition $P_1; P_2$ are the same as in classical programming. 
Quantum measurement $M$, represented by a set of linear operators $M = \{M_i\}_{i\in\mathcal{J}}$ satisfying the completeness equation $\sum_{i\in\mathcal{J}}M_i^\dag M_i = I$, is employed as the guard of branch statements.
The $\kif$-statement means that a measurement $M$ is performed on quantum register $\qbar$, and if the outcome is $m$, then the corresponding subprogram $P_m$ is executed next. 
It is fundamentally different from the quantum $\kif$-statement defined in the above subsection. 
The control flow of this $\kif$-statement is classical and determined by classical information; that is, the outcomes of measurement $M$. 
In sharp contrast, the control flow of a quantum $\kif$-statement is quantum, as discussed in the above subsection. 
We often abbreviate $\iif{M}{\qbar}{m}{P_m}$ to $\iif{}{\qbar}{m}{P_m}$ whenever $M$ is the measurement in the computational basis of qubits.  
The $\kwhile$-loop performs a measurement (with only two possible outcomes $0$ and $1$) on $\qbar$; if the outcome is $1$ then the loop body $P$ is executed, and if the outcome is $0$ then the loop terminates. 
So, its control flow is also classical.   

This simple quantum programming language will be expanded by adding recursion and nondeterministic choice in Sections~\ref{sec-recur} and~\ref{sec-nondet}, respectively. 

\subsection{Semantics}

In this subsection, we briefly describe the denotational semantics of quantum circuits $\mathbf{QC}$ and purely quantum programs $\mathbf{QProg}$ (for more details, we refer to Chapter 3 of~\cite{Ying16}). 
The correctness of our laws of quantum programming will be proved based on them. 
In Section~\ref{sec-coq}, these semantics will be formalised in the Coq proof assistant so that the laws of quantum programming developed in this paper can be mechanically proved in Coq. 
\begin{defn}
\label{def-c-sem}
    The denotational semantics $\sem{C}$ of a quantum circuit $C\in\mathbf{QC}$ is a unitary operator on Hilbert space $\hs_{\qvar(C)}$ inductively defined as follows: 
    \begin{enumerate}
        \item If $C = \iskip$, then $\sem{C}=I$ (the identity operator);
        \item If $C=U[\qbar]$, then $\sem{C}=U_{\qbar}$;
        \item If $C=C_1; C_2$, then $\sem{C}=\left(\sem{C_2}\otimes I_{\qvar(C_1)\setminus\qvar(C_2)}\right)\left(\sem{C_1}\otimes I_{\qvar(C_2)\setminus\qvar(C_1)}\right)$;
        \item If $C=\iqif{\qbar}{\square_{i=1}^d |\psi_i\>}{C_i}$, then $\sem{C}=\sum_{i=1}^d|\psi_i\>_{\qbar}\<\psi_i|\otimes\left (\sem{C_i}\otimes I_{\qvar(C)\setminus\qbar\setminus\qvar(C_i)}\right)$.
    \end{enumerate}
\end{defn}

Let $\mathcal{D}(\hs_\QVar)$ denote the set of partial density operators (i.e., positive operators with traces $\leq 1$) on $\hs_\QVar$. 
Then the denotational semantics of a quantum program $P\in\mathbf{QProg}$ was defined in~\cite{Ying16} as a mapping $\sem{P}:\mathcal{D}(\hs_\QVar)\rightarrow \mathcal{D}(\hs_\QVar).$ 
The L\"{o}wner order $\sqsubseteq$ between operators on $\hs_\QVar$, defined as $A\sqsubseteq B$ if $B-A$ is positive semi-definite, can be lifted pointwise to an order between quantum operations (i.e., super-operators), also denoted $\sqsubseteq$. 
We write $\mathcal{QO}(\hs_\QVar)$ for the set of quantum operations on $\hs_\QVar$. 
It can be proved that $\left(\mathcal{QO}(\hs_\QVar),\sqsubseteq\right)$ is a CPO (Complete Partial Order)~\cite{Ying16}. 
Furthermore, the denotational semantics of quantum programs enjoys the following:  

\begin{prop}[Structural Representation]
\label{sempro} 
    For any input $\rho$, we have: 
    \begin{enumerate}  
        \item\label{sempro-skip} $\sem{\iskip} (\rho)=\rho$; $\sem{\iabort} (\rho)=0$. 
        \item\label{sempro-init} $\sem{ \iinit{q}{|\psi\>}} (\rho)=\sum_n |\psi\>_q\<n|\rho|n\>_q\< \psi|$, where $\{|n\>\}$ is a given orthonormal basis of $\hs_q$. 
        \item\label{sempro-unitary} $\sem{C} (\rho)=\sem{C}\rho \sem{C}^{\dag}$\footnote{We slightly abuse notation here: $\sem{C}$ on the LHS denotes the program semantics, which is a quantum operation, whereas $\sem{C}$ on the RHS denotes the circuit semantics, which is a unitary operator.}.
        \item\label{sempro-seq} $\sem{P_1; P_2} (\rho)=\sem{P_2} (\sem{P_1} (\rho))$.
        \item\label{sempro-if}  $\sem{\iif{M}{\qbar}{m}{P_m}} (\rho)=\sum_m\sem{P_m} \left((M_m)_\qbar \rho (M_m)_\qbar^{\dag}\right)$.
        \item\label{sempro-while}  Let $\kwhile$ stand for quantum loop $\iwhile{M}{\qbar}{P}$. 
        Then
        \begin{equation}
        \label{loop-sem-function-1}
            \sem{\kwhile}(\rho)=\sum_{n=0}^\infty\left[\E_0\circ (\sem{P}\circ \E_1)^n\right]=\bigsqcup_{k=0}^{\infty}\sem{\kwhile^{(k)}}(\rho),
        \end{equation}
        where for every $k\geq 0$, $\kwhile^{(k)}$ is the first $k$ iterations of the loop; that is, $\sem{\kwhile^{(k)}}=\sum_{n=0}^{k-1}\left[\E_0\circ (\sem{P}\circ \E_1)^n\right]$ with $\E_i(\rho)=(M_i)_\qbar\rho (M_i)_\qbar^\dag$ for $i=0,1$, and the symbol $\bigsqcup$ stands for the supremum of quantum operations; i.e., the least upper bound in the CPO $\left(\mathcal{QO}(\hs_\QVar),\sqsubseteq\right)$.
    \end{enumerate}
\end{prop}

\subsection{Further Expansion}

As outlined in the Introduction, at the third layer of our framework, purely quantum programming language $\mathbf{QProg}$ defined in Subsection~\ref{sec-qprog} is further embedded into a classical programming language (e.g., Dijkstra's nondeterministic sequential programming language~\cite{Dij76} employed in Section~\ref{sec-classical-laws} for formulating the laws of classical programming), and a program construct of the form $x:=M[\qbar]$ is used to link classical and quantum computing, where quantum measurement $M$ is performed on quantum register $\qbar$ and the outcome is stored in classical variable $x$. 
Thus, the laws of classical programming, like those presented in Figure~\ref{fig 1}, can be combined with our laws to be established in the sequent sections for reasoning about classical-quantum hybrid programs at this layer. Due to the limitation of space, we are not going to elaborate on such a combination in this paper.    

\section{Algebraic Laws for Quantum Circuits}\label{sec-laws-circ}

Now we are ready to introduce our laws of quantum programming. 
As the first step, in this section, we prove a set of algebraic laws for quantum circuits in $\mathbf{QC}$, constructed from basic gates given by $\mathcal{U}$ using sequential composition and quantum $\kif$-statement. 
To describe these laws, we need the notion of equivalence between quantum circuits:

\begin{defn}
    Let $C_1,C_2\in\mathbf{QC}$ be two quantum circuits. 
    They are called equivalent, written $C_1\equiv C_2$, if $\sem{C_1}\otimes I_{\qvar(C_2)\setminus\qvar(C_1)}=\sem{C_2}\otimes I_{\qvar(C_1)\setminus\qvar(C_2)}.$
\end{defn}

A general form of quantum $\kif$-statement defined by a quantum register is allowed in Definition~\ref{def-circ}. 
But for simplicity of presentation, we often only consider quantum $\kif$-statements defined by a single qubit. 
Following an idea in~\cite{Hoare87}, we choose to use an infix notation for them. 
Let $|\varphi_0\>, |\varphi_1\>$ be a basis of the Hilbert space of qubit $q$. 
Then
\begin{equation}
\label{q-control}
    \sqif{C_0}{q}{C_1}{|\varphi_0\>}{|\varphi_1\>} \triangleq 
    \iqif{q}{|\varphi_0\>\rightarrow C_0\ \square\ |\varphi_1\>\rightarrow C_1}{}.
\end{equation}
In particular, whenever the basis $\{|\varphi_0\>, |\varphi_1\>\}$ is fixed (e.g., the computational basis $\{|0\>,|1\>\}$) or known from the context, then $\sqif{C_0}{q}{C_1}{|\varphi_0\>}{|\varphi_1\>}$ is abbreviated to $\sqif{C_0}{q}{C_1}{}{}$. 
Furthermore, $\sqif{\iskip}{q}{C_1}{}{}$ is abbreviated to $\sqif{}{q}{C_1}{}{}$. 
We also write $\sqif{}{\qbar}{C}{}{|\varphi\>}$ for $\iqif{\qbar}{|\varphi\>}{C}$, where only the $|\varphi\>$ branch is not $\iskip$.
These notations will significantly simplify the presentation of our laws for quantum circuits.

\subsection{Basic Laws}

Let us first present some basic laws for quantum circuits. 
For readability, the proofs of them are postponed to the Appendix.
Our first proposition shows a group of laws for quantum $\kif$-statements. 

\begin{prop}[Laws of quantum $\kif$-statement]
\label{qif-law}
    \quad
    \begin{enumerate}
        \item\label{qif-law-changing-basis} (Changing Basis) If unitary $U$ transforms basis $|\varphi_0\>, |\varphi_1\>$ to basis $|\psi_0\>, |\psi_1\>$ up to a phase; that is, $U|\varphi_i\>= \lambda_i|\psi_i\>$ $(i=0,1)$, then 
        \[
            \sqif{C_0}{q}{C_1}{|\psi_0\>}{|\psi_1\>} \equiv
            U^\dag[q]; \sqif{C_0}{q}{C_1}{|\varphi_0\>}{|\varphi_1\>}; U[q].
        \]
        \item\label{qif-law-symmetry}  (Symmetry) $\sqif{C_0}{q}{C_1}{|\varphi_0\>}{|\varphi_1\>} \equiv \sqif{C_1}{q}{C_0}{|\varphi_1\>}{|\varphi_0\>}$.
        \item\label{qif-law-idempotence} (Idempotence) $\sqif{C}{q}{C}{}{} \equiv C$. 
        \item\label{qif-law-distributivity} (Distributivity) If $q_2\notin\qvar(C)$, then 
        \[
            \sqif{C}{q_1}{(\sqif{C_0}{q_2}{C_1}{}{})}{}{} \equiv 
            \sqif{(\sqif{C}{q_1}{C_0}{}{})}{q_2}{(\sqif{C}{q_1}{C_1}{}{})}{}{}.
        \]
        \item\label{qif-law-nested} (Nested-\textbf{qif})
        $\sqif{(\sqif{C_{00}}{q_2}{C_{01}}{}{})}{q_1}{(\sqif{C_{10}}{q_2}{C_{11}}{}{})}{}{}
        \equiv \iqif{q_1, q_2}{|i,j\>}{C_{ij}}$.
    \end{enumerate}
\end{prop}

Clause (\ref{qif-law-changing-basis}) in the above proposition is quantum-specific and has no classical counterpart. 
Clauses (\ref{qif-law-symmetry}) to (\ref{qif-law-nested}) are generalisations of classical laws (If-3), (If-5), (If-6) and (If-8) in Figure~\ref{fig 1}.  

The next proposition gives several laws for the sequential composition of quantum circuits. 

\begin{prop}[Laws of sequential composition]
\label{circuit-seq-law}
    \quad
    \begin{enumerate}
        \item\label{circuit-seq-law-unit} (Unit) $\iskip; C\equiv C; \iskip\equiv C$.
        \item\label{circuit-seq-law-identity} (Identity) $I[\qbar] \equiv \iskip$, where $I$ is the identity operator.
        \item\label{circuit-seq-law-composition} (Sequential Composition) $U_1[\qbar]; U_2[\qbar]\equiv (U_2U_1)[\qbar]$ if $U_2U_1\in\mathcal{U}$.
        \item\label{circuit-seq-law-disjoint-composition} (Parallel  Composition) $U_1[\qbar]; U_2[\rbar]\equiv U_2[\rbar]; U_1[\qbar]\equiv (U_1\otimes U_2)[\qbar,\rbar]$ if $\qbar\cap\rbar=\emptyset$ and $U_1\otimes U_2\in\mathcal{U}$. 
        \item\label{circuit-seq-law-commutativity} (Commutativity) If $\qvar(C_1)\cap\qvar(C_2)=\emptyset$, then $C_1; C_2\equiv C_2; C_1.$ 
        \item\label{circuit-seq-law-associativity} (Associativity) $(C_1; C_2); C_3\equiv C_1; (C_2; C_3).$
        \item\label{circuit-seq-law-sequentiality} (Sequentiality) 
        $(\sqif{C_0}{q}{C_1}{}{}); (\sqif{D_0}{q}{D_1}{}{}) \equiv \sqif{(C_0; D_0)}{q}{(C_1; D_1)}{}{}.$
        \item\label{circuit-seq-law-distributivity} (Distributivity over \textbf{qif}) If $q\notin\qvar(C)$, then
        $$\qquad C; (\sqif{C_0}{q}{C_1}{}{})= \sqif{(C; C_0)}{q}{(C; C_1)}{}{},\quad 
        (\sqif{C_0}{q}{C_1}{}{}); C = \sqif{(C_0; C)}{q}{(C_1; C)}{}{}.
        $$
    \end{enumerate}
\end{prop}

It is obvious that Clauses (\ref{circuit-seq-law-unit}), (\ref{circuit-seq-law-associativity}) and (\ref{circuit-seq-law-distributivity}) in the above proposition are generalisations of classical laws (Sq-1), (Sq-3) and (Sq-4) in Figure~\ref{fig 1}, respectively. 
As a corollary of the law (Sequentiality), we have: 
\begin{align}
    &\textit{(Merging)}\qquad\quad 
    \sqif{}{\qbar}{(C_1; C_2)}{}{|\varphi\>} \equiv 
    (\sqif{}{\qbar}{C_1}{}{|\varphi\>}); (\sqif{}{\qbar}{C_2}{}{|\varphi\>})\\ 
    \label{qif-law-splitting}
    &\textit{(Splitting)}\qquad\quad 
    \iqif{\qbar}{\square_{i=1}^d|\varphi_i\>}{C_i} \equiv
    (\sqif{}{\qbar}{C_1}{}{|\varphi_1\>}); \cdots; (\sqif{}{\qbar}{C_d}{}{|\varphi_d\>}).
\end{align}

A particularly useful syntactic sugar is quantum choice defined in~\cite{Ying12} as follows: for any quantum circuit $C\in\mathbf{QC}$ with $\qvar(C)=\{q\}$, 
\begin{enumerate}
    \item (Pre-choice) $\sqif{C_0}{C[q]}{C_1}{}{} \triangleq C; (\sqif{C_0}{q}{C_1}{}{}).$
    \item (Post-choice) $\sqif{C_0}{[q]C}{C_1}{}{} \triangleq (\sqif{C_0}{q}{C_1}{}{}); C.$
\end{enumerate}
We often call $C$ the coin tossing circuit of the choice. 
Intuitively, in the pre-choice, $C$ tosses quantum choice $q$ so that $q$ is in a superposition, which then determines the control flow of quantum $\kif$-statement $\sqif{C_0}{q}{C_1}{}{}$. 
In the post-choice, the coin tossing operation $C$ mixes the execution paths represented by the quantum $\kif$-statement. 
Quantum choice can be seen as a quantum counterpart of probabilistic choice in the probabilistic programming language pGCL~\cite{He97}. 
The following example of a quantum random walk shows an application of quantum choice. 
Indeed, the notion of quantum choice introduced in~\cite{Ying12} was inspired by the construction of quantum random walks. 

\begin{exam}[Quantum random walk]
\label{exam-walk}
    The one-dimensional random walk can be thought of as a particle moving on a discrete line whose nodes are denoted by integers $\mathbb{Z}=\{...,-2,-1,0,1,2,...\}$. 
    At each step, the particle moves one position left or right, depending on the flip of a coin. 
    A quantum variant of it is the Hadamard walk, with a quantum coin operator given by the Hadamard gate $\gateH=\frac{1}{\sqrt{2}}\left(\begin{array}{cc} 1 & 1\\ 1 &-1\end{array}\right).$ 
    Formally, the Hilbert space of the Hadamard walk is $\hs_d\otimes\hs_p$, where $\hs_d=\mathit{span} \{|L\>,|R\>\}$ is a $2$-dimensional Hilbert space, called the direction space, $|L\>, |R\>$ are used to indicate the direction Left and Right, respectively, $\hs_p=\overline{\mathit{span} \{|n\>: n\in\mathbb{Z}\}}$ is an infinite-dimensional Hilbert space, $|n\>$ indicates the position marked by integer $n$, and $\overline{\mathit{span}\ X}$ is the Hilbert space spanned by a nonempty set $X$ of vectors. 
    One step of the Hadamard walk is modelled by the unitary operator $W=T(\gateH\otimes I_p),$ where the translation $T$ is a unitary operator on $\hs_d\otimes\hs_p$ defined by $T|L, n\>=|L, n-1\>$ and $T|R, n\>=|R, n+1\>$ for every $n\in\mathbb{Z}$, $\gateH$ is the Hadamard gate on $\hs_d$, and $I_p$ is the identity operator on $\hs_p$. 
    The Hadamard walk is then described by repeated applications of the operator $W$.

    Now we define the left and right translation operators $T_L$ and $T_R$ on the position space $\hs_p$ by $T_L|n\>=|n-1\>$ and $T_R|n\>=|n+1\>$ for each $n\in\mathbb{Z}$. 
    The translation can be written as quantum $\kif$-statement $T = \sqif{T_L}{d}{T_R}{}{}$, the walk operator is a quantum choice $W = \sqif{T_L[p]}{\gateH[d]}{T_R[p]}{}{}$, and an $n$-step quantum walk is the sequential composition of $n$ copies of $W$.
\end{exam}

Several laws for quantum choices are presented in the following: 

\begin{prop}[Laws of Quantum Choice]
\label{choice-laws}
    \quad
    \begin{enumerate}
        \item\label{choice-laws-symmetry} (Symmetry) Let $\gateX$ be the \textsf{NOT} gate. 
        Then 
        \[
            \sqif{C_0}{\gateX[q]}{C_1}{}{} \equiv \sqif{C_1}{[q]\gateX}{C_0}{}{}.
        \]
        More generally, if $U|\varphi_i\>=\lambda_i|\psi_i\>$ for $(i=0,1)$, then 
        \[
            \sqif{C_0}{U[q]}{C_1}{|\psi_0\>}{|\psi_1\>} \equiv 
            \sqif{C_0}{[q]U}{C_1}{|\varphi_0\>}{|\varphi_1\>}.
        \]
        \item\label{choice-laws-sequentiality} (Sequentiality)
        $\ (\sqif{C_0}{B[q]}{C_1}{}{}); (\sqif{D_0}{q}{D_1}{}{})\equiv \sqif{(C_0; D_0)}{B[q]}{(C_1; D_1)}{}{}$,
        \[
            \hspace{1.83cm}(\sqif{C_0}{q}{C_1}{}{}); (\sqif{D_0}{[q]B}{D_1}{}{}) \equiv 
            \sqif{(C_0; D_0)}{[q]B}{(C_1; D_1)}{}{}.
        \]
        \item\label{choice-laws-distributivity} (Distributivity) If $q\notin\qvar(C)$, then 
        \begin{align*}
            &C; (\sqif{C_0}{B[q]}{C_1}{}{}) \equiv \sqif{(C; C_0)}{B[q]}{(C; C_1)}{}{}, \\
            &(\sqif{C_0}{B[q]}{C_1}{}{}); C\equiv \sqif{(C_0; C)}{B[q]}{(C_1; C)}{}{}.
        \end{align*}
        The same holds for post-choice. 
    \end{enumerate}
\end{prop}

It is easy to see the correspondence between the laws (Symmetry), (Sequentiality) and (Distributivity) in the above proposition and the ones in Propositions~\ref{qif-law} and~\ref{circuit-seq-law}. 
On the other hand, the laws in the above proposition can be regarded as quantum counterparts of the corresponding laws of probabilistic programming given in~\cite{He97}. 

To show that the algebraic laws proposed above are actually useful for circuit optimisation, let us see the following: 

\begin{exam}[Circuit Optimisation]
\label{exam:circuit opt}
    We take an example from an earlier speech by Ross Duncan\footnote{\url{https://www.youtube.com/watch?v=QbtSSqeYuFM\#t=01h35m30s}}:
    \[
        \Qcircuit @C=1em @R=1em {
        p & & \ctrl{2} & \gate{Z} & \ctrl{1} &
        \gate{Z} & \qw & & & \ctrl{1} & \qw & \qw & \qw \\
        q & & \qw & \ctrl{1} & \targ &
        \gate{X} & \qw & = & & \targ & \gate{X} & \ctrl{1} & \qw \\
        r & & \targ & \targ & \qw & \gate{X} &
        \qw & & & \qw & \qw & \targ & \qw
        }.
    \]
    It is used to motivate using ZX-calculus for circuit optimisation in~\cite{staudacher2021optimization}. 
    The equivalence is not obvious at first glance.
    Let us write the two circuits in our language as:
    \[
        \sqif{}{p}{\gateX[r]}{}{}; \gateZ[p]; \sqif{}{q}{\gateX[r]}{}{}; \sqif{}{p}{\gateX[q]}{}{}; \gateZ[p]; \gateX[q]; \gateX[r]
        \equiv \sqif{}{p}{\gateX[q]}{}{}; \gateX[q]; \sqif{}{q}{\gateX[r]}{}{}.
    \]

    We first check that:
    \begin{align*}
        & \sqif{}{p}{\gateX[r]}{}{}; \gateZ[p]; \sqif{}{q}{\gateX[r]}{}{}; \sqif{}{p}{\gateX[q]}{}{}; \gateZ[p]; \gateX[q]; \gateX[r] & \\
        \equiv\ & \sqif{}{p}{\gateX[r]}{}{}; \sqif{}{q}{\gateX[r]}{}{}; (\gateZ^\dag[p]; \sqif{}{p}{\gateX[q]}{}{}; \gateZ[p]); \gateX[q]; \gateX[r] & \hspace{-1cm}
        \reason{Prop.~\ref{circuit-seq-law}(\ref{circuit-seq-law-commutativity},~\ref{circuit-seq-law-associativity}); \(\gateZ^\dag = \gateZ\)}\\
        \equiv\ & \sqif{}{p}{\gateX[r]}{}{}; \sqif{}{q}{\gateX[r]}{}{}; \sqif{}{p}{\gateX[q]}{}{}; \gateX[q]; \gateX[r] &
        \reason{Prop.~\ref{qif-law}(\ref{qif-law-changing-basis})}\\
        \equiv\ & \sqif{(\sqif{}{q}{\gateX[r]}{}{}; \gateX[q]; \gateX[r])}{p}{(\gateX[r]; \sqif{}{q}{\gateX[r]}{}{}; \gateX[q]; \gateX[q]; \gateX[r])}{}{}. &
        \reason{Props.~\ref{qif-law}(\ref{qif-law-idempotence}),~\ref{circuit-seq-law}(\ref{circuit-seq-law-sequentiality})}
    \end{align*}
    Next, we simplify the sub-circuits as follows:
    \begin{align*}
        & \sqif{}{q}{\gateX[r]}{}{}; \gateX[q]; \gateX[r] & \\
        \equiv\ & \gateX[q]; \sqif{\gateX[r]}{q}{}{}{}; \gateX[r] &
        \reason{Prop.~\ref{choice-laws}(\ref{choice-laws-symmetry})}\\
        \equiv\ & \gateX[q]; \sqif{(\gateX[r]; \gateX[r])}{q}{\gateX[r]}{}{} &
        \reason{Props.~\ref{qif-law}(\ref{qif-law-idempotence}),~\ref{circuit-seq-law}(\ref{circuit-seq-law-sequentiality})} \\
        \equiv\ & \gateX[q]; \sqif{}{q}{\gateX[r]}{}{}, &
        \reason{Prop.~\ref{circuit-seq-law}(\ref{circuit-seq-law-composition},~\ref{circuit-seq-law-identity}); \(\gateX\gateX = I\)}
    \end{align*}
    and
    \begin{align*}
        &\gateX[r]; \sqif{}{q}{\gateX[r]}{}{}; \gateX[q]; \gateX[q]; \gateX[r] & \\
        \equiv\ & \gateX[r]; \sqif{}{q}{\gateX[r]}{}{}; \gateX[r] &
        \reason{Prop.~\ref{circuit-seq-law}(\ref{circuit-seq-law-composition},~\ref{circuit-seq-law-identity}); \(\gateX\gateX = I\)} \\
        \equiv\ & \sqif{(\gateX[r]; \gateX[r])}{q}{(\gateX[r]; \gateX[r]; \gateX[r])}{}{} &
        \reason{Props.~\ref{qif-law}(\ref{qif-law-idempotence}),~\ref{circuit-seq-law}(\ref{circuit-seq-law-sequentiality})} \\
        \equiv\ & \sqif{}{q}{\gateX[r]}{}{} &
        \reason{Prop.~\ref{circuit-seq-law}(\ref{circuit-seq-law-composition},~\ref{circuit-seq-law-identity}); \(\gateX\gateX = I\)} \\
        \equiv\ & \gateX[q]; \gateX[q]; \sqif{}{q}{\gateX[r]}{}{}. &
        \reason{Prop.~\ref{circuit-seq-law}(\ref{circuit-seq-law-composition},~\ref{circuit-seq-law-identity}); \(\gateX\gateX = I\)}
    \end{align*}
    Plug them into the previous formula:
    \begin{align*}
        & \sqif{}{p}{\gateX[r]}{}{}; \gateZ[p]; \sqif{}{q}{\gateX[r]}{}{}; \sqif{}{p}{\gateX[q]}{}{}; \gateZ[p]; \gateX[q]; \gateX[r] & \\
        \equiv\ & \sqif{(\gateX[q]; \sqif{}{q}{\gateX[r]}{}{})}{p}{(\gateX[q]; \gateX[q]; \sqif{}{q}{\gateX[r]}{}{})}{}{} & \\
        \equiv\ & \sqif{}{p}{\gateX[q]}{}{}; \sqif{(\gateX[q]; \sqif{}{q}{\gateX[r]}{}{})}{p}{(\gateX[q]; \sqif{}{q}{\gateX[r]}{}{})}{}{} & 
        \reason{Prop.~\ref{circuit-seq-law}(\ref{circuit-seq-law-unit},\ref{circuit-seq-law-sequentiality})} \\
        \equiv\ & \sqif{}{p}{\gateX[q]}{}{}; \gateX[q]; \sqif{}{q}{\gateX[r]}{}{} & 
        \reason{Prop.~\ref{qif-law}(\ref{qif-law-idempotence})}
    \end{align*}
    which completes the proof.
\end{exam}

\subsection{Normal Forms}

In this subsection, we establish a normal form for a large class of quantum circuits using the laws presented in the previous subsection. 
For each quantum variable $q\in\QVar$, we choose a fixed orthonormal basis of $\hs_q$ as the standard basis of $q$. 
The standard basis of a quantum register $\qbar$ is then defined as the set of tensor products of the standard basis states of variables in $\qbar$; for example, the computational basis $\left\{|i\>: i\in\{0,1\}^n\right\}$ of $n$ qubits. 

\begin{defn}
    \begin{enumerate}
        \item A flat quantum $\kif$-statement is a quantum circuit without nested quantum $\kif$-statements, i.e., of the form $\iqif{\qbar}{\square_i|i\>}{G_i}$, where $\{|i\>\}$ is the standard basis of $\qbar$, and each $G_i$ is a sequential composition of basic quantum gates or $\iskip$. 
        \item A normal form of the quantum circuit is a sequential composition of the form $C_1; \ldots; C_n$, where $C_i$ is a flat quantum $\kif$-statement for each $1\leq i\leq n$. 
    \end{enumerate}
\end{defn}

\begin{defn}
    A quantum circuit $C\in\mathbf{QC}$ is called regular if for every guard basis $\{|\psi_i\>\}$ of a quantum $\kif$-statement in $C$, the unitary transformation from $\{|\psi_i\>\}$ to the standard basis of the guard register $\qbar$ and its inverse, i.e., unitary transformations $\sum_i|\psi_i\>\<i|$ and $\sum_i|i\>\<\psi_i|$, can be implemented by a sequence of basic gates or $\iskip$.
\end{defn}

The following normal form theorem eliminates nested quantum $\kif$-statements in a quantum circuit and reduces them to flat (but large) quantum $\kif$-statements.
We use the term ``normal form'' in the sense of~\cite{Hoare87,Hoare93}, without requiring uniqueness---similar to conjunctive or disjunctive normal forms in propositional logic. This differs from the notion of normal forms in rewriting theory.

\begin{thm}[Normal form]
\label{normal-thm-cir} 
    Each regular quantum circuit $C\in\mathbf{QC}$ is equivalent to a normal form. 
\end{thm}

\begin{proof} 
We proceed by induction on the structure of $C$.

\textbf{Case 1}. 
$C=U[\qbar]$ is a basic gate. 
Then we introduce an external quantum variable $r\notin\qbar$. 
It follows from a simple generalisation of the idempotence law given in Proposition~\ref{qif-law}(\ref{qif-law-idempotence}) that $C\equiv \iqif{r}{\square|i\>}{G_i}$ where $\{|i\>\}$ is the standard basis of $r$, and $G_i=U[\qbar]$ for all $i$. 

\textbf{Case 2}. 
$C=C_1; C_2$, where both $C_1\equiv C_{11}; ...; C_{im}$ and $C_2\equiv C_{21}; ...; C_{2n}$ are in a normal form; that is, all $C_{1i}$ $(1\leq i\leq m)$ and $C_{2j}$ $(1\leq j\leq n)$ are flat quantum $\kif$-statements. 
Then it is obvious that $C\equiv C_{11}; ...; C_{1m}; C_{21}; ...; C_{2n}$ is a normal form. 

\textbf{Case 3}. 
$C= \iqif{\qbar}{\square_{i=1}^d |\psi_i\>}{C_i}$, where each $C_i$ is in a normal form, i.e., $C_i\equiv C_{i1}; \cdots; C_{im_i}$ with each $C_{ij}$ being a flat quantum $\kif$-statement.
First, let $\{|i\>\}$ be the standard basis of $\qbar$, and let $U$ be the unitary transformation such that $|\psi_i\>=U|i\>$ for all $i$.
Since $C$ is regular, there exist two sequential compositions $G_1,G_2$ of basic gates such that $\sem{G_1}=U$ and $\sem{G_2}=U^\dag$.
Then we have:
\begin{equation}
\label{norm-proof1}
    C\equiv G_2; \iqif{\qbar}{\square_{i=1}^d |i\>}{C_i}; G_1.
\end{equation}
The equivalence (\ref{norm-proof1}) is indeed a generalisation of the law of changing basis given in Proposition~\ref{qif-law}(1).
To prove (\ref{norm-proof1}), we note that $\qvar(G_1)\subseteq\qbar$, $\qvar(G_2)\subseteq\qbar$ and $\qbar\cap\left(\bigcup_i\qvar(C_i)\right)=\emptyset$.
Thus, it holds that
\begin{align*}
    \sem{C} &=\sum_i|\psi_i\>\<\psi_i|\otimes \sem{C_i} = \sum_i(U|i\>\< i|U^\dag\otimes \sem{C_i})\\
    &=U\Big(\sum_i|i\>\< i|\otimes \sem{C_i}\Big)U^\dag =\sem{G_2}\sem{\iqif{\qbar}{\square_{i=1}^d |i\>}{C_i}}\sem{G_1} =\sem{RHS}.
\end{align*}

Next, by applying Proposition~\ref{qif-law}(\ref{qif-law-splitting}) and~\ref{circuit-seq-law}(\ref{circuit-seq-law-sequentiality}), we rewrite $C$ as
\[
    C\equiv G_2; \big((\sqif{}{\qbar}{C_{11}}{}{|1\>}) ; \cdots ; (\sqif{}{\qbar}{C_{1m_1}}{}{|1\>}) \big); 
    \cdots; \big((\sqif{}{\qbar}{C_{d1}}{}{|d\>}) ; \cdots ; (\sqif{}{\qbar}{C_{dm_d}}{}{|d\>}) \big); G_1.
\]
Now, by the conclusions of Case 2, it is sufficient to show that each $\sqif{}{\qbar}{C_{ij}}{}{|i\>}$ is equivalent to a normal form.
Since $C_{ij}$ is a flat quantum $\kif$-statement, we can assume that $C_{ij} \equiv \iqif{\rbar}{\square_{k=1}^{d_r} |k\>}{D_k}$,  
where $D_k$ is a sequence of basic gates or $\iskip$. 
Then we claim
\[
    \sqif{}{\qbar}{C_{ij}}{}{|i\>} \equiv
    \iqif{\qbar,\rbar}{\square_{k=1}^{d_r} |i, k\>}{D_k},
\]
where the RHS is a normal form.
To show this, note that $\qbar\cap\rbar=\emptyset$ and $(\qbar,\rbar)\cap \qvar(D_k) = \emptyset$.
This can be rigorously proved by showing that the normal form we constructed for each sub-program shares the same variable sets with the original sub-program, except the external variable $r_0$ we used in Case 1, which can be chosen disjoint from the original program. We omit the tedious details here, but the reader can see them from the Coq implementation discussed in Section~\ref{sec-coq}. Thus, we have:
\begin{align*}
    \sem{\sqif{}{\qbar}{C_{ij}}{}{|i\>}} &= |i\>_{\qbar}\<i|\otimes\sem{C_{ij}} + \sum_{i^\prime\neq i}|i^\prime\>_{\qbar}\<i^\prime|\otimes I \\
    &= |i\>_{\qbar}\<i|\otimes\sum_k|k\>_{\rbar}\<k|\otimes\sem{D_k} + \sum_{i^\prime\neq i}|i^\prime\>_{\qbar}\<i^\prime|\otimes I\\
    &= \sum_k|i,k\>_{\qbar,\rbar}\<i,k|\otimes\sem{D_k} + \sum_{i^\prime\neq i}\sum_k|i^\prime,k\>_{\qbar}\<i^\prime,k|\otimes I \\
    &= \sem{\iqif{\qbar,\rbar}{\square_{k=1}^{d_r} |i, k\>}{D_k}}.
\end{align*}
\end{proof}

An application of the above normal form is equivalence checking of two quantum circuits $C_1$ and $C_2$, which is a key issue in the design and verification of quantum circuits and has been widely studied in the literature.
We can transform them to two normal forms $C_1^\prime$ and $C_2^\prime$, respectively.
Obviously, it is much easier to check the equivalence of $C_1^\prime$ and $C_2^\prime$ than that of $C_1$ and $C_2$. This approach is different from the existing ones in the literature (see for example~\cite{Wille, Markov1}).
A general exposition of this application is out of the scope of the present paper. Here, we only give a very simple example:

\begin{exam}[Equivalent checking via normal form]
    Consider the following two quantum circuits $C_1$ and $C_2$:
    \begin{align*}
        &C_1 \triangleq \sqif{\gateH[y]}{x}{(\sqif{\iskip}{y}{\gateX[z]}{|-\>}{|+\>}; \gateH[y])}{|-\>}{|+\>}, \\
        &C_2 \triangleq \gateH[x]; \gateH[y]; \sqif{}{y}{(\sqif{}{x}{\gateX[z]}{}{})}{}{}; \gateH[x].
    \end{align*}
    Since $C_1$ and $C_2$ contain nested quantum $\kif$-statements and are not in the standard basis, the equivalence is not obvious.
    We translate them into normal forms.
    Note that the proof of Theorem~\ref{normal-thm-cir} is not the unique way to rewrite the circuits; here we proceed in a slightly different and simpler manner:
    \begin{align*}
        C_1 &\equiv \sqif{\gateH[y]}{x}{(\gateH[y]; \sqif{}{y}{\gateX[z]}{}{}; \gateH[y]; \gateH[y])}{|-\>}{|+\>}
        & \reason{Prop.~\ref{qif-law}(\ref{qif-law-changing-basis}); \(\gateH^\dag = \gateH\)}\\
        &\equiv \gateH[x]; \sqif{\gateH[y]}{x}{(\gateH[y]; \sqif{}{y}{\gateX[z]}{}{}; \gateH[y]; \gateH[y])}{}{}; \gateH[x] 
        & \reason{Prop.~\ref{qif-law}(\ref{qif-law-changing-basis}); \(\gateH^\dag = \gateH\)}\\ 
        &\equiv \gateH[x]; \sqif{\gateH[y]}{x}{(\gateH[y]; \sqif{}{y}{\gateX[z]}{}{})}{}{}; \gateH[x] 
        & \reason{Prop.~\ref{circuit-seq-law}(\ref{circuit-seq-law-composition},~\ref{if-laws-associativity}); \(\gateH\gateH = I\)}\\ 
        &\equiv \gateH[x]; \gateH[y]; \sqif{}{x}{(\sqif{}{y}{\gateX[z]}{}{})}{}{}; \gateH[x] 
        & \reason{Prop.~\ref{circuit-seq-law}(\ref{circuit-seq-law-sequentiality},~\ref{circuit-seq-law-distributivity})} \\
        &\equiv \gateH[x]; \gateH[y]; 
        \sqif{}{y}{(\sqif{}{x}{\gateX[z]}{}{})}{}{}; \gateH[x] 
        & \reason{Props.~\ref{qif-law}(\ref{qif-law-distributivity},~\ref{qif-law-idempotence})} \\
        &\equiv \gateH[x]; \gateH[y]; \iqif{y,x}{|11\>}{\gateX[z]}; \gateH[x],
        & \reason{Prop.~\ref{qif-law}(\ref{qif-law-nested})} \\
        C_2 &\equiv \gateH[x]; \gateH[y]; \iqif{y,x}{|11\>}{\gateX[z]}; \gateH[x].
        & \reason{Prop.~\ref{qif-law}(\ref{qif-law-nested})}
    \end{align*}
    Now their equivalence is trivial, since they share the same normal form.
\end{exam}

An approach to classical compiler design based on normal forms of programs was proposed in~\cite{Hoare93} and extended to quantum programs in~\cite{Zu05}.
Thus, another possible application of Theorem~\ref{normal-thm-cir} is to combine it with the normal form of quantum programs (Theorem~\ref{thm-normal-prog}) for quantum compiler design.

\section{Algebraic Laws for Quantum Programs}\label{sec-laws-prog}

In this section, we turn to investigate laws for purely quantum programs in $\mathbf{QProg}$.
To present them, we need the notion of equivalence of two quantum programs: 

\begin{defn}
    Two quantum programs $P$ and $Q$ are called equivalent, written $P\equiv Q$, if they have the same denotational semantics; that is, $\sem{P}(\rho)=\sem{Q}(\rho)$ for all inputs $\rho\in\mathcal{D}(\hs_\QVar).$
\end{defn}

First of all, by induction on the structure of $P$, it is routine to prove the following lifting lemma, which gives a link from the equivalence of quantum programs to the equivalence of quantum circuits.   

\begin{lem}[Congruence - Lifting law]
    Let quantum program $P=P(C_1,...,C_n)$ contain quantum circuits $C_1,...,C_n$, and let $Q=P(D_1,...,D_n)$ be the quantum program obtained through substituting $C_1,...,C_n$ in $P$ by quantum circuits $D_1,...,D_n$, respectively.
    If $C_i\equiv D_i$ for $i=1,...,n$, then $P\equiv Q.$
\end{lem}

This lemma enables us to employ the laws of quantum circuits established in the last section in reasoning about the equivalence of quantum programs. 

The $\kif$-statement in Eq. (\ref{prog-syntax}) is defined for a general quantum measurement.
For simplicity of presentation, we mainly consider a special class of measurements in our laws for quantum programs.
By the term \textit{binary measurement}, we mean a measurement $M$ with only two possible outcomes, say $0$ and $1$, which can then be written as $M=\{M_0, M_1\}$; for instance, the measurement on a qubit in the computational basis.   
Following~\cite{Hoare87} again, we introduce an infix notation for an $\kif$-statement defined by a binary measurement $M=\{M_0, M_1\}$: 
\begin{align*}
    \sif{P_0}{M[q]}{P_1} \triangleq \iif{M}{q}{0 \rightarrow P_0\ \square\ 1\rightarrow P_1}{}.
\end{align*}
We often drop the measured quantum variable $q$ and simply write $\sif{P_0}{M}{P_1}$ for $\sif{P_0}{M[q]}{P_1}$ when the exact name $q$ is not essential.
In particular,
\begin{itemize}
    \item $\sif{\iskip}{M}{P}$ is often abbreviated to $\sif{}{M}{P}$;
    \item $\sif{}{M}{\iskip}$ is further abbreviated to $[M]$;
    \item $\sif{\iabort}{M}{\iskip}$ is abbreviated to $(M]$; 
    \item $\sif{\iskip}{M}{\iabort}$ is abbreviated to $[M)$. 
\end{itemize}
Intuitively, $[M]$ can be understood as a test or assertion using a measurement $M$ where we care about both outcomes $0$ and $1$.
But in the test $(M]$ (respectively, $[M)$), we only care the case of measurement outcome $1$ (respectively, outcome $0$).
The three tests $[M], (M]$, and $[M)$ are useful in testing, debugging, and runtime assertion of quantum programs.  

\subsection{Laws of Initialisation}

Let us now present several laws of initialisation.
A special kind of quantum measurement will be used in one of these laws: 

\begin{defn}
We write $M=[|\psi\>]$ for a binary measurement $M=\{M_0,M_1\}$ if it is defined by a pure quantum state $|\psi\>$ (normalized state vector), i.e., if $M_0=I-|\psi\>\<\psi|$ and $M_1=|\psi\>\<\psi|$.  
\end{defn}

Obviously, the measurement $[|\psi\>]$ tests whether a quantum system is in state $|\psi\>$ or not.
The laws given in the following proposition describe how initialisation interacts with other program constructs.

\begin{prop}[Laws of Initialization]
\label{init-laws}
    \quad
    \begin{enumerate}
        \item\label{init-laws-cancellation} (Cancellation)
        $\iinit{\qbar}{|\varphi\>}; \iinit{\qbar}{|\psi\>}\equiv \iinit{\qbar}{|\psi\>}.$ 
        If $\qbar_1\cap\qbar_2=\emptyset$, then 
        $$\iinit{\qbar_1}{|\phi\>}; \iinit{\qbar_2}{|\psi\>} \equiv \qbar_1, \iinit{\qbar_2}{|\phi\>}|\psi\>.$$
        In particular, if $\qvar(\qbar) = \qvar(\qbar_1)\cup \qvar(\qbar_2)$, then $\iinit{\qbar_1}{}; \iinit{\qbar_2}{} \equiv \iinit{\qbar}{}.$
        \item\label{init-laws-unitary-elimination} (Unitary-Elimination)
        If $\sem{C} = U_\qbar$ for some unitary operator $U$, then
        $\iinit{\qbar}{|\psi\>}; C\equiv \iinit{\qbar}{U|\psi\>}.$
        In particular, $\iinit{\qbar}{|\psi\>}; U[\qbar] \equiv \iinit{\qbar}{U|\psi\>}$. 
        \item\label{init-laws-qif-elimination} (qif-Elimination)
        We write $|\alpha\>\bot|\beta\>$ if two states $|\alpha\>, |\beta\>$ are orthogonal, i.e., $\<\alpha|\beta\>=0$.
        \[
            \hspace{1cm}\frac{|\psi\>\bot |\varphi_0\>}{ \iinit{q}{|\psi\>}; \sqif{C_0}{q}{C_1}{|\varphi_0\>}{|\varphi_1\>}
            \equiv \iinit{q}{|\varphi_1\>}; C_1\equiv \iinit{q}{|\psi\>}; C_1}.
        \]
        More generally, if for each $i$, there exists $U_i$ such that $\sem{C_i} = (U_i)_{\rbar}$, then
        \[
            \hspace{1cm}\iinit{\qbar,\rbar}{\sum_i\lambda_i|\phi_i\>|\psi_i\>}; 
            \iqif{\qbar}{\square_{i=1}^d |\phi_i\>}{C_i}
            \equiv \iinit{\qbar,\rbar}{\sum_i\lambda_i|\phi_i\>(U_i|\psi_i\>)}.
        \]
        \item\label{init-laws-if-elimination} (if-Elimination)
        The left-elimination 
        $$[M[q]]; \iinit{q}{|\phi\>} \equiv \iinit{q}{|\phi\>},$$
        and the right-eliminations
        \[
            \hspace{1cm}\frac{M=[|\psi\>]}{ \iinit{q}{|\psi\>}; \sif{P_0}{M[q]}{P_1} \equiv \iinit{q}{|\psi\>}; P_1},\qquad \frac{M=[|\psi\>]\qquad |\varphi\>\bot |\psi\>}{ \iinit{q}{|\varphi\>}; \sif{P_0}{M[q]}{P_1}\equiv \iinit{q}{|\varphi\>}; P_0}.
        \]
        \item \label{init-laws-if-expansion} (if-Expansion)
        If $\qbar\subseteq \rbar$ and for all $m$, $\rbar\cap\qvar(P_m) = \emptyset$, then
        \[
            \iif{}{\qbar}{m}{P_m}; \iinit{\rbar}{}
            \equiv \iif{}{\rbar}{k}{P_{k{\downarrow}_{\qbar}}}; \iinit{\rbar}{}, 
        \]
        where $m$ and $k$ are bit strings, and $k{\downarrow}_\qbar$ denotes the restriction of $k$ on $\qbar$.
        Intuitively, a register can be additionally measured before its initialisation.
    \end{enumerate}
\end{prop}

\subsection{Laws of $\kif$-statement}

In this subsection, we give some laws of $\kif$-statements.
To this end, we need the following: 

\begin{defn}
    Given two linear operators $A$ and $B$ and non-negative real number $k\in \mathbb{R}$, we define:
    \begin{enumerate}
        \item ($k$-\textbf{proportion}) $A\lrtimes_k B$ if there exists $c\in\mathbb{C}$ such that  $|c|=k$ and $A = cB$.
        Whenever $k = 1$, we simply write $A\lrtimes B$.
        \item ($k$-\textbf{left-absorption}) $A\ltimes_k B \triangleq AB\lrtimes_k B$, i.e., if there exists $c\in\mathbb{C}$ such that $|c|=k$ and $AB = cB$.
        Whenever $k = 1$, we simply write $A\ltimes B$.
        \item ($k$-\textbf{right-absorption}) $A\rtimes_k B \triangleq AB\lrtimes_k A$, i.e., if there exists $c\in\mathbb{C}$ such that $|c|=k$ and $AB = cA$.
        Whenever $k = 1$, we simply write $A\rtimes B$.
        \item (\textbf{Orthogonal}) $A\perp B$ if $AB = 0$.
        It is obvious that $A\perp B \leftrightarrow A\ltimes_0 B\leftrightarrow A\rtimes_0 B$.
    \end{enumerate}
\end{defn}

\begin{defn}
\label{defn-measurement-relation}
    Let $M=\{M_0,M_1\}$ and $N=\{N_0,N_1\}$ be two binary measurements.
    Then:
    \begin{enumerate}
        \item $N$ is called the complement of $M$, written $N=M^\bot$, if $N_0=M_1$ and $N_1=M_0$. 
        \item We say that $M$ entails $N$, written $M\blacktriangleright N$, if $N_0 \perp M_1$. 
        \item We say $M$ is logically weaker than $N$, written $M\propto N$, if $M_1\ltimes N_1$.
        Whenever $M\propto N^\bot$, we simply write $M\gg N$, and say $M$ logically contradicts to $N$.
        \item A binary measurement $K=\{K_0,K_1\}$ is called a pseudo-meet of $M$ and $N$, written $K \cong M\smeet N$, if $K_1\lrtimes N_1M_1$ (or, equivalently, $N_1M_1 \lrtimes K_1$).
    \end{enumerate}
\end{defn}

The following example can help us understand the notions introduced in the above definition. 
\begin{exam}
\label{exam-measurement-relation}
    \begin{enumerate}
    \item The intuition behind entailment, logical weakness and contradiction becomes clear when we consider projective measurements. 
    Let two projective measurements $M=\{M_0,M_1\}$ and $N=\{N_0,N_1\}$ with $M_0=I-P_X$, $M_1=P_X$, $N_0=I-P_Y$ and $N_1=P_Y$, where $P_X,P_Y$ are the projections onto closed subspaces $X$ and $Y$, respectively, then:
    \begin{itemize}
        \item $M\blacktriangleright N$ if and only if $X\subseteq Y$;
        \item $M\propto N$ if and only if $Y\subseteq X$;
        \item $M\gg N$ if and only if $X^\bot\subseteq Y$.
    \end{itemize}
    
    Thus, the following properties hold for projective measurements $M$ and $N$:
    \[
        M\blacktriangleright N \Leftrightarrow 
        N^\bot\blacktriangleright M^\bot \Leftrightarrow 
        N\propto M \Leftrightarrow 
        M^\bot \propto N^\bot \Leftrightarrow 
        M^\bot \gg N \Leftrightarrow
        N \gg M^\bot.
    \]

    \item In the definition of pseudo-meet $K\cong M \smeet N$, many different choices of $K_0$ are possible; one of them is $K_0=\sqrt{M_0^\dag M_0+M_1^\dag N_0^\dag N_0M_1}$.
    
    \item $M\propto N$ ($M$ is logically weaker than $N$) if and only if $N\cong N\smeet M$.
    Also, we have: $M\propto N$ implies $N\blacktriangleright M$.
    \end{enumerate}
\end{exam}

Now we are able to present a set of simple laws for $\kif$-statements: 
\begin{prop}[Laws of $\kif$-statement]
\label{if-laws}
    \quad
    \begin{enumerate}
        \item\label{if-laws-truth-falsity} (Truth-Falsity)
        For the trivial measurements $M_T=\{M_{T0},M_{T1}\}$ and $M_F=\{M_{F0},M_{F1}\}$ with $M_{T0}=M_{F1}=0$ and $M_{T1}=M_{F0}=I$, we have: 
        \[
            \sif{P}{M_T}{Q}\equiv Q,\qquad
            \sif{P}{M_F}{Q}\equiv P.
        \]
        \item\label{if-laws-idempotence} (Idempotence)
        $\sif{P}{M}{P}\equiv [M]; P,$ where $[M]$ stands for test $\sif{}{M}{\iskip}$. 
        \item\label{if-laws-complementation} (Complementation)
        $\sif{P}{M^\bot}{Q} \equiv \sif{Q}{M}{P}.$
        \item\label{if-laws-associativity} (Associativity)
        If $M$ is a projective measurement, then 
        \[
            \sif{P}{M}{(\sif{Q}{M}{R})} \equiv
            \sif{(\sif{P}{M}{Q})}{M}{R} \equiv \sif{P}{M}{R}.
        \]
        \item\label{if-laws-if-elimination} (if-Elimination)
        \[
            \frac{M\blacktriangleright N\qquad K\cong M\smeet N}{(M]; (\sif{P}{N}{Q})\equiv (K]; Q}, \qquad 
            \frac{M\propto N}{\sif{P}{N}{(\sif{Q}{M}{R})}\equiv \sif{P}{N}{R}}.
        \]
    \end{enumerate}
\end{prop}

The laws given in the above proposition deserve some explanation.
In particular, it is interesting to compare them with the corresponding laws of classical programming given in Figure~\ref{fig 1}.
Clause (\ref{if-laws-truth-falsity}) mimics the truth-falsity laws (If-1) and (If-2) in Figure~\ref{fig 1}.
There is a basic difference between the idempotence law (If-3) in classical programming and that in quantum programming.
In the classical case, we have $\sif{P}{b}{P}\equiv P$.
In contrast, a testing $[M]$ is added in the right-hand side of Clause (\ref{if-laws-idempotence}).
This difference comes from the fact that checking condition $b$ does not change the state of classical program variables, but performing measurement $M$ changes the state of quantum variable $q$.
Clause (\ref{if-laws-complementation}) is a quantum generalisation of law (If-5).
The associativity law (\ref{if-laws-associativity}) is a quantum generalisation of laws (If-4) and (If-7).
Note that it only holds for a projective measurement $M$, and we can easily find a counter-example showing that it is not true for a general measurement $M$.
Clause (\ref{if-laws-if-elimination}) is a complement of Proposition~\ref{init-laws}(\ref{init-laws-if-elimination}).

To present quantum generalisations of laws (If-6) and (If-8) in Figure~\ref{fig 1}  for nested $\kif$-statements, we need the following: 
\begin{defn}
    Let $M=\{M_0,M_1\}$ and $N=\{N_0,N_1\}$ be two binary quantum measurements.
    Then we say that $M$ and $N$ left-commute (respectively, right-commute), written $M\diamond_L N$ (respectively, $M\diamond_R N$), if $M_0N_1=N_1M_0$ and $M_1N_1=N_1M_1$ (respectively, $M_0N_0=N_0M_0$ and $M_1N_0=N_0M_1$).
\end{defn}

It is obvious that $M\diamond_L N$ if and only if $M^\bot\diamond_R N$.
The following example can help us to understand the commutativity introduced in the above definition. 

\begin{exam}
    If both $M=\{M_0,M_1\}$ and $N=\{N_0,N_1\}$ are projective measurements, say, $M_0=I-P_X, M_1=P_X, N_0=I-P_Y$ and $N_1=P_Y$, where $P_X$ and $P_Y$ are the projection onto closed subspaces $X,Y$, respectively, then $M\diamond_L N$ and $M\diamond_R N$ if and only if $P_XP_Y$ is also a projection. In this case, $P_XP_Y=P_{X\cap Y}.$
    Furthermore, if $M\blacktriangleright N$ or $M\propto N$, then $M\diamond_L N$, $M\diamond_R N$, $M\diamond_L N^\bot$ and $M\diamond_R N^\bot$.
\end{exam}

Now our laws for nested $\kif$-statements in quantum programming can be given in the following:
\begin{prop}[Laws of nested $\kif$-statement]
\label{nested-if-law}
    \quad
    \begin{enumerate}
        \item\label{nested-if-law-reduction} (Reduction)
        \begin{equation*}
            \frac{(N]\equiv \sif{(K]}{M}{(L]} \qquad [N)\equiv \sif{[K)}{M}{[L)}}
            {\sif{(\sif{P}{K}{Q})}{M}{(\sif{P}{L}{Q})} \equiv \sif{P}{N}{Q}}.
        \end{equation*}

        For measurement in the computational basis, we further have: if $\qbar\cap\rbar = \emptyset$, then
        \begin{equation*}
            \iif{}{\qbar}{m}{\iif{}{\rbar}{n}{P_{mn}}} \equiv
            \iif{}{\qbar,\rbar}{(m,n)}{P_{mn}}.
        \end{equation*}
        
        \item\label{nested-if-law-left-distributivity} (Left-Distributivity)
        \begin{equation*}
            \frac{M\diamond_L N\qquad [N)\equiv [M]; [N)}{\sif{R}{N}{(\sif{P}{M}{Q})}\equiv 
            \sif{(\sif{R}{N}{P})}{M}{(\sif{R}{N}{Q})}}.
        \end{equation*}
        
        \item\label{nested-if-law-right-distributivity} (Right-Distributivity)
        \begin{equation*}
            \frac{M\diamond_R N\qquad (N]\equiv [M]; (N]}{\sif{(\sif{P}{M}{Q})}{N}{R} \equiv \sif{(\sif{P}{N}{R})}{M}{(\sif{Q}{N}{R})}}.
        \end{equation*}
        
        \item\label{nested-if-law-left-distributivity-projection} (Left-Distributivity for Projection)
        If $M,N$ are projective measurements,
        \begin{equation*}
            \frac{ M \blacktriangleright N \quad\mbox{ or }\quad  M^\bot \blacktriangleright N }{\sif{R}{N}{(\sif{P}{M}{Q})} \equiv \sif{(\sif{R}{N}{P})}{M}{(\sif{R}{N}{Q})}}.
        \end{equation*}
        
        \item\label{nested-if-law-right-distributivity-projection} (Right-Distributivity for Projection) If $M,N$ are projective measurements,
        \begin{equation*}
            \frac{ N \blacktriangleright M \quad\mbox{ or }\quad  N \blacktriangleright M^\bot }{\sif{(\sif{P}{M}{Q})}{N}{R} \equiv \sif{(\sif{P}{N}{R})}{M}{(\sif{Q}{N}{R})}}.
        \end{equation*}
    \end{enumerate}
\end{prop}

Obviously, Clause (\ref{nested-if-law-reduction}) in the above proposition is a quantum generalisation of law (If-6) in Figure~\ref{fig 1}, and Clauses (\ref{nested-if-law-left-distributivity}) and (\ref{nested-if-law-right-distributivity}) are quantum generalisations of law (If-8) and its dual.
But as one can expect, certain conditions about the involved measurements are needed in the quantum case.
These laws will be very useful for reducing the number of $\kif$-statements in quantum program optimisation.
For example, by applying equations in Clause (\ref{nested-if-law-reduction}) from left to right, and equations in Clauses (\ref{nested-if-law-left-distributivity}) - (\ref{nested-if-law-right-distributivity-projection}) from right to left, the number of $\kif$-statements can be reduced.

\subsection{Laws of sequential composition}

Several basic laws for the sequential composition of quantum programs are collected in the following: 

\begin{prop}[Laws of sequential composition]
\label{sequential-laws}
    \quad
    \begin{enumerate}
        \item\label{sequential-laws-unit-zero} (Unit and Zero)
        $\iskip; P\equiv P; \iskip; P\equiv P$ and $\iabort; P\equiv P; \iabort; P\equiv \iabort.$
        \item\label{sequential-laws-commutativity} (Commutativity)
        If $\qvar(P_1)\cap\qvar(P_2)=\emptyset$, then $P_1; P_2\equiv P_2; P_1$.  
        \item\label{sequential-laws-associativity} (Associativity)
        $(P; Q); R=P; (Q; R).$
        \item\label{sequential-laws-sequentiality} (Sequentiality)
        If $M$ is a projective measurement, and
        $\qbar\cap \qvar(P_0) = \qbar\cap \qvar(P_1) = \emptyset$, then
        \[
            (\sif{P_0}{M[\qbar]}{P_1}); 
            (\sif{Q_0}{M[\qbar]}{Q_1}) \equiv
            \sif{(P_0; Q_0)}{M[\qbar]}{(P_1; Q_1)}.
        \]
        \item\label{sequential-laws-right-distributivity} (Right-Distributivity)
        $(\sif{P_0}{M}{P_1}); P\equiv \sif{(P_0; P)}{M}{(P_1; P)}.$
        \item\label{sequential-laws-left-distributivity1} (Left-Distributivity-I)
        If $\qbar\cap \qvar(P) = \emptyset$, then
        $$P; (\sif{P_0}{M[\qbar]}{P_1})\equiv \sif{(P; P_0)}{M[\qbar]}{(P; P_1)}.$$
        \item\label{sequential-laws-left-distributivity2} (Left-Distributivity-II)
        Generally, we have:
        \[
            \frac{P; [M)\equiv [N); P'\qquad P; (M]\equiv (N]; P'}{P; (\sif{P_0}{M}{P_1})\equiv \sif{(P'; P_0)}{N}{(P'; P_1)}}.
        \]
        In particular, if $U$ is a unitary operator and $U(M)=\{M_0^\prime,M_1^\prime\}$ is the binary measurement with $M_i^\prime=M_i U$ for $i=0,1$, then $U[q]; (\sif{P_0}{M[q]}{P_1})\equiv \sif{P_0}{U(M)[q]}{P_1}.$ 
    \end{enumerate}
\end{prop}

Clauses (\ref{sequential-laws-unit-zero}), (\ref{sequential-laws-associativity}) and (\ref{sequential-laws-right-distributivity}) in the above proposition are quantum generalisations of classical laws (Sq-1),(Sq-2), (Sq-3) and (Sq-4) in Figure~\ref{fig 1}, respectively. 

\subsection{Interplay between if-Statements and Quantum if-Statements}

Up to now, in this section, all laws are concerned only with the second layer of purely quantum programs in our quantum programming framework.
In this subsection, we present the following law that shows how a quantum $\kif$-statement is absorbed into an $\kif$-statement, and thus connects the second layer with the first layer of quantum circuits. 
\begin{prop}
\label{prop:interplay-qif-if}
    Let $\mathcal{B} = \{|\psi_i\>\}_{i=1}^d$ be an orthonormal basis and $M = \{M_i\}$ be the measurement in basis $\mathcal{B}$; that is, $M_i=|\psi_i\>\<\psi_i|$ for every $i$.
    Then 
    \begin{align*}
        \iqif{\qbar}{\square_{i=1}^d |\psi_i\>}{C_i}; 
        \iif{M}{\qbar}{i}{P_i}
        \equiv \iif{M}{\qbar}{i}{C_i; P_i}.
    \end{align*}
\end{prop}
This law will play a crucial role in the application given in Section~\ref{sec-app}.

\newcommand {\tblue}[1]{{{\color{blue}{#1}}}}
\newcommand {\tred}[1]{{{\color{blue}{#1}}}}
\newcommand {\tgreen}[1]{{{\color{green}{#1}}}}
\newcommand {\eqnsmall}[1]{{{\mbox{$#1$}}}}

\begin{exam}[Correctness of Quantum Teleportation]
    Quantum teleportation is a simple yet important protocol that reveals how entanglement helps transmit quantum states without a quantum channel.
    It can be written in our language as follows:
    \begin{align*}
        \mathbf{Tel} \triangleq \ & \iinit{q}{}; \iinit{r}{}; \gateH[q]; \sqif{}{q}{\gateX[r]}{}{}; \\
        & \iinit{p}{a|0\> + b|1\>}; \sqif{}{p}{\gateX[q]}{}{}; \gateH[p]; \\
        & \sif{}{q}{\gateX[r]}; \sif{}{p}{\gateZ[r]}; \iinit{p}{}; \iinit{q}{}.
    \end{align*}
    The first line prepares the initial Bell state $|\Phi^+\> \triangleq \frac{1}{\sqrt{2}}(|00\> + |11\>)$ between qubits $q,r$.
    The second line prepares the state $a|0\>+b|1\>$ in qubit $p$ and performs the local computation.
    The third line sends the classical information and performs the recovery, and finally resets $p,q$ to the default state $|0\>$.
    The correctness of the teleportation protocol can be specified as the following equivalence: 
    $$\mathbf{Tel} \equiv \iinit{p}{}; \iinit{q}{}; \iinit{r}{a|0\> + b|1\>},$$
    which asserts that the final state of qubit $r$ is the original state we want to transmit.

    To prove it, we first show the functionality of the first line:
    \begin{align*}
        &\iinit{q}{}; \iinit{r}{}; \gateH[q]; \sqif{}{q}{\gateX[r]}{}{} & \\
        \equiv\ & \iinit{q}{\gateH|0\>}; \iinit{r}{}; \sqif{}{q}{\gateX[r]}{}{}
            & \reason{Prop.~\ref{sequential-laws}(\ref{sequential-laws-commutativity},~\ref{sequential-laws-associativity}),~\ref{init-laws}(\ref{init-laws-unitary-elimination})}\\
        \equiv\ & \iinit{q,r}{\eqnsmall{\frac{1}{\sqrt{2}}}(|00\> + |10\>)}; \sqif{}{q}{\gateX[r]}{}{}
            & \reason{Prop.~\ref{init-laws}(\ref{init-laws-cancellation})}\\
        \equiv\ & \iinit{q,r}{\eqnsmall{\frac{1}{\sqrt{2}}}(|00\> + |11\>)}.
            & \reason{Prop.~\ref{init-laws}(\ref{init-laws-qif-elimination},~\ref{init-laws-unitary-elimination})}
    \end{align*}

    Next, we observe that the third line is equivalent to:
    \begin{align*}
        &\sif{}{q}{\gateX[r]}; \sif{}{p}{\gateZ[r]}; \iinit{p}{}; \iinit{q}{} & \\
        \equiv\ &\sqif{}{q}{\gateX[r]}{}{}; [M[q]]; \sqif{}{p}{\gateZ[r]}{}{}; [M[p]]; \iinit{p}{}; \iinit{q}{} &
        \reason{Prop.~\ref{prop:interplay-qif-if}}\\
        \equiv\ &\sqif{}{q}{\gateX[r]}{}{}; \sqif{}{p}{\gateZ[r]}{}{}; ([M[p]]; \iinit{p}{}); ([M[q]]; \iinit{q}{}) &
        \reason{Prop.~\ref{sequential-laws}(\ref{sequential-laws-commutativity},~\ref{sequential-laws-associativity})} \\
        \equiv\ &\sqif{}{q}{\gateX[r]}{}{}; \sqif{}{p}{\gateZ[r]}{}{}; \iinit{p}{}; \iinit{q}{}. &
        \reason{Prop.~\ref{init-laws}(\ref{init-laws-if-elimination})}
    \end{align*}

    Finally, we combine these results to complete the proof:
    \begin{align*}
        & \iinit{q,r}{|\Phi^+\>}; \iinit{p}{a|0\> + b|1\>}; \sqif{}{p}{\gateX[q]}{}{}; \gateH[p];
        \sqif{}{q}{\gateX[r]}{}{}; \sqif{}{p}{\gateZ[r]}{}{}; \iinit{p}{}; \iinit{q}{} \hspace{-5cm} & \\
        \equiv\ & \iinit{q,r}{|\Phi^+\>}; \iinit{p}{a|0\> + b|1\>}; \sqif{}{p}{\gateX[q]}{}{}; 
        \sqif{}{q}{\gateX[r]}{}{}; \gateH[p]; \sqif{}{p}{\gateZ[r]}{}{}; \iinit{p}{}; \iinit{q}{} \hspace{-5cm} & \\
        \equiv\ & \iinit{q}{|+\>};  \iinit{p,r}{a|00\>+b|11\>}; \gateH[p]; \sqif{}{p}{\gateZ[r]}{}{}; \iinit{p}{}; \iinit{q}{} & \\ 
        \equiv\ & \iinit{q}{|+\>}; \iinit{p,r}{a|{+}0\>+b|{-}1\>}; \sqif{}{p}{\gateZ[r]}{}{}; \iinit{p}{}; \iinit{q}{} & 
        \reason{Prop.~\ref{init-laws}(\ref{init-laws-unitary-elimination})}\\
        \equiv\ & \iinit{q}{|+\>}; \iinit{p,r}{a|{+}0\>+b|{+}1\>}; \iinit{p}{}; \iinit{q}{} &
        \reason{Prop.~\ref{init-laws}(\ref{init-laws-qif-elimination},~\ref{init-laws-unitary-elimination})} \\
        \equiv\ & \iinit{q}{|+\>}; \iinit{p}{|+\>}; \iinit{r}{a|0\>+b|1\>}; \iinit{p}{}; \iinit{q}{} & 
        \reason{Prop.~\ref{init-laws}(\ref{init-laws-cancellation})} \\
        \equiv\ & \iinit{p}{|+\>}; \iinit{p}{}; \iinit{q}{|+\>}; \iinit{q}{}; \iinit{r}{a|0\>+b|1\>} & 
        \reason{Prop.~\ref{sequential-laws}(\ref{sequential-laws-commutativity},~\ref{sequential-laws-associativity})}\\
        \equiv\ & \iinit{p}{}; \iinit{q}{}; \iinit{r}{a|0\>+b|1\>}, &
        \reason{Prop.~\ref{init-laws}(\ref{init-laws-cancellation})}
    \end{align*}
    where the first equality is due to Proposition~\ref{sequential-laws}(\ref{sequential-laws-commutativity},~\ref{sequential-laws-associativity}) and the second equality can be derived by
    \begin{align*}
        &\iinit{q,r}{|\Phi^+\>}; \iinit{p}{a|0\> + b|1\>}; \sqif{}{p}{\gateX[q]}{}{}; \sqif{}{q}{\gateX[r]}{}{} & \\
        \equiv\ & \iinit{p,q,r}{\eqnsmall{\frac{a}{\sqrt{2}}}|0\>(|00\>+|11\>) + \eqnsmall{\frac{b}{\sqrt{2}}}|1\>(|00\>+|11\>)}; \sqif{}{p}{\gateX[q]}{}{}; \sqif{}{q}{\gateX[r]}{}{} &
        \reason{Prop.~\ref{init-laws}(\ref{init-laws-cancellation})} \\
        \equiv\ & \iinit{p,q,r}{\eqnsmall{\frac{a}{\sqrt{2}}}(|000\>+|011\>) + \eqnsmall{\frac{b}{\sqrt{2}}}(|110\>+|101\>)}; \sqif{}{q}{\gateX[r]}{}{} &
        \reason{Prop.~\ref{init-laws}(\ref{init-laws-qif-elimination},~\ref{init-laws-unitary-elimination})} \\
        \equiv\ & \iinit{p,q,r}{\eqnsmall{\frac{a}{\sqrt{2}}}(|000\>+|010\>) + \eqnsmall{\frac{b}{\sqrt{2}}}(|111\>+|101\>)} &
        \reason{Prop.~\ref{init-laws}(\ref{init-laws-qif-elimination},~\ref{init-laws-unitary-elimination})} \\
        \equiv\ & \iinit{q}{|+\>}; \iinit{p,r}{a|00\>+b|11\>}. &
        \reason{Prop.~\ref{init-laws}(\ref{init-laws-cancellation})}.
    \end{align*}
\end{exam}

\subsection{Normal Forms of Finite, Purely Quantum Programs}\label{sec-finite}

In this subsection, we establish a normal form of purely quantum programs in $\mathbf{QProg}$ using the laws introduced in the previous subsections. 
\begin{defn}
    A normal form of quantum programs is an $\kif$-statement 
    \begin{equation}
    \label{normal-prog}
        \iif{M}{\qbar}{m}{C_m}
    \end{equation}
    where each $C_m$ is either $\iabort$ or a quantum circuit. 
\end{defn}

\begin{thm}[Normal forms]
\label{thm-normal-prog}
    Each finite, purely quantum program (i.e., a program that does not contain iteration or recursion) $P\in\mathbf{Qprog}$ is equivalent to a normal form.
\end{thm} 

\begin{proof}We proceed by induction on the structure of $P$. 

\textbf{Case 1}.
$P=\iskip$ or $\iabort$.
Then we can choose a trivial measurement $M=\{M_0,M_1\}$ with $M_0=I$ and $M_1=0$, where $I,0$ are the identity and zero operators, respectively, on the $2$-dimensional Hilbert space.
Let $q$ be a qubit. Then we have $\iskip\equiv \sif{\iskip}{M[q]}{\iabort}$ and $\iabort\equiv \sif{\iabort}{M[q]}{\iskip}.$
Both of them are in a normal form. 

\textbf{Case 2}.
$P=P_1; P_2$, and $P_1,P_2$ have the following normal forms: 
\begin{align*}
    P_1\equiv \iif{M}{\qbar}{m}{C_m},\qquad 
    P_2\equiv \iif{N}{\rbar}{n}{D_n}.
\end{align*}
For each $m$, let the semantics of quantum circuit $C_m$ be a unitary operator $U_m$.
Then for any density operator $\rho$, by the linearity and Proposition~\ref{sempro}(\ref{sempro-seq}) and (\ref{sempro-if}) we obtain:
\begin{align*}
    &\sem{P}(\rho)=\sem{P_2}(\sem{P_1}(\rho)) =\sem{P_2}\left(\sum_m\sem{C_m}\left(M_m\rho M_m^\dag\right)\right)
    =\sum_m\sem{P_2}\left(\sem{C_m}\left(M_m\rho M_m^\dag\right)\right)\\  
    &=\sum_m\left(\sum_n \sem{D_n}\left(N_n\sem{C_m}\left(M_m\rho M_m^\dag\right)N_n^\dag\right)\right)=\sum_{m,n}\sem{D_n}\left(\left(N_nU_mM_m\right)\rho\left(N_nU_mM_m\right)^\dag\right).
\end{align*}
For any $m,n$, we define operator $K_{m,n}=N_nU_mM_m$.
Then by the unitarity of $U_m$, it is easy to check that $\sum_{m,n}K_{m,n}^\dag K_{m,n}=I$ and $K=\left\{K_{m,n}\right\}$ is a quantum measurement.
Thus, we have: 
$$P\equiv \iif{K}{\qbar\cup\rbar\cup(\mbox{$\bigcup_m$}\qvar(C_m))}{(m,n)}{D_{m,n}},$$ which is in a normal form, where $D_{m,n}=D_n$ for every $m$. 

\textbf{Case 3}.
$P=\iif{M}{\qbar}{m}{P_m}$, and each $P_m$ has normal form: 
$$P_m\equiv \iif{}{}{\square n\in\Delta_m\cdot N_m[\rbar_m] = n \rightarrow C_{m,n}}{}.$$
Then for any density operator $\rho$, with Proposition~\ref{sempro}(6) we have:
\begin{align*}
    \sem{P}(\rho)&=\sum_m\sem{P_m}\left(M_m\rho M_m^\dag\right) 
    = \sum_m\left(\sum_{n\in\Delta_m}\sem{C_{m,n}}\left(N_{m,n}M_m\rho M_m^\dag N_{m,n}^\dag\right)\right)\\ 
    &=\sum_{(m,n)\in\Delta}\sem{C_{m,n}} \left(\left(N_{m,n}M_m\right)\rho\left(N_{m,n}M_m\right)^\dag\right),
\end{align*}
where $\Delta=\bigcup_m\left(\{m\}\times\Delta_m\right)$.
For each $(m,n)\in\Delta$, we define operator $K_{m,n}=N_{m,n}M_m$.
Then it is easy to verify that $\sum_{(m,n)\in\Delta}K_{m,n}^\dag K_{m,n}=I$ and $K=\left\{K_{m,n}\right\}_{(m,n)\in\Delta}$ is a quantum measurement.
Consequently, it holds that
$$P\equiv
\iif{K}{\qbar\cup(\mbox{$\bigcup_m$}\rbar_m)}{(m,n)}{C_{m,n}},$$
which is in a normal form. 
\end{proof}

\section{Laws for Recursion}\label{sec-recur}

In this section, we prove some laws for recursive quantum programs.
The syntax and semantics of general recursive quantum programs (with classical control flows) were defined in~\cite{Ying16, Xu-Ying21}.
For simplicity of presentation, we consider a simple variant of them here.
Let $X$ be a program identifier, and let $F(X)$ be a quantum program defined by adding $X$ into the syntax (\ref{prog-syntax}).
Then we write $\mu X.F(X)$ for the recursive quantum program declared by the recursive equation $X=F(X)$.
Intuitively, each occurrence of $X$ within $F(X)$ is a call on the recursive program. 

For any program $P$, we write $F(P)\triangleq F(X)[P/X]$ for the program obtained by replacing $X$ in $F(X)$ with $P$.
Furthermore, we define $F^{(n)}$ as the $n$th syntactic approximation of recursive quantum program $\mu X.F(X)$, i.e., the $n$-fold iteration of $F$:
\begin{align*}
    \begin{cases}
        F^{(0)}=\iabort,\\
        F^{(n+1)}=F\left(F^{(n)}\right)\ {\rm for}\ n\geq 0.
    \end{cases}
\end{align*}
From its definition, we see that $F^{(n)}$ behaves like $\mu X.F(X)$ up to the depth $n$ of recursive calls.
Then the following proposition, proved in~\cite{Ying16}, describes the denotational semantics of recursive quantum programs.  

\begin{prop}\label{rec-sem}
    $\sem{\mu X.F(X)}=\bigsqcup_{n=0}^\infty\sem{F^{(n)}}$, where $\bigsqcup$ stands for the least upper bound in the CPO $\left(\mathcal{QO}(\hs_\QVar),\sqsubseteq\right)$ of quantum operations on $\hs_\QVar$, as in Proposition~\ref{sempro}(\ref{sempro-while}).
\end{prop}

A fixed point characterisation of recursive quantum programs follows immediately: 
\begin{prop}[Fixed Point]
\label{prop-fix}
    $\mu X.F(X)$ is the least solution of equation $X=F(X)$; that is,
    \begin{enumerate}
        \item $\mu X.F(X)\equiv F(\mu X.F(X))$, and 
        \item For any quantum program $P$, $P\equiv F(P)\Rightarrow \mu X.F(X)\sqsubseteq P$, where $Q\sqsubseteq P$ means that $\sem{Q}(\rho)\sqsubseteq \sem{P}(\rho)$ for any density operator $\rho$. 
    \end{enumerate}  
\end{prop}

Intuitively, the above proposition indicates that a recursive call is equivalent to the whole program.
The two laws in this proposition for recursive quantum programs look the same as laws (Re-2) and (Re-3) in Figure~\ref{fig 1} for classical recursive programs. 
Sharing a similar spirit of using the Kleene fixed-point theorem in both the classical and quantum cases, the proofs of these laws are nevertheless quite different, particularly in the Coq proof assistant, as will be described in Section~\ref{sec-coq}, because we need to handle limits of quantum operations (i.e., super-operators) on $\hs_\QVar$ in the quantum case.

\subsection{Laws of Loops}
As in classical programming, the quantum $\kwhile$-loop introduced in Definition~\ref{def-prog} is a special case of the notion of recursion defined above.
Following~\cite{Hoare87}, we introduce a convenient notation for quantum $\kwhile$-loop: 
$$\swhile{M[\qbar]}{P} \triangleq \iwhile{M}{\qbar}{P}.$$
If the name of measured quantum variables $\qbar$ is not essential or can be known from the context, we often write $\swhile{M}{P}$ for $\swhile{M[\qbar]}{P}$.
Then it is easy to see that 
\begin{equation}
\label{loop-to-rec}
    \swhile{M}{P}\equiv \mu X. (\sif{}{M}{P}; X),
\end{equation}
if $P$ does not contain any further recursive call.
Obviously, Eq. (\ref{loop-to-rec}) is a quantum generalisation of law (Re-1) in Figure~\ref{fig 1}.

As shown in Example~\ref{exam-walk}, quantum random walks (without absorbing boundaries) can be modelled as a fixed number of iterations of a unitary operator $W$.
Their variants with absorbing boundaries can be properly described as a quantum loop: 

\begin{exam}[Quantum random walk with absorbing boundary]
    \[
        \swhile{M[p]}{(\sqif{T_L[p]}{\gateH[d]}{T_R[p]}{}{})}
    \]
    where we adopt the notations from Example~\ref{exam-walk}, and $M$ represents a measurement that detects whether the particle has reached the left or right boundary.
    More precisely, if the left and right boundaries are set at position $n_l$ and $n_r$, respectively, then the measurement is modelled by $M=\{M_0,M_1\}$, where $M_0=|n_l\>\< n_l| + |n_r\>\< n_r|$ and $M_1=I_p-M_0$, and $I_p$ is the identity operator on the position space $\hs_p$. 
\end{exam}

The following proposition presents several useful laws for loops.
It generalises the laws established in~\cite{Hoare93} for classical loops to quantum loops. 

\begin{prop}[Laws for  loop]
\label{loop-law}
    \quad
    \begin{enumerate}
        \item\label{loop-law-fixpoint} (Loop as a Fixpoint)
        $ \swhile{M}{P} \equiv \sif{}{M}{(P; (\swhile{M}{P}))}.$
        \item\label{loop-law-unfolding} (Loop Unfolding)
        $$\frac{M\propto N}{\sif{}{N}{(\swhile{M}{P})}\equiv \sif{}{N}{(P; (\swhile{M}{P}))}}.$$
        \item\label{loop-law-elimination} (Loop Elimination)
        $$\frac{M^\bot \propto N}{\sif{}{N}{(\swhile{M}{P})}\equiv [N]}.$$
        \item\label{loop-law-postcondition} (Postcondition)
        $$\frac{N\gg M}{\swhile{M}{P}\equiv (\swhile{M}{P}); [N]}.$$
    \end{enumerate}
\end{prop}

Proposition~\ref{loop-law}(\ref{loop-law-unfolding}) indicates that if a measurement $N$ that is logically stronger than $M$ is performed initially, then quantum loop $\swhile{M}{P}$ executes the loop body $P$, and then proceeds like $\swhile{M}{P}$ itself, and Proposition~\ref{loop-law}(\ref{loop-law-elimination}) says that if a measurement $N$ that is logically stronger than $M^\bot$ is performed initially, then loop $\swhile{M}{P}$ terminates immediately.
Proposition~\ref{loop-law}(\ref{loop-law-postcondition}) shows that at its termination, loop $\swhile{M}{P}$ establishes a postcondition $N$ which logically contradicts its guard that the outcome of measurement $M$ is $1$.

Recall the notions of logical weakness and contradiction defined in Definition~\ref{defn-measurement-relation}.
Proposition~\ref{loop-law} shows that they are also useful for simplifying loops.

\subsection{Tail Recursion}

Tail recursion is a scheme of recursion in which the recursive call is the last statement to be executed by the program.
It is particularly useful in practice because it is often easy to implement and optimise in compilation.
It is well-known in classical programming that tail recursion can always be realised as a loop.
This result extends to quantum programming as shown in the following:

\begin{thm}[Tail recursion]
\label{thm-tail}
    For any quantum programs $P,Q$ (without recursive calls), it holds that $\mu X.(\sif{Q}{M}{(P; X)})\equiv (\swhile{M}{P}); Q.$
\end{thm}
\begin{proof}
    We write $F(X)\triangleq \sif{Q}{M}{(P; X)}$.
    Let measurement $M=\{M_0,M_1\}$.
    Define quantum operations $\E_i$ by $\E_i(\rho)=M_i\rho M_i^\dag$ for all density operators $\rho$ $(i=0,1)$.
    Then we assert
    \begin{equation}
    \label{tail-proof}
        \sem{F^{(n)}}(\rho)=\sem{Q}\left(\sum_{k=0}^{n-1}\left(\E_0\circ(\sem{P}\circ\E_1)^k\right)(\rho)\right),
    \end{equation}
    for $n\geq 0$.
    Indeed, (\ref{tail-proof}) can be proved by induction on $n$.
    The case of $n=0$ is trivial.
    Now assume (\ref{tail-proof}) is true for $n$.
    Then, for the case of $n+1$, it holds that
    \begin{align*}
        \sem{F^{(n+1)}}(\rho) 
        &= \sem{Q}(\E_0(\rho))+\sem{P; F^{(n)}}(\E_1(\rho))\\
        &=\sem{Q}(\E_0(\rho))+\sem{F^{(n)}}\left((\sem{P}\circ\E_1)(\rho)\right)\\
        &=\sem{Q}(\E_0(\rho))+\sem{Q}\left(\sum_{k=0}^{n-1}\left(\E_0\circ(\sem{P} \circ \E_1)^k\right)\left((\sem{P}\circ\E_1)(\rho)\right)\right)\\
        &=\sem{Q}\left(\E_0(\rho)+\sum_{k=0}^{n-1}\left(\E_0\circ(\sem{P} \circ \E_1)^{k+1}\right)(\rho)\right) \\
        &=\sem{Q}\left(\sum_{k=0}^{n}\left(\E_0\circ(\sem{P} \circ \E_1)^k\right)(\rho)\right).
    \end{align*}
    Here, the third equality comes from the induction hypothesis on $n$.
    Thus, Eq. (\ref{tail-proof}) holds for $n+1$, and we complete the proof of Eq. (\ref{tail-proof}).
    Consequently, using Propositions~\ref{sempro}(\ref{sempro-while}) and~\ref{rec-sem} we obtain \begin{align*}
        \sem{\mu X.(\sif{Q}{M}{(P; X)})}(\rho)
        &=\bigsqcup_{n=0}^\infty\sem{Q}\left(\sum_{k=0}^{n}\left(\E_0\circ(\sem{P} \circ \E_1)^k\right)(\rho)\right)\\
        &= \sem{Q}\left(\bigsqcup_{n=0}^\infty\sum_{k=0}^{n}\left(\E_0\circ(\sem{P} \circ \E_1)^k\right)(\rho)\right)\\
        &=\sem{Q}\left(\sum_{k=0}^\infty\left(\E_0\circ(\sem{P} \circ \E_1)^k\right)(\rho)\right)\\
        &=\sem{(\swhile{M}{P}); Q}(\rho).
    \end{align*}
    Here, the first and fourth equalities come from Propositions~\ref{rec-sem} and~\ref{sempro}, respectively, and the second from the linearity and continuity of $\sem{Q}$, which were proved in~\cite{Ying16}. 
\end{proof}

\section{Laws for Nondeterminism}\label{sec-nondet}

All quantum programs considered in the previous sections are deterministic.
The program construct of nondeterministic choice was introduced by Dijkstra~\cite{Dij75, Dij76} to support the notion of refinement in program development methodology, in which an implemented program is required to be more deterministic than its specification. 
Intuitively, nondeterministic choice $P\sqcup Q$ is executed by executing either $P$ or $Q$, with the choice between them being arbitrary.
To facilitate refinement technique in quantum programming, nondeterministic choice $P\sqcup Q$ was generalised in~\cite{Zu04} and more recently in~\cite{Feng23} to quantum programs $P, Q$.
In this section, we derive some algebraic laws for nondeterministic quantum programs. 

The syntax of nondeterministic quantum programs is given by adding nondeterministic choice $P\sqcup Q$ into Eq. (\ref{prog-syntax}).
Recall from the previous sections that the denotational semantics of a deterministic quantum program $P$ (without nondeterministic choice) is a quantum operation (i.e., trace-nonincreasing super-operators) $\sem{P}:\mathcal{D}(\hs_\QVar)\rightarrow \mathcal{D}(\hs_\QVar)$.
To accommodate nondeterminism, however, the denotational semantics of a quantum program $P$ that may contain nondeterministic choice should be defined as a subset $\sem{P}\subseteq \mathcal{D}(\hs_\QVar)\rightarrow \mathcal{D}(\hs_\QVar)$ of quantum operations.
Intuitively, for each input state $\rho$, $\sem{P}(\rho)=\{\E(\rho) \mid \E\in\sem{P}\}$ is the set of possible output states of $P$.
As shown in~\cite{Feng23}, the structural representation of the denotational semantics of deterministic quantum programs (see Proposition~\ref{sempro}) can be easily lifted to this case: 
\begin{enumerate}
    \item $\sem{\iskip}=\{I\}$ and $\sem{\iabort}=\{0\}$, where $I$ and $0$ are the identity and zero operators, respectively, on $\hs_\QVar$.  
    \item $\sem{ \iinit{q}{|\psi\>}}=\left\{\E_{|\psi\>}\right\}$, where quantum operation $\E_{|\psi\>}$ is defined by $\E_{|\psi\>}(\rho) = \sum_n |\psi\>\< n|\rho|n\>\<\psi|$ for any $\rho$, and $\{|n\>\}$ is a given orthonormal basis of $\hs_\QVar$. 
    \item $\sem{U[\qbar]}=\left\{\E_U\right\}$, where quantum operation $\E_U$ is defined by $\E_U(\rho)=U\rho U^\dag$ for any $\rho$. 
    \item $\sem{P; Q}=\{\mathcal{F}\circ\mathcal{E} \mid \E\in\sem{P}\ {\rm and}\ \mathcal{F}\in \sem{Q}\}.$
    \item $\sem{\sif{P}{M}{Q}}=\left\{\E\circ\E_0+\mathcal{F}\circ\E_1 \mid \E\in\sem{P}\ {\rm and}\ \mathcal{F}\in\sem{Q}\right\},$ where $M=\{M_0,M_1\}$, and quantum operations $\E_0,\E_1$ are defined by $\E_i(\rho)=M_i\rho M_i^\dag$ for $i=0,1$. 
    \item $\sem{\swhile{M}{P}}=\left\{\sum_{i=0}^\infty\E_0\circ(\mathcal{F}_i\circ\E_1)\circ\cdots\circ(\mathcal{F}_1\circ\E_1) \mid \mathcal{F}_1, \ldots, \mathcal{F}_i, \ldots\in\sem{P}\right\}$, where $\E_0,\E_1$ are the same as in (5). 
\end{enumerate}
The denotational semantics of nondeterministic choice $P\sqcup Q$ of quantum programs is then defined in~\cite{Feng23} as 
\begin{equation}
    \sem{P\sqcup Q} =\sem{P}\cup \sem{Q}.
\end{equation}
 
Now we are ready to present our laws for nondeterministic quantum programs.
Some of them were already applied in the example of quantum error correction code transformation in Subsection~\ref{illu-exam}.
We first point out that all laws presented in the previous sections, except those for recursion, can be easily generalised to nondeterministic quantum programs; but some of the laws need certain further assumptions on nondeterminism.
More precisely, we need to assume that the following programs do not contain non-deterministic choice $\sqcup$: $P$ in Proposition~\ref{if-laws}(\ref{if-laws-idempotence}), $P,Q$ in Proposition~\ref{nested-if-law}(\ref{nested-if-law-reduction}), $R$ in Proposition~\ref{nested-if-law}(\ref{nested-if-law-left-distributivity}--\ref{nested-if-law-right-distributivity-projection}), $P$ in Proposition~\ref{sequential-laws}(\ref{sequential-laws-right-distributivity}--\ref{sequential-laws-left-distributivity1}) and $C$ in Proposition~\ref{sequential-laws}(\ref{sequential-laws-left-distributivity2}).
The next proposition introduces some new laws about non-deterministic choice:

\begin{prop}[Laws for Nondeterminism]
\label{laws-nd}
    \quad
    \begin{enumerate}
        \item\label{laws-nd-commutativity} (Commutativity)
        $P\sqcup Q\equiv Q\sqcup P$.
        \item\label{laws-nd-associativity} (Associativity)
        $P\sqcup(Q\sqcup R)\equiv (P\sqcup Q)\sqcup R$.
        \item\label{laws-nd-idempotence} (Idempotence)
        $P\sqcup P\equiv P$.
        \item\label{laws-nd-distributivity-if} (Distributivity I)
        If-conditional distributes over nondeterministic choice: 
        \begin{align*}
            &\sif{(P\sqcup Q)}{M}{R}= (\sif{P}{M}{R})\sqcup (\sif{Q}{M}{R}),\\
            &\sif{R}{M}{(P\sqcup Q)}= (\sif{R}{M}{P})\sqcup (\sif{R}{M}{Q}).
        \end{align*}
        \item\label{laws-nd-distributivity-seq} (Distributivity II)
        Sequential composition distributes over nondeterministic choice:
        \begin{align*}
            P; (Q\sqcup R)\equiv (P; Q) \sqcup (P; R),\qquad (Q\sqcup R); P\equiv (Q; P)\sqcup (R; P).
        \end{align*}
    \end{enumerate}
\end{prop}

These laws in the above proposition were already used in Subsection~\ref{illu-exam} for verification of the correctness of a simple quantum error correction code.
Clearly, Clauses (\ref{laws-nd-commutativity}), (\ref{laws-nd-associativity}), (\ref{laws-nd-distributivity-if}) and (\ref{laws-nd-distributivity-seq}) in this proposition are quantum generalisations of classical laws (Nd-1), (Nd-2), (Nd-3,4) and (Nd-6,7) in Figure~\ref{fig 1}, respectively.
However, it is worth pointing out that the quantum counterpart of law (Nd-5) for classical nondeterministic choice does not hold.
For example, if $P\equiv Q\equiv \iabort$, $M=\{M_0,M_1\}$ and $R$ is a deterministic quantum program (without choice $\sqcup$, then $\llbracket (\sif{P}{M}{Q})\sqcup R\rrbracket =\{0,\llbracket R \rrbracket\}\neq \{0,\llbracket R \rrbracket\circ\E_0, \llbracket R \rrbracket\circ \E_1, \llbracket R \rrbracket\circ (\E_0+\E_1)\}=\llbracket \sif{(P\sqcup R)}{M}{(Q\sqcup R)}\rrbracket$, where for $i=0,1$, quantum operations $\E_i$ is defined by $\E_i(\rho)=M_i\rho M_i^\dag$ for any $\rho$.   

As in classical programming, we are able to prove a series of laws for the refinement of quantum programs.

\section{Laws for Refinement}\label{Sec-Refine}

Refinement is one of the central ideas in program development methodology~\cite{Morgan98}.
As mentioned in Section~\ref{sec-nondet}, it has recently been introduced into quantum programming~\cite{Ped23, Feng23a}. 
Formally, let $P$ and $Q$ be two (nondeterministic) quantum programs.
Then the refinement relation between them is defined as follows: 
\[
    P\sqsubseteq Q\ {\rm if\ and\ only\ if}\ \sem{P}\subseteq\ \overline{\mathit{Con}(\sem{Q})},
\]
where $\mathit{Con}(\cdot)$ stands for the convex hull, and $\overline{\ \cdot\ }$ denotes topological closure.  
Intuitively, $P\sqsubseteq Q$ means that $P$ is a more deterministic program than $Q$.
The statement that $P$ is a refinement of $Q$, written $P\sqsubseteq Q$, should be understood in the following way~\cite{Hoare87}: whenever $Q$ \textit{reliably} serves some useful purpose, then $Q$ can be replaced by $P$, in the certainty that it will serve the same purpose, but not vice versa; that is, there may be some purposes for which $P$ is adequate, but $Q$ cannot be relied on.
Therefore, $P\sqsubseteq Q$ conveys the idea that program $P$ is better than, or at least as good as, program $Q$.  

Obviously, $\sqsubseteq$ is a partial order, $P\sqsubseteq P\sqcup Q$ and $Q\sqsubseteq P\sqcup Q$.
Furthermore, we can prove that the refinement relation is preserved by all quantum program constructs considered in this paper except loops and recursion:
\begin{prop}\label{refine-law1}
If $P\sqsubseteq Q$, then:
    \begin{enumerate}
        \item\label{refine-law1-seq} $P; R\sqsubseteq Q; R$ and $R; P\sqsubseteq R; Q$.
        \item\label{refine-law1-cond} $\sif{P}{M}{R}\sqsubseteq \sif{Q}{M}{R}$ and $\sif{R}{M}{P}\sqsubseteq \sif{R}{M}{Q}$. 
        \item\label{refine-law1-nchoice} $P\sqcup R\sqsubseteq Q\sqcup R$ and $R\sqcup P\sqsubseteq R\sqcup Q.$
    \end{enumerate}
\end{prop}

A discussion of preservation of refinement by loop is given in Appendix~\ref{ref-loop}.  

Now, let us further consider two different refinements of nondeterministic choice.
According to different strategies in the resolution of nondeterminism, we can define:
\begin{enumerate}
    \item \textit{Probabilistic implementation}:
    \begin{equation*}
        \sem{P\sqcup_p Q} =\left\{p\E+(1-p)\mathcal{F} \mid \E\in \sem{P},\mathcal{F}\in \sem{Q}\ {\rm and}\ p\in[0,1]\right\}.
    \end{equation*}
    \item \textit{Measurement-based implementation}:
    Suppose we have an extra qubit $q$ in a pure state $|\phi\>$.
    It is possible to refine $P\sqcup Q$ by: (i) measuring $q$ according to an arbitrary measurement $M$; (ii) for each branch, recovering the $q$ by initialising $\iinit{q}{|\phi\>}$ and then executing $P$ or $Q$, respectively.
\end{enumerate}
In $P\sqcup_p Q$, $P$ and $Q$ are selected to execute according to a probability distribution $p$, $1-p$. 
In contrast,  in the measurement-based implementation, the computational outcome can be seen as obtained by a post-selection according to a quantum measurement $M=\{M_0,M_1\}$. 
The following proposition clarifies the refinement relation between nondeterministic choice and its probabilistic and measurement-based resolutions:
\begin{prop}\label{refine-law2}
\begin{enumerate}
    \item\label{refine-law2-npchoice}
    $P\sqcup_p Q\sqsubseteq P\sqcup Q$. 
    \item\label{refine-law2-pchoice}
    If $P\sqsubseteq Q$ then $P\sqcup_p R\sqsubseteq Q\sqcup_p R$ and $R\sqcup_p P\sqsubseteq R\sqcup_p Q$.
    \item\label{refine-law2-meas}
    $[M]; P \sqsubseteq \sif{P}{M}{P}$. 
    \item\label{refine-law2-init-meas}
    $\iinit{q}{|\phi\>}; (\sif{(\iinit{q}{|\phi\>}; P)}{M[q]}{(\iinit{q}{|\phi\>}; Q)}) \sqsubseteq \iinit{q}{|\phi\>}; (P \sqcup Q)$. 
\end{enumerate}
\end{prop}

\section{An Application: Formal Derivation of the Principle of Deferred Measurements}\label{sec-app}

A basic principle of great utility in quantum computing is:
\begin{itemize}
    \item \textbf{Principle of deferred measurements}~\cite{NC00}: \textit{Measurements can always be moved from an intermediate stage of a quantum circuit to the end of the circuit; if the measurement results are used at any stage of the circuit, then the classically controlled operations can be replaced by conditional quantum operations.} 
\end{itemize} 
Quantum circuits with measurements in the middle, whose results are used to control later computations, are essentially finite programs, as we called them in Section~\ref{sec-finite} of this paper. They are also often called dynamic quantum circuits~\cite{Dynamic}. 
 
The above principle can be formally formulated in a programming language with both $\kif$-statements and quantum $\kif$-statements:
\begin{thm}
\label{thm-defer}
    If plenty of fresh auxiliary variables are assumed, then for any finite quantum program (i.e., a program that does not contain iteration or recursion) $P\in\mathbf{QProg}$, there exists a quantum circuit $C\in\mathbf{QC}$ such that:
    \begin{equation}
    \label{defer-2}
        \iinit{\qbar_a}{}; {\color{blue}P} \equiv \iinit{\qbar_a}{}; {\color{blue}C; \iif{}{\qbar_a}{m}{P_m}}; \iinit{\qbar_a}{}.
    \end{equation}
    where $P_m\in \{\iskip,\iabort\}$, $\qbar_a$ is the auxiliary register that initialised and reset to $|0\>$ through the program. 
    In particular, quantum circuit $C$ can be recursively generated using quantum $\kif$-statements according to the structure of $P$.
\end{thm}
Informally, the above theorem suggests that $P\equiv C; \iif{}{\qbar_a}{m}{P_m}$ if we ignore auxiliary variables, that is, a finite quantum program $P$ can be equivalently modelled as a circuit followed by a computational measurement.
Here, since classical variables are not included in our syntax, we implicitly discard any intermediate measurement outcome; thus, the irreversibility of $\iinit{\qbar_a}{}$ in the RHS of (\ref{defer-2}) might be interpreted as the irreversibility of intermediate measurement in $P$. On the other hand, if we are allowed to gather measurement outcomes, the measurement outcome of $\qbar_a$ is exactly the collection of all measurement outcomes in $P$. 

As an application of the laws we developed in the previous sections, we show how Theorem~\ref{thm-defer} can be formally derived using them. 

\begin{lem}[Implementation of $\kif$-statement with auxiliary variable, c.f.~\cite{NC00}]
\label{lem:implement-if-statement}
    For any if statement $\iif{M}{\qbar}{m}{P_m}$ where $M = \{M_i\}_{i\in I}$, let $q_a$ be an arbitrary fresh auxiliary variable with Hilbert space of dimension $|I|$ and initialised in basis state $|0\>$. 
    Then there exists a unitary operator $U_M$, which is regarded as necessary information for implementing the measurement, such that:
    \[
        \iinit{q_a}{}; {\color{blue}\iif{M}{\qbar}{m}{P_m}}\equiv 
        \iinit{q_a}{}; {\color{blue}U_M[\qbar,q_a]; \iif{}{q_a}{m}{P_m}}; \iinit{q_a}{}.
    \]
\end{lem}

\begin{proof}[Proof of Theorem~\ref{thm-defer}]
    We proceed by induction on the structure of $P$.
    Here, we consider only the case $P = \iif{M}{\qbar}{i}{P_i}$; the proofs for the other cases are similar and postponed to Appendix~\ref{proof-app}.
    By the induction hypothesis, we assume that each $P_i$ has the following normal form such that $\qvar(C_i)\subseteq\qvar(P_i)\cup \rbar_i$ and
    \begin{align}
    \label{eqn:thm8.1_hyp}
        & \iinit{\rbar_{i}}{}; \bt{P_i} \equiv \iinit{\rbar_{i}}{}; \bt{C_i; \iif{}{\rbar_{i}}{j}{P_{ij}}}; \iinit{\rbar_{i}}{}.
    \end{align}
    Let $a$ be a fresh variable that is not in  all $\rbar_{i}$, and let $\qbar_a\triangleq a\cup\bigcup_{i}\rbar_{i}$.
    Then we claim:
    \begin{align*}
        \iinit{\qbar_{a}}{}; \bt{\iif{M}{\qbar}{i}{P_i}}
        \equiv\ 
        & \iinit{\qbar_{a}}{}; \bt{(U_M[\qbar, a] ; \iqif{a}{\square |i\>}{C_i}); } \\   
        & \bt{\iif{}{\qbar_{a}}{k}{P_{k{\downarrow}_{a} k{\downarrow}_{\rbar_{i}}}}}; \iinit{\qbar_{a}}{},
    \end{align*}
    for some unitary $U_M$.
    For simplicity, we set $\rbar\triangleq \bigcup_i\rbar_i$ and $\rbar_i^\prime\triangleq \rbar \backslash \rbar_i$.
    Then
    \begin{align*}
        & \bt{ \iinit{\qbar_{a}}{}}; \iif{M}{\qbar}{i}{P_i} 
        & \\
        \equiv\ & \iinit{\qbar_{a}}{}; \iinit{\rbar}{}; \bt{ \iinit{a}{}; \iif{M}{\qbar}{i}{P_i}} 
        & \reason{Prop.~\ref{init-laws}(\ref{init-laws-cancellation})}\\
        \equiv\ & \rt{ \iinit{\qbar_{a}}{}}; \bt{ \iinit{\rbar}{}}; \rt{ \iinit{a}{}}; U_M[\qbar,a]; \bt{\iif{}{a}{i}{P_i}}; \iinit{a}{} \hspace{-0.3cm}
        & \reason{Lemma~\ref{lem:implement-if-statement}}\\
        \equiv\ & \iinit{\qbar_{a}}{}; U_M[\qbar,a]; \iif{}{a}{i}{\bt{ \iinit{\rbar}{}}; P_i}; \iinit{a}{} 
        & \reason{Props.~\ref{init-laws}(\ref{init-laws-cancellation}),~\ref{sequential-laws}(\ref{sequential-laws-commutativity},~\ref{sequential-laws-left-distributivity1})}\\
        \equiv\ & \iinit{\qbar_{a}}{}; U_M[\qbar,a]; \iif{}{a}{i}{ \iinit{\rbar}{}; \iinit{\rbar_i}{}; \bt{ \iinit{\rbar_i^\prime}{}; P_i}}; \iinit{a}{} \hspace{-2cm}
        & \reason{Prop.~\ref{init-laws}(\ref{init-laws-cancellation})}\\
        \equiv\ & \iinit{\qbar_{a}}{}; U_M[\qbar,a]; \iif{}{a}{i}{ \iinit{\rbar}{}; \bt{ \iinit{\rbar_i}{}; P_i}; \iinit{\rbar_i^\prime}{}}; \iinit{a}{}  \hspace{-2cm}
        & \reason{Prop.~\ref{sequential-laws}(\ref{sequential-laws-commutativity})}\\
        \equiv\ & \iinit{\qbar_{a}}{}; U_M[\qbar,a]; \mathbf{if}\,[a]\left(\square i\rightarrow \bt{ \iinit{\rbar}{}; \iinit{\rbar_{i}}{}}; \vphantom{P_{ij}}\right. & \\ 
            &\!\!\left.C_i; \iif{}{\rbar_i}{j}{P_{ij}}; \bt{ \iinit{\rbar_{i}}{}; \iinit{\rbar_i^\prime}{}}\right) \mathbf{fi}; \iinit{a}{} 
        & \reason{Ind. hyp. Eq. (\ref{eqn:thm8.1_hyp})}\\
        \equiv\ & \iinit{\qbar_{a}}{}; U_M[\qbar,a]; \mathbf{if}\,[a]\left(\square i\rightarrow \bt{ \iinit{\rbar}{}}; C_i; \vphantom{P_{ij}}\right. & \\
            &\!\!\left.\rt{\iif{}{\rbar_i}{j}{P_{ij}}; \iinit{\rbar}{}}\right) \mathbf{fi}; \iinit{a}{} 
        & \reason{Prop.~\ref{init-laws}(\ref{init-laws-cancellation})}\\
        \equiv\ & \rt{ \iinit{\qbar_{a}}{}}; U_M[\qbar,a]; \rt{ \iinit{\rbar}{}}; \mathbf{if}\,[a]\left(\square i\rightarrow \vphantom{P_{ij{\downarrow}_{\rbar_i}}} \right. & \\
            &\!\!\left.\bt{C_i; \iif{}{\rbar}{j}{P_{ij{\downarrow}_{\rbar_i}}}; \iinit{\rbar}{}}\right) \mathbf{fi}; \iinit{a}{} 
        & \reason{Props.~\ref{init-laws}(\ref{init-laws-if-expansion}),~\ref{sequential-laws}(\ref{sequential-laws-commutativity},~\ref{sequential-laws-left-distributivity1})} \\
        \equiv\ & \iinit{\qbar_{a}}{}; U_M[\qbar,a]; \bt{\iif{}{a}{i}{C_i; \iskip}}; & \\
            & \iif{}{a}{i}{\iif{}{\rbar}{j}{P_{ij{\downarrow}_{\rbar_i}}}; \rt{ \iinit{\rbar}{}}}; \iinit{a}{} 
        & \reason{Props.~\ref{init-laws}(\ref{init-laws-cancellation}),~\ref{sequential-laws}(\ref{sequential-laws-unit-zero},~\ref{sequential-laws-commutativity},~\ref{sequential-laws-sequentiality})}\\
        \equiv\ & \iinit{\qbar_{a}}{}; U_M[\qbar,a]; \iqif{a}{\square |i\>}{C_i}; \bt{\iif{}{a}{i}{\iskip}; } \hspace{-0.5cm}& \\
        & \bt{\iif{}{a}{i}{\iif{}{\rbar}{j}{P_{ij{\downarrow}_{\rbar_i}}}}}; \rt{ \iinit{\rbar}{}; \iinit{a}{}} 
        & \reason{Props.~\ref{sequential-laws}(\ref{sequential-laws-right-distributivity}),~\ref{prop:interplay-qif-if}}\\
        \equiv\ & \iinit{\qbar_{a}}{}; U_M[\qbar,a]; \iqif{a}{\square |i\>}{C_i}; & \\ 
        & \bt{\iif{}{a}{i}{\iif{}{\rbar}{j}{P_{ij{\downarrow}_{\rbar_i}}}}}; \iinit{\qbar_a}{} 
        & \reason{Props.~\ref{init-laws}(\ref{init-laws-cancellation}),~\ref{sequential-laws}(\ref{sequential-laws-unit-zero},~\ref{sequential-laws-sequentiality})}\\
        \equiv\ & \iinit{\qbar_{a}}{}; (U_M[\qbar,a]; \iqif{a}{\square |i\>}{C_i}); & \\ 
        & \iif{}{\qbar_a}{k}{P_{k{\downarrow}_a k{\downarrow}_{\rbar_i}}}; \iinit{\qbar_a}{}.
        & \reason{Prop.~\ref{nested-if-law}(\ref{nested-if-law-reduction})}
    \end{align*}
    It is particularly interesting to note that the law in Proposition~\ref{prop:interplay-qif-if} is used in the 10th equivalence of the above equation. 
\end{proof}

\section{Formalisation in the Coq Proof Assistant}
\label{sec-coq}

We formalised the theory in which our laws are established and verified all of the laws presented in previous sections in the Coq proof assistant~\cite{coq}.
We achieved this by employing CoqQ~\cite{Zhou23}, a general-purpose framework for quantum program verification built upon the state-of-the-art mathematical libraries MathComp~\cite{Mathcomp} and MathComp-Analysis~\cite{mathcomp-analysis}. 
The development can be found at \url{https://github.com/coq-quantum/CoqQ/tree/main/src/example/qlaws}.

There are several mechanised approaches to quantum circuits and quantum programs; see~\cite{CVLreview, LSZreview} for a comprehensive review. 
Among them, targeting the verification of high-level quantum programs, CoqQ offers an extensive library on abstract linear algebra and analysis, e.g., existing theorems on convergence and fixed points are crucial for formalising our semantics of while and recursive programs. Recent development~\cite{Feng23a} enhances CoqQ by a formalised mechanism for quantum registers with automatic type checking, greatly facilitating the mature use of quantum variables.
Other tools, such as \qwire~\cite{Rand17, RPZ17} and SQIR~\cite{Hiet21,hietala2020proving}, have their respective focuses, as we will discuss in Section~\ref{sec-related}.

\subsection{Formalizing the theory}\label{sec-formal-th}
The existing qwhile language in CoqQ is not sufficient to describe and prove the laws proposed in this paper. The further developments of formalising the theory in this work include:
\begin{itemize}
    \item We implement the two layers of quantum circuits and purely quantum programs in our three-layer quantum programming framework (see Figure~\ref{fig 0}) from the bottom up. In particular, the purely quantum program layer is divided into quantum $\kwhile$-programs, quantum programs with recursion, and quantum programs with nondeterminism, so that we can prove the laws in different scenarios with slightly adjusted premises.
    As far as we know, this is the first time formalising the semantics of recursive quantum programs, in particular \textit{with mutual recursive calls}.
    \item We provide automatic checks on the sets of variables, such as the well-formedness and disjointness of the programs, to facilitate the use of the laws.
    The disjointness condition asserts that a quantum register should consist of a set of distinct variables, which is required due to the no-cloning theorem.
    Disjointness is further employed in (1) the well-formedness of a quantum $\kif$-statement and (2) side-conditions of multiple laws.
    To this end, we extend the register mechanism in~\cite{Feng23a}, making Coq automatically infer the disjointness conditions.
    \item \textit{Convex hull} is employed as the basis of developing the refinement laws of nondeterministic quantum programs.
    As a foundation for proving these laws, we implemented the theory of convex hulls with related basic properties of addition $+$ (i.e., $A+B = \{a+b\mid a\in A\ \&\ b\in B\}$) and composition $\circ$ (i.e., $A\circ B = \{a\circ b\mid a\in A\ \&\ b\in B\}$),
    such as $\mathit{Conv}(A+B) = \mathit{Conv} (A) + \mathit{Conv} (B)$ and $\mathit{Conv}(A \circ B) = \mathit{Conv}(\mathit{Conv} (A)\circ \mathit{Conv} (B))$.
\end{itemize}

\subsection{Verification of the laws}

Based on the developments described in the above subsection, we proved all of the laws of quantum programming introduced in this paper in Coq. More precisely: 
\begin{itemize} 
    \item The basic laws for quantum circuits in Section~\ref{sec-laws-circ} can be easily proved.
    More interestingly, we prove the existence of normal forms for quantum circuits in Coq.
    In detail, we first give a specific construction of normal form, which is defined inductively on the program syntax, and then prove that the construction is semantically equivalent to the original program.
    \item The laws for purely quantum programs are proved sound with respect to the semantics formalised in Subsection~\ref{sec-formal-th}. 
    Some laws are proved in a more general setting; for example, the rules for recursive programs hold even if more than one recursive functions are involved.
\end{itemize}

\subsection{Applications}
    Furthermore, the formalisation and proof of our laws of quantum programming, described in the previous two subsections, support formal verification of more complicated applications in Coq.
    As a simple example showing this prospect, we provide the proof of Example~\ref{illu-exam} in Coq, which can be found in the development.

\section{Related Work}\label{sec-related}

\subsection{Quantum Programming Framework}
We described a three-layer framework of quantum programming in Section~\ref{sec-framework} so that the laws of quantum programming developed in this paper can be grouped properly.
As pointed out in the Introduction, all ingredients of this framework are not new. At the first layer is a quantum circuit description language.
A basic difference between it and other quantum circuit languages like QWIRE~\cite{Rand17} is that the construct of quantum $\kif$-statement is included in the former (see Clause (4) in Definition~\ref{def-circ}) but not in the latter.
But quantum $\kif$-statement has been extensively studied in quantum programming literature~\cite{Alt05, Ying12, Yuan,  Bich, Voi23, Yuan24}. 
At the second layer of purely quantum programs, a quantum extension of classical while-language is employed, gradually expanded by adding recursion and nondeterministic choice.
The quantum $\kwhile$-language is standard, as exposed in the book~\cite{Ying16}.
The recursion considered in this paper is defined with classical control flow rather than quantum control flow.
This kind of recursion was carefully studied in Section 3.4 of~\cite{Ying16} and~\cite{Xu-Ying21}.
Nondeterministic quantum programs were first introduced in~\cite{Zu04}, and program logic for them was recently developed in~\cite{Feng23}.
The third layer of our framework is quantum programs embedded in a classical programming language, which is a common practice in quantum programming.
Indeed, all of these program constructs have already been widely adopted in mainstream quantum programming languages and platforms~\cite{Qiskit, Cirq, ProjectQ, Omer05, LanQ, GKMW10, GP13, Qsharp, Quipper, Scaffold, TKet, isQ}. Therefore, one can expect that the laws established in this paper can be applied to the compilation, transformation, and optimisation of quantum programs in these languages.
However, determining how to practically apply these laws in each case remains an important direction for future research.

\subsection{Algebraic reasoning about quantum programs}
Our laws of quantum programming, established within the framework discussed above, are essentially quantum generalisations of Hoare et al.'s algebraic laws of classical programming~\cite{Hoare87, Hoare88, Hoare93}. 

Several other approaches to algebraic reasoning about quantum programs have been proposed in the literature.
The ZX-calculus is a graphical language for reasoning about linear maps between qubits~\cite{ZX-book, DKP20, coecke2011interacting,PyZX}.
A rich theory of equational reasoning and algebraic laws has been developed around it.
It has found successful applications across a wide range of areas in quantum computation and quantum information, particularly in quantum circuit verification, compilation, synthesis, and optimisation.
The relationship between the ZX-calculus and the laws presented in this paper remains to be fully understood; for example, it is an open question how our laws—especially those concerning loops and recursion—can be represented within the ZX-calculus framework.

An interesting approach to quantum programming was proposed in~\cite{GKMW10, GP13}, where linear algebraic tools are used to provide basic insights into the idealised process of unitary evolution.
Clearly, exploring how the laws developed in this paper can be combined with the ideas of~\cite{GKMW10, GP13} is an interesting and worthwhile topic for future research.

An algebraic theory of unitary gates and quantum measurements was proposed in~\cite{Staton15}, from which an equational theory for quantum programs was extracted. 
A key difference between~\cite{Staton15} and this paper is that the former emphasises high-level theories, whereas the latter aims to derive concrete laws that can be directly applied to quantum program development.

The successful applications of KAT (Kleene Algebras with Tests)~\cite{Kozen} in reasoning about classical programs were extended to quantum programs and networks~\cite{Peng22, QKAT24}.
However, most of the laws of quantum programming derived in this paper were not obtained in~\cite{Staton15, Peng22}. 
A normal form of quantum $\kwhile$-programs consisting of only a single loop was proved in~\cite{Peng22, Yu23}.
However, it is different from our normal form (Theorem~\ref{thm-normal-prog}) because the latter is about finite quantum programs without loops.
Indeed, they can be combined to transform an arbitrary quantum $\kwhile$-program $P$ to a stronger normal form $P\equiv\iwhile{M}{\qbar}{Q}$ with a single loop and a flat $\kif$-statement $Q$ in the form of (\ref{normal-prog}).   

\subsection{Rewrite Rules used in Compilation and  Optimisation of Quantum Programs}
In the last few years, several quantum-circuit optimisers have been implemented using rewrite rules, including VOQC~\cite{Hiet21}, Quartz~\cite{Xu22}, and QUESO~\cite{Xu23}.
Some rewrite rules were also introduced in Giallar---a verification tool for the Qiskit compiler~\cite{Tao22}. 
The rewrite rules employed in these works are essentially algebraic laws for quantum circuits, but they are quite different from our laws for quantum circuits given in Section~\ref{sec-laws-circ} because the majority of our laws concern the quantum $\kif$-statement, which is not defined in the syntax of quantum circuits considered in~\cite{Hiet21, Xu22, Xu23, Tao22}.
Indeed, the laws given in~\cite{Hiet21, Tao22, Xu22, Xu23} and our laws presented in Section~\ref{sec-laws-circ} can be used complementarily for the optimisation of quantum circuits.
They can be further combined with our laws presented in Section~\ref{sec-laws-prog} for optimisation of quantum programs. 

The normal form approach to compiler design in~\cite{Hoare93} was already generalised in~\cite{Zu05} to the compilation of quantum programming language qGCL (quantum Guarded-Command Language)~\cite{Zu00}, which extends pGCL (probabilistic GCL)~\cite{He97} with constructs of initialisation, unitary transformation, and measurement.
The compiler described in~\cite{Zu05} mainly deals with transformations at the classical-quantum hybrid layer, with the layers of quantum circuits and purely quantum programs untouched.
The normal form employed in~\cite{Zu05} is inherited directly from that in~\cite{Hoare93, He97} for GCL via pGCL, and thus is fundamentally different from the normal forms presented in this paper (Theorems~\ref{normal-thm-cir} and~\ref{thm-normal-prog}).
Our normal forms can help to further compile the compiled programs in~\cite{Zu05} into a lower level. 

\subsection{Mechanised Approach for Quantum Programming}
There is an active line of work that develops formal methods and, more specifically, machine-checkable verification of quantum programs (see recent reviews~\cite{CVLreview, LSZreview}).
Here, we only briefly discuss those that are mostly related to our work: 

\textit{\textbf{Coq-based formalisation}}. 
\qwire~\cite{Rand17,RPZ17} formalises the circuit-like programming language and provides a denotational semantics of quantum circuits in terms of density matrices.
Its extension \reqwire~\cite{rand2019reqwire} is a verified compiler from classical functions to quantum circuits.
SQIR~\cite{Hiet21} is a low-level language for intermediate representation.
Its semantics is based on a density matrix representation of quantum states, and has been successfully applied to developing a verified optimiser of quantum circuits~\cite{Hiet21}, as well as semantics-based verification of quantum programs~\cite{hietala2020proving, PHT23}.
VyZX~\cite{VyZX22,VyZX23,shah2024vicarvisualizingcategoriesautomated} is a certified (in Coq) formalisation of the ZX-calculus~\cite{coecke2011interacting,PyZX}, which is proposed for reasoning about quantum circuits in a flexible graphical structure.

\textit{\textbf{Isabelle/HOL-based formalisation}}.
QHLProver~\cite{LZW19} provides the first formalisation of quantum Hoare logic~\cite{Ying12} and is used to verify examples such as Grover's algorithm.
Isabelle Marries Dirac~\cite{BLH21,bordg2020isabelle} is an ongoing effort to formalise quantum information theory and to provide verified quantum algorithms; the formalisation is based on a matrix representation for quantum circuits.
qrhl-tool~\cite{Unr19, Unr20} is a verification tool for (post-)quantum cryptography based on quantum relational Hoare logic.

\textit{\textbf{Other tools}}.
Building upon Why3~\cite{bobot2011why3}, \qbricks~\cite{chareton2020deductive} provides a highly automated verification framework for circuit-building quantum programs based on path-sum representations of quantum states~\cite{amy2018towards}. EasyPQC~\cite{BBF21} is an extension of EasyCrypt~\cite{barthe2012easycrypt} that aims to verify the security of post-quantum cryptography based on post-quantum relational Hoare logic.
The proof assistant Quantomatic~\cite{KZ15} has been successfully used for diagrammatic reasoning in quantum information.
A particularly interesting problem for future research is how the laws developed in this paper can be implemented in Quantomatic so that they can be used more conveniently for reasoning about quantum programs (with $\kif$-statements and $\kwhile$-loops).  

\section{Conclusion}\label{sec-concl}

A series of basic laws for quantum programming has been developed in this paper, generalising Hoare et al.'s fundamental laws for classical programming.
These laws are formally verified in the Coq proof assistant, ensuring that they can be confidently applied in the compilation, transformation, optimisation, analysis, and verification of quantum programs.

For further development and application of this line of research, we propose the following topics:
\begin{itemize}
    \item \textbf{\textit{Application in Quantum Compilers}}:
    As briefly discussed earlier, the laws presented in this paper—particularly the two normal forms and the realisation of tail recursion by loops—can be incorporated into optimising compilers for quantum programs.
    However, it remains unclear whether these laws will be practically useful for quantum circuit and program optimisation.
    We plan to collaborate with the isQ team~\cite{isQ} to implement these laws in the isQ compiler and evaluate their effectiveness using benchmarks of quantum circuits and programs.

    \item \textbf{\textit{Laws for Recursion with Quantum Control Flows}}:
    Quantum $\kif$-statement is the basic program construct representing quantum control flows.
    Recently, some recursion schemes with quantum control flows have been introduced into quantum programming; for example, recursion in linear-algebraic (quantum) lambda-calculus~\cite{Ben22, ADV17, DGMV19}, quanta-morphism (a structural form of quantum recursion implementing cycles and folds on lists with quantum control flow)~\cite{IBM-Q}, quantum recursion in Fock spaces~\cite{Ying16}, and recursively defined quantum circuits~\cite{YingZ23, YingZ24, ZhangY25}.
    A set of algebraic laws for quantum $\kif$-statements was established at the quantum circuit layer in Section~\ref{sec-laws-circ}.
    On the other hand, we proved several laws for loop and recursive quantum programs with classical control flows; in particular, Theorem~\ref{thm-tail} for tail recursion.
    A natural problem is how to extend these laws to quantum recursive programs with quantum control flows.
    It seems that this problem cannot be solved simply by adopting the approach used in Section~\ref{sec-recur}.  

    \item \textbf{\textit{Automated Verified Verifier/Optimiser}}:
    Although we have formalised the proposed laws in Coq, the development currently only supports interactive proofs, which significantly limits its usefulness.
    Automation is particularly important, especially for verifying or optimising large-scale circuits and programs.
    Building rewriting systems is one of the common approaches for equivalence checking.
    Designing heuristics or AI-assisted algorithms to optimise the application of different laws is also a potential direction.
\end{itemize}

\section*{Acknowledgement}
This research was partly supported by the National Key R\&D Program of China under Grant \linebreak No. 2023YFA1009403.

\bibliography{main}

\newpage

\appendix

{\centering\LARGE\textbf{Appendices}}

\section{Proofs of the Laws}

For readability, we omitted the detailed proofs of the laws developed in this paper from the main text.
For the reader's convenience, we present them here.

\subsection{Proofs of the Laws in Section~\ref{sec-laws-circ}}

\begin{proof}[Proof of Proposition~\ref{qif-law}]
    (\ref{qif-law-changing-basis})
    By assumption, we obtain:
    \begin{align*}
        \sem{\sqif{C_0}{q}{C_1}{|\psi_0\>}{|\psi_1\>}}
        &= |\psi_0\>_q\<\psi_0|\otimes \sem{C_0} +|\psi_1\>_q\<\psi_1|\otimes \sem{C_1}\\
        &= U|\varphi_0\>_q\<\varphi_0|U^\dag \otimes \sem{C_0} +U|\varphi_1\>_q\<\varphi_1| U^\dag \otimes \sem{C_1}\\ 
        &= U(|\varphi_0\>_q\<\varphi_0| \otimes \sem{C_0} +|\varphi_1\>_q\<\varphi_1|  \otimes \sem{C_1})U^\dag\\ 
        &= \sem{RHS}.
    \end{align*}
    
    (\ref{qif-law-symmetry})
    Directly by definition:
    \begin{align*}
      \sem{LHS} = |\phi_0\>\<\phi_0|\otimes\sem{C_0}+|\phi_1\>\<\phi_1|\otimes\sem{C_1} = \sem{RHS}.
    \end{align*}
    
    (\ref{qif-law-idempotence})
    It holds that $\sem{\sqif{C}{q}{C}{}{}}=|0\>\< 0|\otimes \sem{C}+|1\>\< 1|\otimes \sem{C}=(|0\>\< 0|+|1\>\< 1|)\otimes \sem{C}=\sem{C}.$
    
    (\ref{qif-law-distributivity})
    By definition, we have:
    \begin{align*}
        \sem{RHS}&=|0\>_{q_2}\< 0|\otimes \sem{\sqif{C}{q_1}{C_0}{}{}} + |1\>_{q_2}\< 1|\otimes \sem{\sqif{C}{q_1}{C_1}{}{}}\\
        &= |0\>_{q_2}\< 0|\otimes(|0\>_{q_1}\< 0|\otimes \sem{C} +|1\>_{q_1}\< 1|\otimes \sem{C_0})\\
        &\quad +|1\>_{q_2}\< 1|\otimes(|0\>_{q_1}\< 0|\otimes \sem{C}+  |1\>_{q_1}\< 1|\otimes \sem{C_1})\\ 
        &=|0\>_{q_1}\< 0|\otimes(|0\>_{q_2}\< 0|+|1\>_{q_2}\< 1|)\otimes \sem{C}\\
        &\quad +|1\>_{q_1}\< 1|\otimes(|0\>_{q_2}\< 0| \otimes \sem{C_0}+ |1\>_{q_2}\< 1| \otimes \sem{C_1})\\
        &=|0\>_{q_1}\< 0|\otimes \sem{C} +|1\>_{q_1}\< 1|\otimes(\sem{\sqif{C_0}{q_2}{C_1}{}{}})\\
        &=\sem{LHS}. 
    \end{align*}
     
    (\ref{qif-law-nested})
    This follows by a straightforward expansion using the definition.
\end{proof}

\begin{proof}[Proof of Proposition~\ref{circuit-seq-law}]
    (\ref{circuit-seq-law-unit}--\ref{circuit-seq-law-associativity})
    These follow immediately from the definition.
    
    (\ref{circuit-seq-law-sequentiality})
    We have:
    \begin{align*}
        \sem{LHS} &= (|0\>_q\< 0|\otimes\sem{D_0}+|1\>_q\< 1|\otimes\sem{D_1})(|0\>_q\< 0|\otimes\sem{C_0}+|1\>_q\< 1|\otimes \sem{C_1})\\
        &=|0\>_q\< 0|0\>_q\< 0|\otimes \sem{D_0}\sem{C_0} +|0\>_q\< 0|1\>_q\< 1|\otimes \sem{D_0}\sem{C_1}\\ 
        &\quad +|1\>_q\< 1|0\>_q\< 0|\otimes \sem{D_1} \sem{C_0} +|1\>_q\< 1|1\>_q\< 0|\otimes \sem{D_1} \sem{C_1}\\
        &=|0\>_q\< 0|\otimes\sem{C_0; D_0}+|1\>_q\< 1|\otimes\sem{C_1; D_1}\\
        &=\sem{RHS}. 
    \end{align*}

    (\ref{circuit-seq-law-distributivity})
    This follows by a calculation similar to (\ref{circuit-seq-law-sequentiality}). 
\end{proof}

\begin{proof}[Proof of Proposition~\ref{choice-laws}]
    (\ref{choice-laws-symmetry})
    This can be shown by a routine calculation:
    \begin{align*}
        \sem{LHS} &= (|\psi_0\>_q\<\psi_0|\otimes\sem{C_0} + |\psi_1\>_q\<\psi_1|\otimes\sem{C_1}) U_{q} \\
        &= U_q|\phi_0\>\<\phi_0| \otimes\sem{C_0} + U_q|\phi_1\>_q\<\phi_1|\otimes\sem{C_1} \\
        &= \sem{RHS}.
    \end{align*}

    (\ref{choice-laws-sequentiality})
    This follows immediately from Proposition~\ref{circuit-seq-law}(\ref{circuit-seq-law-sequentiality}).

    (\ref{choice-laws-distributivity})
    We first note that $q\notin\qvar(C)$ and $\qvar(D)=\{q\}$ imply $\qvar(C)\cap\qvar(D)=\emptyset$.
    Then by Proposition~\ref{circuit-seq-law}(\ref{circuit-seq-law-commutativity}) we obtain:
    \begin{align*}
        C; \sqif{C_0}{D[q]}{C_1}{}{}
        &\equiv C; D; \sqif{C_0}{q}{C_1}{}{}\\
        &\equiv D; C; \sqif{C_0}{q}{C_1}{}{}\\
        &\equiv D; \sqif{(C; C_0)}{q}{(C; C_1)}{}{}\\
        &\equiv \sqif{(C; C_0)}{D[q]}{(C; C_1)}{}{}. 
    \end{align*}
    Other equivalences in this clause can be proved in a similar manner. 
\end{proof}

\subsection{Proofs of the Laws in Section~\ref{sec-laws-prog}}

\begin{proof}[Proof of Proposition~\ref{init-laws}]
    (\ref{init-laws-cancellation})
    For any density operator $\rho$, by Proposition~\ref{sempro}(\ref{sempro-init}) we have:
    \begin{align*}
        \sem{ \iinit{\qbar}{|\varphi\>}; \iinit{\qbar}{|\psi\>}}(\rho)&=\sem{ \iinit{\qbar}{|\psi\>}}\left(\sum_n|\varphi\>_\qbar\< n|\rho|n\>_\qbar\<\varphi|\right)\\
        &=\sum_{n^\prime}|\psi\>_\qbar\< n^\prime|\left(\sum_n|\varphi\>_\qbar\< n|\rho|n\>_\qbar\<\varphi|\right)|n^\prime\>_\qbar\<\psi|\\
        &=\left(\sum_{n^\prime}\< n^\prime|\varphi\>\<\varphi|n^\prime\>\right)\left(\sum_n|\psi\>_\qbar\< n|\rho|n\>_\qbar\<\psi|\right)\\
        &=\sum_n|\psi\>_\qbar\< n|\rho|n\>_\qbar\<\psi|\\
        &=\sem{ \iinit{\qbar}{|\psi\>}}(\rho),
    \end{align*}
    because $$\sum_{n^\prime}\< n^\prime|\varphi\>\<\varphi|n^\prime\>=\sum_{n^\prime}|\<\varphi|n^\prime\>|^2=\|\varphi\|^2=1.$$

    When $\qbar_1$ and $\qbar_2$ are distinct, we observe:
    \begin{align*}
        \sem{\iinit{\qbar_1}{|\phi\>}; \iinit{\qbar_2}{|\psi\>}}(\rho)
        &= \sum_{n_2}|\psi\>_{\qbar_2}\<n_2|\left( \sum_{n_1}|\phi\>_{\qbar_1}\<n_1|\rho |n_1\>_{\qbar_1}\<\phi|\right) |n_2\>_{\qbar_2}\<\psi|\\
        &= \sum_{(n_1,n_2)} (|\phi\>|\psi\>)_{\qbar_1,\qbar_2}(\<n_1|\<n_2|)\rho(|n_1\>|n_2\>)_{\qbar_1,\qbar_2}(\<\phi|\<\psi|) \\
        &= \sem{\iinit{\qbar_1, \qbar_2}{|\phi\>}|\psi\>}(\rho).
    \end{align*}

    (\ref{init-laws-unitary-elimination})
    For any density operator $\rho$, by Proposition~\ref{sempro}(\ref{sempro-init},~\ref{sempro-unitary}) we have:
    \begin{align*}
        \sem{\iinit{\qbar}{|\psi\>}; C}(\rho)
        &= \sem{C}\left(\sum_n|\psi\>_\qbar\<n|\rho|n\>_\qbar\<\psi|\right)\sem{C}^\dag\\
        &= \sum_n(\sem{C}|\psi\>)_\qbar\<n|\rho|n\>_\qbar(\sem{C}|\psi\>)_\qbar^\dag\\
        &= \sem{\iinit{\qbar}{\sem{C}|\psi\>}}(\rho).
    \end{align*}

    (\ref{init-laws-qif-elimination})
    Let $A_i=|\varphi_i\>\<\varphi_i|\otimes\sem{C_i}$ for $i=0,1$.
    Then $|\psi\>\bot|\varphi_0\>$ implies: for any density operator $\rho$, 
    \begin{align*}
        A_0|\psi\>\< n|\rho|n\>\<\psi|A_0^\dag=A_0|\psi\>\< n|\rho|n\>\<\psi|A_1^\dag
        = A_1|\psi\>\< n|\rho|n\>\<\psi|A_0^\dag=0.
    \end{align*}
    Additionally, we notice that, since $q$ is a qubit and thus its state space is of dimension two, $|\psi\>$ and $|\varphi_1\>$ differ only in a phase coefficient, i.e., $|\psi\> = e^{i\theta}|\varphi_1\>$ with phase $\theta\in\mathbb{R}$. Therefore, the following equations hold:
    \[
        (\<\varphi_1|\psi\>|\varphi_1\>)(\<\varphi_1|\psi\>|\varphi_1\>)^\dag = |\varphi_1\>\<\varphi_1| = |\psi\>\<\psi|.
    \]
    Consequently, with Proposition~\ref{sempro}(\ref{sempro-init}) and (\ref{sempro-unitary}) we have:
    {
        \allowdisplaybreaks
    \begin{align*}
        \sem{LHS}(\rho)&=(A_0+A_1)\left(\sum_n|\psi\>\< n|\rho|n\>\<\psi|\right)(A_0^\dag+A_1^\dag)\\
        &=\sum_n A_1|\psi\>\< n|\rho|n\>\<\psi|A_1^\dag\\
        &=\sum_n (|\varphi_1\>\<\varphi_1|\psi\>\otimes \sem{C_1})\<n|\rho|n\>((|\varphi_1\>\<\varphi_1|\psi\>)^\dag\otimes \sem{C_1}^\dag)\\
        &=\sem{C_1}\left(\sum_n (\<\varphi_1|\psi\>|\varphi_1\>)\< n|\rho|n\>(\<\varphi_1|\psi\>|\varphi_1\>)^\dag\right) \sem{C_1}^\dag\\
        &=\sem{C_1}\left(\sum_n |\varphi_1\>\< n|\rho|n\>\<\varphi_1|\right) \sem{C_1}^\dag\\
        &=\sem{C_1}(\sem{ \iinit{q}{|\varphi_1\>}}(\rho))\\
        &=\sem{C_1}\left(\sum_n |\psi\>\< n|\rho|n\>\<\psi|\right) \sem{C_1}^\dag\\
        &=\sem{C_1}(\sem{ \iinit{q}{|\psi\>}}(\rho)).
    \end{align*}
    }

    More generally, let $|\varphi\> = \sum_i\lambda_i|\phi_i\>|\psi_i\>$, notice that:
    \begin{align*}
        \left(\sum_{i=1}^{d}|\phi_i\>\<\phi_i|\otimes U_i\right)|\varphi\> =
        \sum_i\sum_j\lambda_j|\phi_i\>\<\phi_i|\phi_j\>(U_i|\psi_j\>) = \sum_i\lambda_i|\phi_i\>(U_i|\psi_i\>),
    \end{align*}
    so we have
    \begin{align*}
        &\sem{\iinit{\qbar,\rbar}{\sum_i\lambda_i|\phi_i\>|\psi_i\>}; 
            \iqif{\qbar}{\square_{i=1}^d |\phi_i\>}{C_i}}(\rho) \\
        =\ & \left(\sum_i|\phi_i\>_\qbar\<\phi_i|\otimes (U_i)_\rbar\right)\left(\sum_k |\varphi\>_{\qbar,\rbar}\<k|\rho|k\>_{\qbar,\rbar}\<\varphi|\right)\left(\sum_i|\phi_i\>_\qbar\<\phi_i|\otimes (U_i)_\rbar\right)^\dag\\
        =\ & \sum_k  \left(\left(\sum_i |\phi_i\>\<\phi_i|\otimes U_i\right)|\varphi\>\right)_{\qbar,\rbar}{}_{\qbar,\rbar}\<k|\rho|k\>_{\qbar,\rbar}\left(\left(\sum_i |\phi_i\>\<\phi_i|\otimes U_i\right)|\varphi\>\right)_{\qbar,\rbar}^\dag \\
        =\ &\sem{\iinit{\qbar,\rbar}{\sum_i\lambda_i|\phi_i\>(U_i|\psi_i\>)}}(\rho).
    \end{align*}
    
    (\ref{init-laws-if-elimination})
    We only prove the second law, and the first can be shown in a similar way. First, from the assumption that $M=[|\psi\>]$ and $|\varphi\>\bot|\psi\>$, it follows that for any density operator $\rho$, 
    \begin{align*}
        M_0|\varphi\>\< n|\rho|n\>\<\varphi|M_0^\dag =(I-|\psi\>\<\psi|)|\varphi\>\< n|\rho|n\>\<\varphi|(I-|\psi\>\<\psi|)=|\varphi\>\< n|\rho|n\>\<\varphi|
    \end{align*}
    and $M_1|\varphi\>\< n|\rho|n\>\<\varphi|M_1^\dag =0$.
    Then using Proposition~\ref{sempro}(\ref{sempro-init}) and (\ref{sempro-if}) we have:
    \begin{align*}
        \sem{LHS}(\rho)&=\sem{\sif{P_0}{M[q]}{P_1}}\left(\sum_n|\varphi\>\< n|\rho|n\>\<\varphi|\right)\\
        &=\sum_n\sem{\sif{P_0}{M[q]}{P_1}}(|\varphi\>\< n|\rho|n\>\<\varphi|)\\ 
        &=\sum_n\left[\sem{P_0}(M_0|\varphi\>\< n|\rho|n\>\<\varphi|M_0^\dag)+\sem{P_1}(M_1|\varphi\>\< n|\rho|n\>\<\varphi|M_1^\dag)\right]\\ 
        &=\sum_n\sem{P_0}(|\varphi\>\< n|\rho|n\>\<\varphi|)\\
        &=\sem{P_0}\left(\sum_n|\varphi\>\< n|\rho|n\>\<\varphi|\right)\\
        &=\sem{RHS}(\rho).
    \end{align*}
    
    (\ref{init-laws-if-expansion})
    We calculate as follows according to Proposition~\ref{sempro}(\ref{sempro-if},\ref{sempro-init}):
    \begin{align*}
        \sem{LHS}(\rho) &= \sum_k |0\>_\rbar\<k| \Big(\sum_m(|m\>_\qbar\<m| \otimes \sem{P_m})\rho (|m\>_\qbar\<m|\otimes \sem{P_m})\Big) |k\>_\rbar\<0| \\
        &= \sum_k |0\>_\rbar\<k| (|k{\downarrow}_\qbar\>_\qbar\<k{\downarrow}_\qbar| \otimes \sem{P_{k{\downarrow}_\qbar}})\rho (|k{\downarrow}_\qbar\>_\qbar\<k{\downarrow}_\qbar|\otimes \sem{P_{k{\downarrow}_\qbar}}) |k\>_\rbar\<0| \\
        &= \sum_k (|0\>_\rbar\<k| \otimes \sem{P_{k{\downarrow}_\qbar}})\rho (|k\>_\rbar\<0|\otimes \sem{P_{k{\downarrow}_\qbar}}) \\
        &= \sum_k |0\>_\rbar\<k| \Big( \sum_{k'}(|k'\>_\rbar\<k'| \otimes \sem{P_{k{\downarrow}_\qbar}})\rho (|k'\>_\rbar\<k'|\otimes \sem{P_{k{\downarrow}_\qbar}})\Big) |k\>_\rbar\<0| \\
        &= \sem{RHS}(\rho).
    \end{align*}
    
\end{proof}

To prove other results in Section~\ref{sec-laws-prog}, we need some technical lemmas:  
\begin{lem}
\label{lemma-technique1}
    Let $A$ and $B$ be two linear operators.
    If for every $|\alpha\>$ there exists $k_{|\alpha\>}$ such that $$B|\alpha\> = k_{|\alpha\>} A|\alpha\>$$ then $B = cA$ for some $c\in \mathbb{C}$, i.e., $B \lrtimes_{|c|} A$.
\end{lem}
\begin{proof}
    Let $A = UDV^\dag$ be the singular value decomposition of $A$, i.e., $U$ and $V$ are unitary operators and $D = \mathrm{diag}(d_1,d_2,\cdots, d_n)$ where $d_i\ge 0$ for all $i$.
    Notice that: 
    \[
        \forall\,|\alpha\>,\ BVV^\dag |\alpha\> = k_{|\alpha\>} UDV^\dag |\alpha\> \quad\Longrightarrow\quad \forall\,|\alpha\>,\ BV|\alpha\> = k_{V^\dag|\alpha\>}UD|\alpha\>.
    \]
    We write $BV = (|v_1\>,|v_2\>,\cdots,|v_n\>)$ and $U = (|u_1\>,|u_2\>,\cdots,|u_n\>)$ in the column vector form.
    Considering $|\alpha\> = |i\>$, we have:
    \[
        \forall\, i,\ |v_i\> = c_id_i|u_i\>
    \]
    where $c_i = k_{V^\dag|i\>}$.
    For any $i\neq j$ such that $d_i\neq 0$ and $d_j\neq 0$, select $|\alpha\> = |i\>+|j\>$, we have 
    \begin{align*}
        BV|\alpha\> &= BV(|i\>+|j\>) = |v_i\>+|v_j\> = c_id_i|u_i\>+c_jd_j|u_j\> \\
        &= k_{V^\dag(|i\>+|j\>)}UD(|i\>+|j\>) = k_{V^\dag(|i\>+|j\>)}(d_i|u_i\>+d_j|u_j\>),
    \end{align*}
    which shows that $c_i = c_j = k_{V^\dag(|i\>+|j\>)}$.
    This leads to $c_i = c$ if $d_i\neq 0$ for some constant $c\in \mathbb{C}$.
    Furthermore, if $d_i = 0$, then $|v_i\> = 0 = cd_i|u_i\>$.
    Therefore
    \[
        BV = (cd_1|u_1\>,\cdots,cd_n|u_n\>) = cUD,
    \]
    and hence
    \[
        B = BVV^\dag = cUDV^\dag = cA.
    \]
\end{proof}

\begin{lem}
\label{lem:app perp and ltimes}
    For linear operators $A, B$, we have the following:
    \begin{enumerate}
        \item $A\perp B$ iff for every state $\rho$, $AB\rho B^\dag A^\dag = 0$.
        \item $A\ltimes B$ iff for every state $\rho$, $AB\rho B^\dag A^\dag = B\rho B^\dag$.
        \item
        The following two statements are equivalent:
        \begin{enumerate}
            \item $A_0\lrtimes_{c_0} B$ and $A_1\lrtimes_{c_1} B$ for some $c_0$ and $c_1$ with $c_0^2+c_1^2 = 1$;
            \item for every state $\rho$, $A_0\rho A_0^\dag + A_1\rho A_1^\dag = B\rho B^\dag$.
        \end{enumerate}
        As a direct corollary, the following two statements are equivalent: 
        \begin{enumerate}
            \item[(c)] $B\rtimes_{c_0} A_0$ and $B\rtimes_{c_1} A_1$ for some $c_0$ and $c_1$ with $c_0^2+c_1^2 = 1$;
            \item[(d)] for every state $\rho$, $BA_0\rho A_0^\dag B^\dag + BA_1\rho A_1^\dag B^\dag = B\rho B^\dag$.
        \end{enumerate}
        \item If $\{A_0,A_1\}$ is a binary measurement, then $A_0\ltimes B$ implies $A_1\perp B$; symmetrically, $A_1\ltimes B$ implies $A_0\perp B$.
    \end{enumerate}
\end{lem}
\begin{proof}
    (1) ``$\Rightarrow$'':
    This is immediate, since $A\perp B$ means $AB = 0$. 
    
    ``$\Leftarrow$'':
    Set $\rho = |\alpha\>\<\alpha|$ be arbitrary pure state.
    Then $AB\rho B^\dag A^\dag = (AB|\alpha\>)(AB|\alpha\>)^\dag = 0$ implies $AB|\alpha\>=0$, so $AB = 0$, i.e., $A\perp B$.

    (2) ``$\Rightarrow$'':
    $A\ltimes B$, so there exists $c\in\mathbb{C}$, such that $|c| = 1$ and $AB = cB$.
    Thus, we have $$AB\rho B^\dag A^\dag = cB\rho (cB)^\dag = |c|^2 B\rho B^\dag = B\rho B^\dag.$$ 
    
    ``$\Leftarrow$'':
    Let $\rho = |\alpha\>\<\alpha|$ be an arbitrary pure state.
    Then $$AB\rho B^\dag A^\dag = (AB|\alpha\>)(AB|\alpha\>)^\dag = (B|\alpha\>)(B|\alpha\>)^\dag$$ which implies $AB|\alpha\> = k_{|\alpha\>}B|\alpha\>$ where $k_{|\alpha\>}\in\mathbb{C}$ such that $|k_{|\alpha\>}| = 1$.
    We apply Lemma~\ref{lemma-technique1} and assert that there exists $c\in\mathbb{C}$ such that $AB=cB$.
    Back to the assumption, we have: $$AB\rho B^\dag A^\dag = cB\rho (cB)^\dag = |c|^2B\rho B^\dag$$ So, $|c| = 1$, i.e., $A\ltimes B$ if $B\neq 0$. If $B = 0$, then $A\ltimes B$ also holds trivially.

    (3) ``(a) $\Rightarrow$ (b)'':
    By assumption, there exists $d_0, d_1\in\mathbb{C}$ such that $|d_0| = c_0$, $|d_1| = c_1$, $A_0 = d_0 B$ and $A_1 = d_1 B$.
    So, $$A_0\rho A_0^\dag + A_1\rho A_1^\dag = |d_0|^2 B \rho B^\dag + |d_1|^2 B \rho B^\dag = (c_0^2+c_1^2)B \rho B^\dag = B \rho B^\dag.$$
    
    ``(b) $\Leftarrow$ (a)'':
    Let $\rho = |\alpha\>\<\alpha|$ be an arbitrary pure state.
    Then $$(A_0|\alpha\>)(A_0|\alpha\>)^\dag + (A_1|\alpha\>)(A_1|\alpha\>)^\dag = (B|\alpha\>)(B|\alpha\>)^\dag.$$
    Since the RHS is rank one, both LHS terms must be proportional to the same vector. Therefore, there exist $k_{0,|\alpha\>}, k_{1,|\alpha\>}\in \mathbb{C}$ such that $A_0|\alpha\> = k_{0,|\alpha\>}B|\alpha\>$ and $A_1|\alpha\> = k_{1,|\alpha\>}B|\alpha\>$.
    According to Lemma~\ref{lemma-technique1}, there exists $d_0, d_1\in\mathbb{C}$ such that $A_0 = d_0B$ and $A_1 = d_1B$.
    Back to the assumption, we have
    \[
        A_0\rho A_0^\dag + A_1\rho A_1^\dag = B\rho B^\dag = (|d_0|^2 + |d_1|^2) B\rho B^\dag,
    \]
    so $|d_0|^2 + |d_1|^2 = 1$ and $A_0\lrtimes_{|d_0|} B$ and $A_1\lrtimes_{|d_1|} B$ if $B \neq 0$.
    If $B = 0$, then it holds trivially.
    
    (4) By the assumption, there exists $c\in\mathbb{C}$ such that $|c| = 1$ and $A_0B = cB$.
    Note that 
    \begin{align*}
        (A_0B)^\dag (A_0B) + (A_1B)^\dag(A_1B) &= B^\dag (A_0^\dag A_0+A_1^\dag A_1) B = B^\dag B \\
        &= (cB)^\dag(cB) + (A_1B)^\dag(A_1B) = B^\dag B + (A_1B)^\dag(A_1B)
    \end{align*}
    which implies $(A_1B)^\dag(A_1B) = 0$.  So $A_1B = 0$, i.e., $A_1\perp B$.

\end{proof}

Now we are able to prove the remaining laws in Section~\ref{sec-laws-prog}.  

\begin{proof}[Proof of Proposition~\ref{if-laws}]
    The proofs of laws (\ref{if-laws-truth-falsity}) and (\ref{if-laws-complementation}) are straightforward and thus omitted. 
    
    (\ref{if-laws-idempotence})
    For any density operator $\rho$, it follows from the linearity of $\sem{P}$ and Proposition~\ref{sempro}(\ref{sempro-if}) that 
    \begin{align*}
        \sem{\sif{P}{M}{P}}(\rho)
        &=\sem{P}(M_0\rho M_0^\dag)+\sem{P}(M_1\rho M_1^\dag)\\
        &=\sem{P}(M_0\rho M_0^\dag+M_1\rho M_1^\dag)\\
        &=\sem{P}(\sem{[M]}(\rho))\\
        &=\sem{[M]; P}(\rho). 
    \end{align*}
    
    (\ref{if-laws-associativity})
    Since $M$ is a projective measurement, we have $M_0M_1=0$ and $M_1M_1=M_1$.
    Then for any $\rho$, using Proposition~\ref{sempro}(\ref{sempro-if}) we obtain: 
    \begin{align*}
        \sem{\sif{P}{M}{(\sif{Q}{M}{R})}}(\rho)
        &=\sem{P}(M_0\rho M_0^\dag)+\sem{\sif{Q}{M}{R}}(M_1\rho M_1^\dag)\\
        &=\sem{P}(M_0\rho M_0^\dag)+(\sem{Q}(M_0M_1\rho M_1^\dag M_0^\dag)+\sem{R}(M_1M_1\rho M_1^\dag M_1^\dag))\\
        &=\sem{P}(M_0\rho M_0^\dag)+\sem{R}(M_1\rho M_1^\dag)\\
        &=\sem{\sif{P}{M}{R}}(\rho). 
    \end{align*}
    Similarly, we can prove that $\sem{\sif{(\sif{P}{M}{Q})}{M}{R}}(\rho)= \sem{\sif{P}{M}{R}}(\rho).$
    
    (\ref{if-laws-if-elimination})
    For the first law, for any $\rho$, by Proposition~\ref{sempro}(\ref{sempro-if}) we have:
    \begin{align*}
        \sem{LHS}(\rho)
        &=\sem{\sif{P}{N}{Q}}(M_1\rho M_1^\dag)\\ 
        &=\sem{P}(N_0M_1\rho M_1^\dag N_0^\dag)+\sem{Q}(N_1M_1\rho M_1^\dag N_1^\dag)\\
        &=\sem{Q}(N_1M_1\rho M_1^\dag N_1^\dag)\\
        &=\sem{Q}(\sem{(K]}(\rho))\\
        &=\sem{RHS}(\rho).
    \end{align*}
    Here, the third equality comes from the assumption that $M\blacktriangleright N$, i.e., $N_0M_1=0$, and the fourth equality comes from the assumption that $K=M\smeet N$, i.e., $K_1=N_1M_1$.
    
    For the second law, since $M\propto N$, i.e., $M_1\ltimes N_1$, by Lemma~\ref{lem:app perp and ltimes}, we must have $M_1N_1\rho N_1^\dag M_1^\dag = N_1\rho N_1^\dag$ and $M_0N_1\rho N_1^\dag M_0^\dag = 0$ for all $\rho$. Thus, the following hold:
    \begin{align*}
        \sem{LHS}(\rho) 
        &= \sem{P}(N_0\rho N_0^\dag) + (\sem{Q}(M_0N_1\rho N_1^\dag M_0^\dag) + \sem{R}(M_1N_1\rho N_1^\dag M_1^\dag) \\
        &= \sem{P}(N_0\rho N_0^\dag) + \sem{R}(N_1\rho N_1^\dag) \\
        &= \sem{\sif{P}{N}{R}}(\rho).
    \end{align*}
    This completes the proof of the second law.
\end{proof}

Before presenting the proof of Proposition~\ref{nested-if-law}, we state a lemma that clarifies the conditions required in the laws of the proposition.
Indeed, this lemma is also needed in the proof of these laws. 

\begin{lem}
\label{lemma-equal-nested-if-cond}
The following equations hold:
\begin{itemize}
    \item $[N)\equiv \sif{[K)}{M}{[L)}$ iff $K_0M_0\lrtimes_{c_0} N_0$ and $L_0M_1\lrtimes_{c_1}N_0$ for some $c_0$ and $c_1$ with $c_0^2+c_1^2 = 1$;
    \item $(N]\equiv \sif{(K]}{M}{(L]}$ iff $K_1M_0\lrtimes_{c_0} N_1$ and $L_1M_1\lrtimes_{c_1}N_1$ for some $c_0$ and $c_1$ with $c_0^2+c_1^2 = 1$;
    \item $[N)\equiv [M]; [N)$ iff $N_0\rtimes_{c_0}M_0$ and $N_0\rtimes_{c_1}M_1$ with $c_0^2+c_1^2 = 1$;
    \item $(N]\equiv [M]; (N]$ iff $N_1\rtimes_{c_0}M_0$ and $N_1\rtimes_{c_1}M_1$ with $c_0^2+c_1^2 = 1$;
\end{itemize}
\end{lem}
\begin{proof}
    Direct calculation from the definition and Lemma~\ref{lem:app perp and ltimes}.
\end{proof}

Now we are in a position to present: 
\begin{proof}[Proof of Proposition~\ref{nested-if-law}]
    (\ref{nested-if-law-reduction})
    For any density operator $\rho$, it follows from the assumptions $(N]\equiv \sif{(K]}{M}{(L]}$ and $[N)\equiv \sif{[K)}{M}{[L)}$ that 
    \[
        N_i\rho N_i^\dag=K_iM_0\rho M_0^\dag K_0^\dag+L_iM_1\rho M_1^\dag L_i^\dag\ (i=0,1).
    \]
    Then, by Proposition~\ref{sempro}(\ref{sempro-if}) we obtain
    \begin{align*}
        \sem{LHS}(\rho)
        &=\sem{\sif{P}{K}{Q}}(M_0\rho M_0^\dag)+\sem{\sif{P}{L}{Q}}(M_1\rho M_1^\dag)\\
        &=(\sem{P}(K_0M_0\rho M_0^\dag K_0^\dag)+ \sem{Q}(K_1M_0\rho M_0^\dag K_1^\dag))\\
        &\quad + (\sem{P}(L_0M_1\rho M_1^\dag L_0^\dag) +\sem{Q}(L_1M_1\rho M_1^\dag L_1^\dag))\\ &=\sem{P}(K_0M_0\rho M_0^\dag K_0^\dag+L_0M_1\rho M_1^\dag L_0^\dag)+ \sem{Q}(K_1M_0\rho M_0^\dag K_1^\dag+L_1M_1\rho M_1^\dag L_1^\dag)\\
        &=\sem{P}(N_0\rho N_0^\dag) +\sem{Q}(N_1\rho N_1^\dag)\\
        &= \sem{RHS}(\rho). 
    \end{align*}
    
    (\ref{nested-if-law-left-distributivity})
    For any $\rho$, by Proposition~\ref{sempro}(\ref{sempro-if}) we have:
    \begin{align*}
        \sem{LHS}(\rho)&=\sem{R}(N_0\rho N_0^\dag)+ \sem{\sif{P}{M}{Q}}(N_1\rho N_1^\dag)\\ 
        &= \sem{R}(N_0\rho N_0^\dag) + \sem{P}(M_0N_1\rho N_1^\dag M_0^\dag)+\sem{Q}(M_1N_1\rho N_1^\dag M_1^\dag)\\ 
        &=\sem{R}(N_0 (M_0\rho M_0^\dag +M_1\rho M_1^\dag) N_0^\dag) + \sem{P}(N_1M_0\rho M_0^\dag N_1^\dag)+\sem{Q}(N_1M_1\rho M_1^\dag N_1^\dag)\\ 
        &=\left(\sem{R}(N_0M_0\rho M_0^\dag N_0^\dag)  + \sem{P}(N_1M_0\rho M_0^\dag N_1^\dag)\right)\\
        &\quad +\left(\sem{R}(N_0M_1\rho M_1^\dag N_0^\dag)+\sem{Q}(N_1M_1\rho M_1^\dag N_1^\dag)\right)\\ 
        &=\sem{\sif{R}{N}{P}}(M_0\rho M_0^\dag) +  \sem{\sif{R}{N}{Q}}(M_1\rho M_1^\dag)\\ 
        &= \sem{RHS}(\rho). 
    \end{align*}
    Here, the third equality follows from the assumption that $[N)\equiv [M]; [N)$ and $M\diamond_L N$. 
    
    (\ref{nested-if-law-right-distributivity}) can be proved in a way similar to (\ref{nested-if-law-left-distributivity}). 
    
    (\ref{nested-if-law-left-distributivity-projection}) and (\ref{nested-if-law-right-distributivity-projection}) are special cases of (\ref{nested-if-law-left-distributivity}) and (\ref{nested-if-law-right-distributivity}), and are therefore omitted.
\end{proof}

\begin{proof}[Proof of Proposition~\ref{sequential-laws}]
    (\ref{sequential-laws-unit-zero}) and (\ref{sequential-laws-commutativity}) are trivial.
    
    (\ref{sequential-laws-associativity})
    By Proposition~\ref{sempro}(\ref{sempro-seq}), it holds that for any density operator $\rho$, \begin{align*}
        \sem{(P; Q); R}(\rho)
        &=\sem{R}(\sem{Q; R}(\rho))\\
        &=\sem{R}(\sem{Q}(\sem{P}(\rho)))\\
        &=\sem{Q; R}(\sem{P}(\rho))\\
        &=\sem{P; (Q; R)}(\rho). 
    \end{align*}
    
    (\ref{sequential-laws-sequentiality})
    Since $M$ is a projective measurement, so $M_m^\dag = M_m$ and $M_mM_m = M_m$ for $m = 0,1$, $M_0M_1 = M_1M_0 = 0$, thus:
    \begin{align*}
        \sem{LHS}(\rho) 
        &= \sem{Q_0}(M_0(\sem{P_0}(M_0\rho M_0^\dag) + \sem{P_1}(M_1\rho M_1^\dag))M_0^\dag) \\
        &\quad + \sem{Q_1}(M_1(\sem{P_0}(M_0\rho M_0^\dag) + \sem{P_1}(M_1\rho M_1^\dag))M_1^\dag) \\
        &= \sem{Q_0}(\sem{P_0}(M_0\rho M_0^\dag)) + \sem{Q_1}(\sem{P_1}(M_1\rho M_1^\dag)) \\
        &= \sem{RHS}(\rho).
    \end{align*}
    
    (\ref{sequential-laws-right-distributivity})
    For any density operator $\rho$, by Proposition~\ref{sempro}(\ref{sempro-seq}) and (\ref{sempro-if}) we have: 
    \begin{align*}
        \sem{(\sif{P_0}{M}{P_1}); P}(\rho)
        &=\sem{P}(\sem{(\sif{P_0}{M}{P_1})}(\rho))\\ 
        &=\sem{P}(\sem{P_0}(M_0\rho M_0^\dag)+\sem{P_1}(M_1\rho M_1))\\ 
        &=\sem{P}(\sem{P_0}(M_0\rho M_0^\dag))+\sem{P}(\sem{P_1}(M_1\rho M_1))\\
        &=\sem{P_0; P}(M_0\rho M_0^\dag)+\sem{P_1; P}(M_1\rho M_1)\\
        &=\sem{\sif{(P_0; P)}{M}{(P_1; P)}}.
    \end{align*}
    
    (\ref{sequential-laws-left-distributivity1}) is a special case of (\ref{sequential-laws-left-distributivity2}).
    
    (\ref{sequential-laws-left-distributivity2})
    For any $\rho$, using Proposition~\ref{sempro}(\ref{sempro-unitary}--\ref{sempro-if}) we obtain:
    \begin{align*}
        \sem{P; (\sif{P_0}{M}{P_1})}(\rho)
        &= \sem{\sif{P_0}{M}{P_1}}(\sem{P}\rho\sem{P}^\dag)\\
        &=\sem{P_0}(M_0\sem{P}(\rho)M_0^\dag)+\sem{P_1}(M_1\sem{P}(\rho) M_1^\dag)\\
        &=\sem{P_0}(\sem{P'}(N_0\rho N_0^\dag))+\sem{P_1}(\sem{P'}(N_1\rho N_1^\dag))\\
        &=\sem{P'; P_0}(N_0\rho N_0^\dag)+\sem{P'; P_1}(N_1\rho N_1^\dag)\\
        &=\sem{\sif{(P'; P_0)}{M}{(P'; P_1)}}(\rho). 
    \end{align*}
    Here, the third equality follows from the assumption that $P; [M)\equiv [N); P'$ and $P; (M]\equiv (N]; P'.$
\end{proof}

\begin{proof}[Proof of Proposition~\ref{prop:interplay-qif-if}]
    Since $M_i = |\psi_i\>\<\psi_i|$, we have:
    \begin{align*}
        \sem{LHS}(\rho) 
        &= \sum_i\sem{P_i}\left(M_i \left(\sum_j|\psi_j\>\<\psi_j|\otimes\sem{C_j}\right)\rho \left(\sum_j|\psi_j\>\<\psi_j|\otimes\sem{C_j}\right)^\dag   M_i^\dag\right) \\
        &= \sum_i\sem{P_i}\left( \left(\sum_j|\psi_i\>\<\psi_i|\cdot |\psi_j\>\<\psi_j|\otimes\sem{C_j}\right)\rho \left(\sum_j|\psi_j\>\<\psi_j|\cdot |\psi_i\>\<\psi_i| \otimes\sem{C_j}\right)^\dag\right)\\
        &= \sum_i\sem{P_i}\left( (|\psi_i\>\<\psi_i|\otimes\sem{C_i})\rho (|\psi_i\>\<\psi_i| \otimes\sem{C_i})^\dag\right) \\
        &= \sum_i\sem{P_i}\left(\sem{C_i} (M_i\rho M_i^\dag)\sem{C_i}^\dag\right) \\
        &= \sem{RHS}(\rho).
    \end{align*}
\end{proof}

\subsection{Proofs of the Laws in Section~\ref{sec-recur}}

\begin{proof}[Proof of Proposition~\ref{prop-fix}]
    It was shown in~\cite{Ying16} (see Proposition 3.3.6 therein) that for each quantum operation $\E\in\mathcal{QO}(\hs_\QVar)$, there exists a quantum program $P_\E$ such that $\sem{P_\E}=\E.$
    Therefore, we can define a mapping $\mathbb{F}: \mathcal{QO}(\hs_\QVar)\rightarrow \mathcal{QO}(\hs_\QVar)$ by 
    \begin{equation}
    \label{sem-rec-func}
        \mathbb{F}(\E)=\sem{F(P_\E)}.
    \end{equation}
    By induction on the structure of $F$, it is straightforward to show that (\ref{sem-rec-func}) is well-defined in the sense that it does not depend on the choice of $P_\E$, and $\mathbb{F}$ is continuous in the CPO $\left(\mathcal{QO}(\hs_\QVar),\sqsubseteq\right)$. 
    
    (1) We compute:
    \begin{align*}
        \sem{F(\mu X.F(X))}
        &= \mathbb{F}\left(\bigsqcup_{n=0}^\infty\sem{F^{(n)}}\right)
         = \bigsqcup_{n=0}^\infty\mathbb{F}\left(\sem{F^{(n)}}\right)\\
        &= \bigsqcup_{n=0}^\infty\sem{F^{(n+1)}}
         = \bigsqcup_{n=1}^\infty\sem{F^{(n)}}\\
        &= \bigsqcup_{n=0}^\infty\sem{F^{(n)}}
         = \sem{\mu X.F(X)}.
    \end{align*}
    
    (2) To prove $\mu X.F(X)\sqsubseteq P$, it suffices to show that $\sem{F^{(n)}}\sqsubseteq \sem{P}$ for every $n\geq 0$. This follows easily by induction on $n$. 
\end{proof}
    
\begin{proof}[Proof of Proposition~\ref{loop-law}] 
    (\ref{loop-law-fixpoint}) 
    \begin{align*}
        \sem{\sif{}{M}{(P; (\swhile{M}{P}))}}(\rho)
        &=M_0\rho M_0^\dag +\sem{P; (\swhile{M}{P})}\left(M_1\rho M_1^\dag\right)\\
        &=\E_0 (\rho) + \sum_{n=0}^\infty\left(\E_0\circ(\sem{P}\circ\E_1)^n\right)\left((\sem{P} \circ \E_1)(\rho)\right)\\ 
        &=\sum_{n=0}^\infty\left(\E_0\circ(\sem{P}\circ\E_1)^n\right)\left(\rho\right)\\ 
        &= \sem{(\swhile{M}{P})}.
    \end{align*}
    
    (\ref{loop-law-unfolding})
    For any density operator $\rho$, we first note that
    \begin{align*}
        &\E_0(N_1\rho N_1^\dag) = M_0N_1\rho N_1^\dag M_0^\dag = 0, \\
        &\E_1(N_1\rho N_1^\dag) = M_1 N_1\rho N_1^\dag M_1^\dag = N_1\rho N_1^\dag,
    \end{align*}
    from the assumption of $M\propto N$ and Lemma~\ref{lem:app perp and ltimes}.
    Together with Proposition~\ref{sempro}(\ref{sempro-seq}--\ref{sempro-while}), this yields:
    {
    \allowdisplaybreaks
    \begin{align*}
        \sem{\sif{}{N}{(P; (\swhile{M}{P}))}}(\rho)
        &=N_0\rho N_0^\dag +\sem{P; (\swhile{M}{P})}\left(N_1\rho N_1^\dag\right)\\
        &=N_0\rho N_0^\dag+\sum_{n=0}^\infty\left(\E_0\circ(\sem{P}\circ\E_1)^n\right)\left(\sem{P}(N_1\rho N_1^\dag)\right)\\ 
        &=N_0\rho N_0^\dag+\sum_{n=0}^\infty\left(\E_0\circ(\sem{P}\circ\E_1)^n\right)\left(\sem{P}(\E_1(N_1\rho N_1^\dag))\right)\\ 
        &=N_0\rho N_0^\dag+\sum_{n=0}^\infty\left(\E_0\circ(\sem{P}\circ\E_1)^{n+1}\right)\left(N_1\rho N_1^\dag\right)\\ 
        &=N_0\rho N_0^\dag+\E_0\left(N_1\rho N_1^\dag\right)+\sum_{n=1}^\infty(\E_0\circ(\sem{P}\circ\E_1)^n)(N_1\rho N_1^\dag)\\ 
        &=N_0\rho N_0^\dag+\sum_{n=0}^\infty\left(\E_0\circ(\sem{P}\circ\E_1)^n\right)\left(N_1\rho N_1^\dag\right)\\ 
        &=\sem{\sif{}{N}{(\swhile{M}{P})}}(\rho).
    \end{align*}
    }
    
    (\ref{loop-law-elimination})
    For any density operator $\rho$, we first note that
    \begin{align*}
        &\E_0(N_1\rho N_1^\dag)=M_0N_1\rho N_1^\dag M_0^\dag =(M^\bot)_0N_1\rho N_1^\dag(M^\bot)_1^\dag =N_1\rho N_1^\dag, \\
        &\E_1(N_1\rho N_1^\dag)=M_1 N_1\rho N_1^\dag M_1^\dag =(M^\bot)_0N_1\rho N_1^\dag(M^\bot)_0^\dag=0,
    \end{align*}
    from the assumption of $M^\bot\propto N$ and Lemma~\ref{lem:app perp and ltimes}.
    Therefore, by Proposition~\ref{sempro}(\ref{sempro-if},~\ref{sempro-while}) we have: 
    \begin{align*}
        \sem{\sif{}{N}{(\swhile{M}{P})}}(\rho)
        &=N_0\rho N_0^\dag +\sem{\swhile{M}{P}}\left(N_1\rho N_1^\dag\right)\\
        &=N_0\rho N_0^\dag + \sum_{n=0}^\infty\left(\E_0\circ(\sem{P}\circ\E_1)^n\right)\left(N_1\rho N_1^\dag\right)\\ 
        &=N_0\rho N_0^\dag+\E_0\left(N_1\rho N_1^\dag\right) + \sum_{n=1}^\infty\left(\E_0\circ(\sem{P}\circ\E_1)^{n-1}\right)\left(\sem{P}\left(\E_1\left(N_1\rho N_1^\dag\right)\right)\right)\\ 
        &=N_0\rho N_0^\dag+N_1\rho N_1^\dag+\sum_{n=1}^\infty\left(\E_0\circ(\sem{P}\circ\E_1)^{n-1}\right)(\sem{P}(0))\\ 
        &=N_0\rho N_0^\dag+N_1\rho N_1^\dag\\ 
        &=\sem{\sif{}{N}{\iskip}}(\rho).
    \end{align*}
    
    (\ref{loop-law-postcondition})
    For any density operator $\rho$, we first note that
    \begin{align*}
        &N_0\E_0(\rho) N_0^\dag=N_0 M_0\rho M_0^\dag N_0^\dag = M_0\rho M_0^\dag = \E_0(\rho), \\
        &N_1\E_0(\rho) N_1^\dag=N_1 M_0\rho M_0^\dag N_1^\dag =0,
    \end{align*}
    from the assumption of $N\gg M$ and Lemma~\ref{lem:app perp and ltimes}.
    Thus,
    \begin{align*}
        \sem{RHS}(\rho)
        &= N_0\Big(\sum_n\big(\E_0\circ (\sem{P}\circ\E_1)^n\big)(\rho)\Big)N_0^\dag
        + N_1\Big(\sum_n\big(\E_0\circ (\sem{P}\circ\E_1)^n\big)(\rho)\Big)N_1^\dag\\
        &= \sum_nN_0\E_0\Big(\big((\sem{P}\circ\E_1)^n\big)(\rho)\Big)N_0^\dag
        + \sum_nN_1\E_0\Big(\big((\sem{P}\circ\E_1)^n\big)(\rho)\Big)N_1^\dag \\
        &= \sum_n\E_0\Big(\big((\sem{P}\circ\E_1)^n\big)(\rho)\Big) \\
        &= \sem{LHS}(\rho).
    \end{align*}
\end{proof}

\subsection{Proofs of the Laws in Section~\ref{sec-nondet}}

\begin{proof}[Proof of Proposition~\ref{laws-nd}]
    The proofs of (\ref{laws-nd-commutativity}), (\ref{laws-nd-associativity}), and (\ref{laws-nd-idempotence}) are straightforward. 
    
    (\ref{laws-nd-distributivity-if})
    Let $\E_0, \E_1$ be the quantum operations defined by the measurement $M=\{M_0,M_1\}$; i.e., $\E_i(\rho)=M_i\rho M_i^\dag$ for all density operators $\rho$.
    Uusing Proposition~\ref{sempro}(\ref{sempro-if}), we have:  
    \begin{align*}
        \sem{\sif{(P\sqcup Q)}{M}{R}}
        &=\{\E\circ\E_0+\mathcal{F}\circ\E_1 \mid \E\in\sem{P\sqcup Q}\ {\rm and}\ \mathcal{F}\in\sem{R}\}\\
        &=\{\E\circ \E_0+\mathcal{F}\circ\E_1 \mid \E\in \sem{P}\cup\sem{Q}\ {\rm and}\ \mathcal{F}\in\sem{R}\}\\
        &=\{\E\circ \E_0+\mathcal{F}\circ\E_1 \mid \E\in \sem{P}\ {\rm and}\ \mathcal{F}\in\sem{R}\}\\
        &\quad \cup \{\E\circ \E_0+\mathcal{F}\circ\E_1 \mid \E\in \sem{Q}\ {\rm and}\ \mathcal{F}\in\sem{R}\}\\
        &= \sem{\sif{P}{M}{R}}\cup\sem{\sif{P}{M}{R}}\\
        &=\sem{(\sif{P}{M}{R})\sqcup (\sif{Q}{M}{R})}.
    \end{align*}
    This proves the first equivalence. 
    The second equivalence can be proved in a similar way. 
    
    (\ref{laws-nd-distributivity-seq})
    We prove only the first equivalence; the second follows analogously.
    It follows from Proposition~\ref{sempro}(\ref{sempro-seq}) that 
    \begin{align*}
        \sem{P; (Q\sqcup R)}
        &=\{\mathcal{F}\circ \E\mid\E\in \sem{P}\ {\rm and}\ \mathcal{F}\in \sem{Q\sqcup R}\}\\
        &=\{\mathcal{F}\circ \E\mid\E\in \sem{P}\ {\rm and}\ \mathcal{F}\in \sem{Q}\cup\sem{R}\}\\
        &=\{\mathcal{F}\circ \E\mid\E\in \sem{P}\ {\rm and}\ \mathcal{F}\in \sem{Q}\}\cup \{\mathcal{F}\circ \E\mid\E\in \sem{P}\ {\rm and}\ \mathcal{F}\in \sem{R}\}\\
        &=\sem{P; Q}\cup \sem{P; R}\\
        &=\sem{(P; Q)\sqcup (P; R)}. 
    \end{align*}
\end{proof}

\section{Detailed Proof of Theorem~\ref{thm-defer}}\label{proof-app}

To provide a detailed proof of Theorem~\ref{thm-defer}, we first introduce several lemmas. 
Let us first recall the following lemma from~\cite{NC00}:

\begin{lem} 
\label{lem:equivalent-measure}
    Consider any quantum measurement $M = \{M_i\}_{i\in I}$ on a Hilbert space $\hs$, and an auxiliary system with a Hilbert space $\hs^\prime$ of dimension $n = |I|$.
    Let $\{|i\>\}_{i\in I}$ be an orthonormal basis of $\hs^\prime$.
    Then there exists a unitary operator $U_M$ on $\hs\otimes \hs^\prime$ such that the measurement $M$ can be equivalently implemented by:
    \begin{enumerate}
        \item The auxiliary system is initialised in $|0\>$ (one of the basis states);
        \item Perform the unitary $U_M$ on the primary and auxiliary system;
        \item Measure the auxiliary system in the basis $\{|i\>\}_{i\in I}$. 
    \end{enumerate}
    If the initial state is $\rho$, then with probability $\tr(M_i\rho M_i^\dag)$, we obtain outcome $i$, and the corresponding post-measurement state is $|i\>\<i|\otimes M_i\rho M_i^\dag$.
\end{lem}

In general, a POVM measurement can be implemented by the scheme of Lemma~\ref{lem:equivalent-measure}.
So, $U_M$ can be understood as necessary information to implement the measurement $M$.

Applying the above lemma to $\kif$-statements, we obtain: 
\begin{lem}[Equivalent implementation of if statement]
\label{lem:equivalent-if-statement}
    For an arbitrary if statement $\iif{M}{\qbar}{m}{P_m}$ where $M = \{M_i\}_{i\in I}$, let $q_a$ be an arbitrary fresh auxiliary variable with Hilbert space of dimension $|I|$ and initialised in basis state $|0\>$.
    Then:  
    \[
        \iif{M}{\qbar}{m}{P_m} \equiv U_M[\qbar,q_a]; \iif{}{q_a}{m}{P_m}.
    \]
    More precisely, if the auxiliary system is taken into account, then:
    \[
        \iinit{q_a}{}; \iif{M}{\qbar}{m}{P_m} \equiv 
        \iinit{q_a}{}; U_M[\qbar,q_a]; \iif{}{q_a}{m}{P_m}; \iinit{q_a}{}.
    \]
\end{lem}
\begin{proof}
    This is a direct calculation of the denotational semantics.
\end{proof}

The following proposition presents some generalised versions of laws (Sequentiality) and (Left-Distributivity) in Proposition~\ref{sequential-laws}, a special case of law (if-Expansion) in Proposition~\ref{init-laws}, as well as laws for \(\gateSWAP\).
They are needed in the proof of Theorem~\ref{thm-defer}.
\begin{prop}
\label{prop:extra-law} 
    \begin{enumerate}
        \item \label{prop:extra-law-if-splitting} (Sequentiality)
        Let $M$ be a projective measurement.
        If for all $m$, $\qbar\cap \qvar(P_m) = \emptyset$, then
        \[
            \iif{M}{\qbar}{m}{P_m; Q_m} \equiv 
            \iif{M}{\qbar}{m}{P_m}; \iif{M}{\qbar}{m}{Q_m}.
        \]
        \item \label{prop:extra-law-left-distributivity} (Left-Distributivity)
        If $\qbar\cap \qvar(P) = \emptyset$, then
        \[
            P; \iif{M}{\qbar}{m}{P_m} \equiv
            \iif{M}{\qbar}{m}{P; P_m}.
        \]
        \item \label{prop:extra-law-if-restrictive} (If-Restrictive)
        If $\qbar\cap \qvar(P) = \emptyset$, then
        \[
            \iif{M}{\qbar}{m}{P}; \iinit{\qbar}{|\phi\>} \equiv \iinit{\qbar}{|\phi\>}; P.
        \]
        \item \label{prop:extra-law-swapC} (\(\gateSWAP\)-Commutative)
        $\gateSWAP[\qbar_1,\qbar_2] \equiv \gateSWAP[\qbar_2,\qbar_1]$.
        Furthermore, if $\qvar(P)\subseteq \qbar_1$, then
        \[
            \gateSWAP[\qbar_1,\qbar_2]; P \equiv P[\qbar_2/\qbar_1]; \gateSWAP[\qbar_1,\qbar_2].
        \]
        \item \label{prop:extra-law-swap_cancel} (\(\gateSWAP\)-Cancellation)
        \[
            \iinit{\qbar_1}{|\phi_1\>}; \iinit{\qbar_2}{|\phi_2\>}; \gateSWAP[\qbar_1,\qbar_2] \equiv \iinit{\qbar_1}{|\phi_2\>}; \iinit{\qbar_2}{|\phi_1\>}.
        \]
    \end{enumerate}
\end{prop}
\begin{proof}
    Similar to the proofs of Proposition~\ref{sequential-laws}.
\end{proof}

Now we are ready to present the detailed proof of Theorem~\ref{thm-defer}. 
\begin{proof}[Proof of Theorem~\ref{thm-defer}]
    We proceed by induction on the structure of $P$.
    The strategy is that we prove Eq. (\ref{defer-2}) and $\qvar(C)\subseteq \qvar(P)\cup \qbar_a$ simultaneously. 

    \textbf{Case 1}.
    $P = \iskip$ or $\iabort$ or $P = C$ (quantum circuit).
    We introduce an auxiliary qubit $q_a$.
    Then:
    \begin{align*}
        & \iinit{q_a}{}; \bt{\iskip} \equiv \iinit{q_a}{}; \bt{\iskip; \iif{}{q_a}{m}{\iskip}}; \iinit{q_a}{}, \\
        & \iinit{q_a}{}; \bt{\iabort} \equiv \iinit{q_a}{}; \bt{\iskip; \iif{}{q_a}{m}{\iabort}}; \iinit{q_a}{}, \\
        & \iinit{q_a}{}; \bt{C} \equiv \iinit{q_a}{}; \bt{C; \iif{}{q_a}{m}{\iskip}}; \iinit{q_a}{}.
    \end{align*}
    It is worth noting that, in order to minimise the number of auxiliary qubits, it is not necessary to introduce $q_a$ for each $\iskip$, or $\iabort$, or $C$ if they are embedded within a branch of $\kif$-statement or sequential composition.

    \textbf{Case 2}.
    $P \equiv \iinit{q}{|\phi\>}$.
    We introduce an auxiliary quantum variable $q_a$ that has the same dimension as $q$.
    Let $U_{|\phi\>}$ be the unitary that prepares $|\phi\>$, i.e., $U_{|\phi\>}|0\> = |\phi\>$, noting that
    \[
        \iinit{q_a}{}; \bt{ \iinit{q}{|\phi\>}} \equiv \iinit{q_a}{}; \bt{(U_{|\phi\>}[q_a]; \gateSWAP[q_a,q]); \iif{}{q_a}{m}{\iskip}}; \iinit{q_a}{}
    \]
    by observing the following derivation:
    \begin{align*}
        RHS \equiv &\ \iinit{\rt{q_a}{}; } (\rt{U_{|\phi\>}[q_a]}; \gateSWAP[q_a,q]); \bt{\iif{}{q_a}{m}{\iskip}; \iinit{q_a}{}} \\
        \equiv &\ \rt{ \iinit{q_a}{}; U_{|\phi\>}[q_a]}; \gateSWAP[q_a,q]; \bt{ \iinit{q_a}{}; \iskip} \\
        \equiv &\ \iinit{q_a}{|\phi\>}; \bt{\gateSWAP[q_a,q]; \iinit{q_a}{}} \\
        \equiv &\ \bt{ \iinit{q_a}{|\phi\>}; \iinit{q}{}; \gateSWAP[q_a,q]} \\
        \equiv &\ \iinit{q_a}{}; \bt{ \iinit{q}{|\phi\>}}.
    \end{align*}
    
    \textbf{Case 3}.
    $P = P_1; P_2$.
    By the induction hypothesis, we assume that $P_1,P_2$ have the following normal forms:
    \begin{align*}
        & \iinit{\qbar_{a1}}{}; P_1 \equiv \iinit{\qbar_{a1}}{}; C_1; \iif{}{\qbar_{a1}}{m_1}{P_{m_1}}; \iinit{\qbar_{a1}}{} \\
        & \iinit{\qbar_{a2}}{}; P_2 \equiv \iinit{\qbar_{a2}}{}; C_2; \iif{}{\qbar_{a2}}{m_2}{P_{m_2}}; \iinit{\qbar_{a2}}{}.
    \end{align*}
    Since we have sufficient auxiliary variables and the construction does not depend on specific auxiliary variables, we can use disjoint auxiliary variables for each normal form, ensuring that $\qbar_{a1}\cap \qbar_{a2} = \emptyset$.
    Set $\qbar_a \triangleq \qbar_{a1},\qbar_{a2}$.
    Then we claim:
    \[
        \iinit{\qbar_{a}}{}; \bt{(P1; P2)} \\
        \equiv \iinit{\qbar_{a}}{}; \bt{(C_1; C_2); \iif{}{\qbar_{a}}{(m_1, m_2)}{P_{m_1m_2}}}; \iinit{\qbar_{a}}{}
    \]
    where $P_{m_1m_2} = \iskip$ if both $P_{m_1} = P_{m_2} = \iskip$, and $P_{m_1m_2} = \iabort$ otherwise (i.e., if $P_{m_1} =\iabort$ or $P_{m_2} =\iabort$).
    This can be proved by observing that:
    \begin{align*}
        & \iinit{\qbar_{a}}{}; \bt{(P1; P2)} \\
        \equiv\ &\iinit{\qbar_{a}}{}; \bt{ \iinit{\qbar_{a1}}{}; \iinit{\qbar_{a2}}{}; \bt{P1; P2}} \\
        \equiv\ &\iinit{\qbar_{a}}{}; \bt{(\iinit{\qbar_{a1}}{}; P1)}; \rt{(\iinit{\qbar_{a2}}{}; P2)} \\
        \equiv\ &\bt{ \iinit{\qbar_{a}}{}; \iinit{\qbar_{a1}}{}}; C_1; \iif{}{\qbar_{a1}}{m_1}{P_{m_1}}; \rt{ \iinit{\qbar_{a1}}{}}; \\
        & \bt{ \iinit{\qbar_{a2}}{}}; C_2; \iif{}{\qbar_{a2}}{m_2}{P_{m_2}}; \rt{ \iinit{\qbar_{a2}}{}} \\
        \equiv\ &\iinit{\qbar_{a}}{}; C_1; \bt{\iif{}{\qbar_{a1}}{m_1}{P_{m_1}}}; \\
        & \bt{C_2}; \iif{}{\qbar_{a2}}{m_2}{P_{m_2}}; \iinit{\qbar_{a}}{} \\
        \equiv &\iinit{\qbar_{a}}{}; (C_1; C_2); \bt{\iif{}{\qbar_{a1}}{m_1}{P_{m_1}}; } \\
        & \bt{\iif{}{\qbar_{a2}}{m_2}{P_{m_2}}}; \iinit{\qbar_{a}}{} \\
        \equiv\ &\iinit{\qbar_{a}}{}; (C_1; C_2); \bt{\iif{}{\qbar_{a1}}{m_1}{\iif{}{\qbar_{a2}}{m_2}{P_{m_1}; P_{m_2}}}};
        \iinit{\qbar_{a}}{} \\
        \equiv\ &\iinit{\qbar_{a}}{}; (C_1; C_2); \iif{}{\qbar_{a1},\qbar_{a2}}{(m_1, m_2)}{P_{m_1}; P_{m_2}}; \iinit{\qbar_{a}}{} \\
        \equiv\ &\iinit{\qbar_{a}}{}; \bt{(C_1; C_2); \iif{}{\qbar_{a}}{(m_1, m_2)}{P_{m_1m_2}}}; \iinit{\qbar_{a}}{}
    \end{align*}
    For the variable sets, it suffices to note that 
    \[
        \qvar(C_1; C_2) = \qvar(C_1)\cup \qvar(C_2)\subseteq (\qvar(P_1)\cup \qbar_{a1})\cup(\qvar(P_2)\cup \qbar_{a2}) = (\qvar(P_1; P_2)\cup \qbar_{a}).
    \]

    \textbf{Case 4}.
    $P = \iif{M}{\qbar}{i}{P_i}$.
    This case was already proved in Section~\ref{sec-app}.
    What remains is to check the condition on variable sets.
    Indeed, we have:
    \begin{align*}
        & \qvar(U_M[\qbar,a]; \iqif{a}{\square |i\>}{C_i}) 
        = \qbar\cup a\cup\bigcup\nolimits_i\qvar(C_i) \\
        \subseteq\ & q\cup \bigcup\nolimits_i\left(\qvar(P_i)\cup \rbar_i\right) \cup a 
        = \left(q\cup\bigcup\nolimits_i\qvar(P_i)\right)\cup \left(a\cup \bigcup\nolimits_i\rbar_i\right)
        = \qvar(P) \cup \qbar_a.
    \end{align*}
\end{proof}

\section{A Discussion about Preservation of Refinement by Loop}\label{ref-loop}

In Section~\ref{Sec-Refine}, we showed that refinement is preserved by sequential composition, $\kif$-statement and nondeterministic choice. Here, we further conjecture:

\begin{conjecture}[Refinement rule for loops]
    If $P\sqsubseteq Q$, then 
    \[
        \swhile{M}{P}\sqsubseteq \swhile{M}{Q}.
    \]
\end{conjecture}

At this point, however, we have not yet proved it.
Let us examine where the difficulty arises.
By definition, we have: 
\[
    \sem{\swhile{M}{P}}=\left\{\sum_{i=0}^\infty\E_0\circ(\mathcal{F}_i\circ\E_1)\circ...\circ(\mathcal{F}_1\circ\E_1):\mathcal{F}_1,...,\mathcal{F}_i, ...\in\sem{P}\right\}.
\]
Hence, we need to show that for any $\mathcal{F}_1,...,\mathcal{F}_i,...\in\sem{P}$, we have
\[
    \sum_{i=0}^\infty\E_0\circ(\mathcal{F}_i\circ\E_1)\circ...\circ(\mathcal{F}_1\circ\E_1)\in\ \overline{\mathit{Con}(\sem{Q})}.
\]
Indeed, since $P\sqsubseteq Q$, for each $i$, it follows from $\mathcal{F}_i\in\sem{P}$ that $\mathcal{F}_i\in \overline{\mathit{Con}(\sem{Q})}$, and 
\[
    \mathcal{F}_i=\sum_{j_i}p_{ij_i}G_{ij_i}
\]
for some (finite or countably infinite) probability distribution $\{p_{ij_i}\}$ and $G_{ij_i}\in\sem{Q}.$
Therefore, we have: 
\begin{align}\label{final-swap}
    &\sum_{i=0}^{\infty}\E_0\circ(\mathcal{F}_i\circ\E_1)\circ...\circ(\mathcal{F}_1\circ\E_1) =\sum_{i=0}^\infty\sum_{j_1,...,j_i}\left(\prod_{t=1}^ip_{ti_t}\right)\mathcal{E}_0\circ\left(\mathcal{G}_{ij_i}\circ\mathcal{E}_1\right)\circ\cdots\circ\left(\mathcal{G}_{1j_1}\circ\mathcal{E}_1\right)
\end{align}
However, it is not immediately clear how to interchange the order of big product and big sum; that is, it is not known whether the RHS of (\ref{final-swap}) is equivalent to $\sum_{j_1,j_2,...}p_{j_1,j_2,...}\sum_{i=0}^\infty\mathcal{E}_0\circ\left(\mathcal{G}_{ij_i}\circ\mathcal{E}_1\right)\circ\cdots\circ\left(\mathcal{G}_{1j_1}\circ\mathcal{E}_1\right)$ for some probability distributions $\{p_{j_1,j_2,\cdots}\}$.

\section{A ZX-Calculus Proof of Example~\ref{exam:circuit opt}}
\label{app:zx}
In this section, we demonstrate how Example~\ref{exam:circuit opt} can be proved using the ZX-calculus, thus providing a concrete comparison with our approach.  
The rules we use are from the standard ZX-calculus, as shown in Figure~\ref{fig:zx-rules}, drawn from~\cite{PyZX}.

\begin{figure}[ht]
\centering
\begin{tabular}{|c|}
\hline \\[-0.2cm]
\quad \scalebox{0.93}{\begin{tikzpicture}
	\begin{pgfonlayer}{nodelayer}
		\node [style=Z phase dot] (0) at (-10.5, 0.75) {$\beta$};
		\node [style=none] (1) at (-9.5, 1.25) {};
		\node [style=none] (2) at (-12.25, 1.25) {};
		\node [style=none] (3) at (-12.25, 0) {};
		\node [style=none] (4) at (-9.5, 0) {};
		\node [style=none] (5) at (-9.5, 1) {};
		\node [style=none] (6) at (-12.25, 1) {};
		\node [style=none, rotate=90] (7) at (-12.5, 0.5) {...};
		\node [style=none, rotate=90] (8) at (-9.5, 0.5) {...};
		\node [style=none] (9) at (-13, 2.5) {};
		\node [style=Z phase dot] (10) at (-12, 2.25) {$\alpha$};
		\node [style=none] (11) at (-10, 2.75) {};
		\node [style=none] (12) at (-13, 2.75) {};
		\node [style=none] (13) at (-10, 1.5) {};
		\node [style=none, rotate=90] (14) at (-10, 2) {...};
		\node [style=none] (15) at (-10, 2.5) {};
		\node [style=none] (16) at (-13, 1.5) {};
		\node [style=none, rotate=90] (17) at (-12.75, 2) {...};
		\node [style=none] (18) at (-9.5, 1.5) {};
		\node [style=none] (19) at (-9.5, 2.75) {};
		\node [style=none] (20) at (-9.5, 2.5) {};
		\node [style=none] (21) at (-13, 1) {};
		\node [style=none] (22) at (-13, 0) {};
		\node [style=none] (23) at (-13, 1.25) {};
		\node [style=none] (24) at (-8.5, 1.5) {$=$};
		\node [style=none, rotate=90] (25) at (-11.25, 1.5) {...};
		\node [style=none] (26) at (-7.5, 2.75) {};
		\node [style=none] (27) at (-4, 2.25) {};
		\node [style=none, rotate=90] (28) at (-7, 1.25) {...};
		\node [style=none] (29) at (-4, 0) {};
		\node [style=none, rotate=90] (30) at (-4.25, 1.25) {...};
		\node [style=none] (31) at (-7.5, 2.25) {};
		\node [style=Z phase dot] (32) at (-5.75, 1.5) {$\ \alpha\!+\!\beta\ $};
		\node [style=none] (33) at (-4, 2.75) {};
		\node [style=none] (34) at (-7.5, 0) {};
		\node [style=none] (35) at (-8.5, 2.25) {\footnotesize $(\bm f)$};
		\node [style=Z phase dot] (36) at (-6.25, -3) {$\ (\textrm{-}1)^a \alpha\ $};
		\node [style=none] (37) at (-4, -1.25) {};
		\node [style=none] (38) at (-8.5, -3) {$=$};
		\node [style=none] (39) at (-7.5, -3) {};
		\node [style=none] (40) at (-4, -4.25) {};
		\node [style=X phase dot] (41) at (-4.75, -2.25) {$a\pi$};
		\node [style=X phase dot] (42) at (-4.75, -4.25) {$a\pi$};
		\node [style=X phase dot] (43) at (-12, -3) {$a \pi$};
		\node [style=none] (44) at (-9.5, -1.5) {};
		\node [style=Z phase dot] (45) at (-11, -3) {$\alpha$};
		\node [style=none, rotate=90] (46) at (-9.75, -3.25) {...};
		\node [style=none] (47) at (-4, -2.25) {};
		\node [style=none, rotate=90] (48) at (-4.75, -3.25) {...};
		\node [style=none] (49) at (-9.5, -2.5) {};
		\node [style=none] (50) at (-12.75, -3) {};
		\node [style=none] (51) at (-9.5, -4.25) {};
		\node [style=X phase dot] (52) at (-4.75, -1.25) {$a\pi$};
		\node [style=none] (53) at (-8.5, -2.25) {\footnotesize $(\bm \pi)$};
		\node [style=X phase dot] (54) at (3, -1.25) {$a\pi$};
		\node [style=none, rotate=90] (55) at (3.75, -3) {...};
		\node [style=X phase dot] (56) at (3, -2.25) {$a\pi$};
		\node [style=Z phase dot] (57) at (-0.75, -2.75) {$\alpha$};
		\node [style=none] (58) at (0.5, -3.75) {};
		\node [style=none] (59) at (1.5, -2.75) {$=$};
		\node [style=none] (60) at (4, -3.75) {};
		\node [style=none] (61) at (0.5, -2.25) {};
		\node [style=none] (62) at (4, -1.25) {};
		\node [style=none, rotate=90] (63) at (0.25, -3) {...};
		\node [style=none] (64) at (4, -2.25) {};
		\node [style=X phase dot] (65) at (3, -3.75) {$a\pi$};
		\node [style=none] (66) at (1.5, -2) {\footnotesize $(\bm c)$};
		\node [style=none] (67) at (0.5, -1.25) {};
		\node [style=X phase dot] (68) at (-1.75, -2.75) {$a\pi$};
		\node [style=hadamard] (69) at (4.25, 2.5) {};
		\node [style=Z phase dot] (70) at (-1, 1.25) {$\alpha$};
		\node [style=none] (71) at (2.25, 1.25) {};
		\node [style=none, rotate=90] (72) at (4.25, 1) {...};
		\node [style=none] (73) at (1.5, 1.25) {$=$};
		\node [style=hadamard] (74) at (4.25, 0.25) {};
		\node [style=none] (75) at (-2.5, 1.25) {};
		\node [style=X phase dot] (76) at (3.25, 1.25) {$\alpha$};
		\node [style=none] (77) at (4.75, 2.5) {};
		\node [style=hadamard] (78) at (-1.75, 1.25) {};
		\node [style=none] (79) at (4.75, 1.75) {};
		\node [style=hadamard] (80) at (4.25, 1.75) {};
		\node [style=none] (81) at (0.5, 0.25) {};
		\node [style=none] (82) at (4.75, 0.25) {};
		\node [style=none, rotate=90] (83) at (0.25, 1) {...};
		\node [style=none] (84) at (0.5, 2.5) {};
		\node [style=none] (85) at (0.5, 1.75) {};
		\node [style=none] (86) at (1.5, 2) {\footnotesize $(\bm{h})$};
		\node [style=Z dot] (87) at (7, 2.5) {};
		\node [style=none] (88) at (6.25, 2.5) {};
		\node [style=none] (89) at (7.75, 2.5) {};
		\node [style=none] (90) at (9.75, 2.5) {};
		\node [style=none] (91) at (10.75, 2.5) {};
		\node [style=none] (92) at (8.75, 3.25) {\footnotesize $(\bm{i1})$};
		\node [style=none] (93) at (8.75, 2.5) {$=$};
		\node [style=none] (94) at (8.75, 0.25) {$=$};
		\node [style=none] (95) at (8.75, 1) {\footnotesize $(\bm{i2})$};
		\node [style=none] (96) at (10.75, 0.25) {};
		\node [style=none] (97) at (7.75, 0.25) {};
		\node [style=none] (98) at (6.25, 0.25) {};
		\node [style=none] (99) at (9.75, 0.25) {};
		\node [style=hadamard] (100) at (6.75, 0.25) {};
		\node [style=hadamard] (101) at (7.25, 0.25) {};
		\node [style=none] (102) at (8.75, -1.75) {\footnotesize $(\bm b)$};
		\node [style=X dot] (103) at (11.75, -3.25) {};
		\node [style=none] (104) at (12.5, -1.75) {};
		\node [style=X dot] (105) at (6, -2.5) {};
		\node [style=Z dot] (106) at (7.25, -2.5) {};
		\node [style=none] (107) at (9.5, -3.25) {};
		\node [style=none] (108) at (5.25, -1.75) {};
		\node [style=none] (109) at (9.5, -1.75) {};
		\node [style=Z dot] (110) at (10.25, -1.75) {};
		\node [style=none] (111) at (8, -1.75) {};
		\node [style=none] (112) at (8.75, -2.5) {$=$};
		\node [style=none] (113) at (5.25, -3.25) {};
		\node [style=X dot] (114) at (11.75, -1.75) {};
		\node [style=none] (115) at (12.5, -3.25) {};
		\node [style=Z dot] (116) at (10.25, -3.25) {};
		\node [style=none] (117) at (8, -3.25) {};
	\end{pgfonlayer}
	\begin{pgfonlayer}{edgelayer}
		\draw [in=-141, out=0, looseness=0.75] (3.center) to (0);
		\draw [in=180, out=-39, looseness=0.75] (0) to (4.center);
		\draw [in=180, out=22, looseness=0.75] (0) to (5.center);
		\draw [in=180, out=39, looseness=0.75] (0) to (1.center);
		\draw [in=0, out=180] (0) to (6.center);
		\draw [in=165, out=0] (2.center) to (0);
		\draw [in=-141, out=0, looseness=0.75] (16.center) to (10);
		\draw [in=180, out=-15] (10) to (13.center);
		\draw [in=180, out=22] (10) to (15.center);
		\draw [in=180, out=39, looseness=0.75] (10) to (11.center);
		\draw [in=0, out=158, looseness=0.75] (10) to (9.center);
		\draw [in=141, out=0, looseness=0.75] (12.center) to (10);
		\draw [bend left] (10) to (0);
		\draw (11.center) to (19.center);
		\draw (15.center) to (20.center);
		\draw (13.center) to (18.center);
		\draw (23.center) to (2.center);
		\draw (21.center) to (6.center);
		\draw (22.center) to (3.center);
		\draw [bend left] (0) to (10);
		\draw [in=-120, out=0] (34.center) to (32);
		\draw [in=180, out=-60] (32) to (29.center);
		\draw [in=180, out=45, looseness=0.75] (32) to (27.center);
		\draw [in=180, out=60] (32) to (33.center);
		\draw [in=0, out=135, looseness=0.75] (32) to (31.center);
		\draw [in=120, out=0] (26.center) to (32);
		\draw [in=180, out=-60, looseness=0.75] (45) to (51.center);
		\draw [in=180, out=45, looseness=0.75] (45) to (49.center);
		\draw [in=180, out=75, looseness=0.75] (45) to (44.center);
		\draw [in=180, out=0, looseness=0.50] (43) to (45);
		\draw (50.center) to (43);
		\draw [in=-60, out=180, looseness=0.75] (42) to (36);
		\draw [in=0, out=180, looseness=0.75] (36) to (39.center);
		\draw [in=180, out=30] (36) to (41);
		\draw [in=60, out=180, looseness=0.75] (52) to (36);
		\draw (37.center) to (52);
		\draw (47.center) to (41);
		\draw (40.center) to (42);
		\draw [in=180, out=-60, looseness=0.75] (57) to (58.center);
		\draw [in=180, out=45, looseness=0.75] (57) to (61.center);
		\draw [in=180, out=75, looseness=0.75] (57) to (67.center);
		\draw [in=180, out=0, looseness=0.50] (68) to (57);
		\draw (62.center) to (54);
		\draw (64.center) to (56);
		\draw (60.center) to (65);
		\draw [in=180, out=-60, looseness=0.75] (70) to (81.center);
		\draw [in=180, out=45, looseness=0.75] (70) to (85.center);
		\draw [in=180, out=75, looseness=0.75] (70) to (84.center);
		\draw [in=180, out=0, looseness=0.75] (78) to (70);
		\draw (75.center) to (78);
		\draw [in=-60, out=180, looseness=0.75] (74) to (76);
		\draw [in=0, out=180, looseness=0.75] (76) to (71.center);
		\draw [in=180, out=45, looseness=0.75] (76) to (80);
		\draw [in=75, out=180, looseness=0.75] (69) to (76);
		\draw (77.center) to (69);
		\draw (79.center) to (80);
		\draw (82.center) to (74);
		\draw (88.center) to (89.center);
		\draw (90.center) to (91.center);
		\draw (98.center) to (97.center);
		\draw (99.center) to (96.center);
		\draw (116) to (114);
		\draw (103) to (110);
		\draw (110) to (114);
		\draw (116) to (103);
		\draw (115.center) to (103);
		\draw (107.center) to (116);
		\draw (110) to (109.center);
		\draw (114) to (104.center);
		\draw [bend right] (113.center) to (105);
		\draw [bend right] (105) to (108.center);
		\draw (105) to (106);
		\draw [bend right] (106) to (117.center);
		\draw [bend left] (106) to (111.center);
	\end{pgfonlayer}
\end{tikzpicture}} \quad\ \\[2.3cm]
\hline 
\end{tabular}
\caption{\label{fig:zx-rules} A convenient presentation of the ZX-calculus rules from~\cite{PyZX}. These rules hold
for all $\alpha, \beta \in [0, 2\pi)$, and each rule also holds with the colours interchanged.}
\end{figure}

By rewriting the circuits as ZX-diagrams, the equivalence of the two circuits is translated into the equivalence of the following two ZX-diagrams:
\begin{equation}
\label{eqn: zx-example}
\begin{tikzpicture}
    \begin{pgfonlayer}{nodelayer}
        \node [style=none] (0) at (0.00, 0.00) {};
        \node [style=none] (1) at (0.00, -1.50) {};
        \node [style=none] (2) at (0.00, -3.00) {};
        \node [style=Z dot] (3) at (1.00, 0.00) {};
        \node [style=X dot] (4) at (1.00, -3.00) {};
        \node [style=Z phase dot] (5) at (2.00, 0.00) {$\pi$};
        \node [style=Z dot] (6) at (2.00, -1.50) {};
        \node [style=X dot] (7) at (2.00, -3.00) {};
        \node [style=Z dot] (8) at (3.00, 0.00) {};
        \node [style=X dot] (9) at (3.00, -1.50) {};
        \node [style=Z phase dot] (10) at (4.00, 0.00) {$\pi$};
        \node [style=X phase dot] (11) at (4.00, -1.50) {$\pi$};
        \node [style=X phase dot] (12) at (4.00, -3.00) {$\pi$};
        \node [style=none] (13) at (5.00, 0.00) {};
        \node [style=none] (14) at (5.00, -1.50) {};
        \node [style=none] (15) at (5.00, -3.00) {};
    \end{pgfonlayer}
    \begin{pgfonlayer}{edgelayer}
        \draw (0) to (3);
        \draw (1) to (6);
        \draw (2) to (4);
        \draw (3) to (4);
        \draw (3) to (5);
        \draw (4) to (7);
        \draw (5) to (8);
        \draw (6) to (7);
        \draw (6) to (9);
        \draw (7) to (12);
        \draw (8) to (9);
        \draw (8) to (10);
        \draw (9) to (11);
        \draw (10) to (13);
        \draw (11) to (14);
        \draw (12) to (15);
    \end{pgfonlayer}
\end{tikzpicture}\quad
\begin{tikzpicture}
    \begin{pgfonlayer}{nodelayer}
        \node [style=none] (1) at (0.00, -1.50) {$=$};
    \end{pgfonlayer}
    \begin{pgfonlayer}{edgelayer}
    \end{pgfonlayer}
\end{tikzpicture}\quad
\begin{tikzpicture}
    \begin{pgfonlayer}{nodelayer}
        \node [style=none] (0) at (0.00, 0.00) {};
        \node [style=none] (1) at (0.00, -1.50) {};
        \node [style=none] (2) at (0.00, -3.00) {};
        \node [style=Z dot] (3) at (1.00, 0.00) {};
        \node [style=X dot] (4) at (1.00, -1.50) {};
        \node [style=X phase dot] (5) at (2.00, -1.50) {$\pi$};
        \node [style=Z dot] (6) at (3.00, -1.50) {};
        \node [style=X dot] (7) at (3.00, -3.00) {};
        \node [style=none] (8) at (4.00, 0.00) {};
        \node [style=none] (9) at (4.00, -1.50) {};
        \node [style=none] (10) at (4.00, -3.00) {};
    \end{pgfonlayer}
    \begin{pgfonlayer}{edgelayer}
        \draw (0) to (3);
        \draw (1) to (4);
        \draw (2) to (7);
        \draw (3) to (4);
        \draw (3) to (8);
        \draw (4) to (5);
        \draw (5) to (6);
        \draw (6) to (7);
        \draw (6) to (9);
        \draw (7) to (10);
    \end{pgfonlayer}
\end{tikzpicture}.
\end{equation}

Applying the fusion rule ($f$) to merge all spiders of the same type, we immediately obtain the simplification of the LHS:
\begin{align*}
\quad&
\begin{tikzpicture}
    \begin{pgfonlayer}{nodelayer}
        \node [style=none] (0) at (0.00, 0.00) {};
        \node [style=none] (1) at (0.00, -1.50) {};
        \node [style=none] (2) at (0.00, -3.00) {};
        \node [style=Z dot] (3) at (1.00, 0.00) {};
        \node [style=Z dot] (4) at (1.00, -1.50) {};
        \node [style=X phase dot] (5) at (3.00, -1.50) {$\pi$};
        \node [style=X phase dot] (6) at (3.00, -3.00) {$\pi$};
        \node [style=none] (7) at (4.00, 0.00) {};
        \node [style=none] (8) at (4.00, -1.50) {};
        \node [style=none] (9) at (4.00, -3.00) {};
        \node [style=none] (24) at (-1.00, -0.90) {($f$)};
        \node [style=none] (24) at (-1.00, -1.50) {$=$};
    \end{pgfonlayer}
    \begin{pgfonlayer}{edgelayer}
        \draw (0) to (3);
        \draw (1) to (4);
        \draw (2) to (6);
        \draw (3) to (5);
        \draw (3) to (6);
        \draw (3) to (7);
        \draw (4) to (5);
        \draw (4) to (6);
        \draw (5) to (8);
        \draw (6) to (9);
    \end{pgfonlayer}
\end{tikzpicture}.
\end{align*}

The simplification of the RHS proceeds in three steps. First, we apply the $\pi$-commutation rule ($\pi$) to swap the intermediate $X$ spider with the subsequent $Z$ spider. Then, we use the strong complementarity rule ($b$) to interchange the $X$ and $Z$ spiders. Finally, we apply the fusion rule ($f$) to obtain the same ZX-diagram as above:
\begin{align*}
\begin{tikzpicture}
    \begin{pgfonlayer}{nodelayer}
        \node [style=none] (0) at (0.00, 0.00) {};
        \node [style=none] (1) at (0.00, -1.50) {};
        \node [style=none] (2) at (0.00, -3.00) {};
        \node [style=Z dot] (3) at (1.00, 0.00) {};
        \node [style=X dot] (4) at (1.00, -1.50) {};
        \node [style=X phase dot] (5) at (2.00, -1.50) {$\pi$};
        \node [style=Z dot] (6) at (3.00, -1.50) {};
        \node [style=X dot] (7) at (3.00, -3.00) {};
        \node [style=none] (8) at (4.00, 0.00) {};
        \node [style=none] (9) at (4.00, -1.50) {};
        \node [style=none] (10) at (4.00, -3.00) {};
    \end{pgfonlayer}
    \begin{pgfonlayer}{edgelayer}
        \draw (0) to (3);
        \draw (1) to (4);
        \draw (2) to (7);
        \draw (3) to (4);
        \draw (3) to (8);
        \draw (4) to (5);
        \draw (5) to (6);
        \draw (6) to (7);
        \draw (6) to (9);
        \draw (7) to (10);
    \end{pgfonlayer}
\end{tikzpicture}\quad
\begin{tikzpicture}
    \begin{pgfonlayer}{nodelayer}
        \node [style=none] (0) at (0.00, 0.00) {};
        \node [style=none] (1) at (0.00, -1.50) {};
        \node [style=none] (2) at (0.00, -3.00) {};
        \node [style=Z dot] (3) at (1.00, 0.00) {};
        \node [style=X dot] (4) at (1.00, -1.50) {};
        \node [style=Z dot] (5) at (2.00, -1.50) {};
        \node [style=X phase dot] (6) at (2.00, -2.25) {$\pi$};
        \node [style=X dot] (7) at (2.00, -3.00) {};
        \node [style=X phase dot] (8) at (3.00, -1.50) {$\pi$};
        \node [style=none] (9) at (4.00, 0.00) {};
        \node [style=none] (10) at (4.00, -1.50) {};
        \node [style=none] (11) at (4.00, -3.00) {};
        \node [style=none] (12) at (-1.00, -0.90) {($\pi$)};
        \node [style=none] (12) at (-1.00, -1.50) {$=$};
        \node [style=none] (13) at (5.00, -0.90) {($b$)};
        \node [style=none] (13) at (5.00, -1.50) {$=$};
    \end{pgfonlayer}
    \begin{pgfonlayer}{edgelayer}
        \draw (0) to (3);
        \draw (1) to (4);
        \draw (2) to (7);
        \draw (3) to (4);
        \draw (3) to (9);
        \draw (4) to (5);
        \draw (5) to (6);
        \draw (5) to (8);
        \draw (6) to (7);
        \draw (7) to (11);
        \draw (8) to (10);
    \end{pgfonlayer}
\end{tikzpicture}\quad&
\begin{tikzpicture}
    \begin{pgfonlayer}{nodelayer}
        \node [style=none] (0) at (0.00, 0.00) {};
        \node [style=none] (1) at (0.00, -1.50) {};
        \node [style=none] (2) at (0.00, -3.00) {};
        \node [style=Z dot] (3) at (1.00, 0.00) {};
        \node [style=X dot] (4) at (1.00, -1.50) {};
        \node [style=Z dot] (5) at (2.00, -0.75) {};
        \node [style=X dot] (6) at (3.00, -1.50) {};
        \node [style=X dot] (7) at (4.00, -0.75) {};
        \node [style=X phase dot] (8) at (4.00, -2.25) {$\pi$};
        \node [style=X dot] (9) at (4.00, -3.00) {};
        \node [style=X phase dot] (10) at (5.00, -1.50) {$\pi$};
        \node [style=none] (11) at (6.00, 0.00) {};
        \node [style=none] (12) at (6.00, -1.50) {};
        \node [style=none] (13) at (6.00, -3.00) {};
    \end{pgfonlayer}
    \begin{pgfonlayer}{edgelayer}
        \draw (0) to (3);
        \draw (1) to (4);
        \draw (2) to (9);
        \draw (3) to (5);
        \draw (3) to (11);
        \draw (4) to (6);
        \draw (4) to (7);
        \draw (5) to (6);
        \draw (5) to (7);
        \draw (6) to (8);
        \draw (7) to (10);
        \draw (8) to (9);
        \draw (9) to (13);
        \draw (10) to (12);
    \end{pgfonlayer}
\end{tikzpicture}\quad
\begin{tikzpicture}
    \begin{pgfonlayer}{nodelayer}
        \node [style=none] (0) at (0.00, 0.00) {};
        \node [style=none] (1) at (0.00, -1.50) {};
        \node [style=none] (2) at (0.00, -3.00) {};
        \node [style=Z dot] (3) at (1.00, 0.00) {};
        \node [style=Z dot] (4) at (1.00, -1.50) {};
        \node [style=X phase dot] (5) at (3.00, -1.50) {$\pi$};
        \node [style=X phase dot] (6) at (3.00, -3.00) {$\pi$};
        \node [style=none] (7) at (4.00, 0.00) {};
        \node [style=none] (8) at (4.00, -1.50) {};
        \node [style=none] (9) at (4.00, -3.00) {};
        \node [style=none] (24) at (-1.00, -0.90) {($f$)};
        \node [style=none] (24) at (-1.00, -1.50) {$=$};
    \end{pgfonlayer}
    \begin{pgfonlayer}{edgelayer}
        \draw (0) to (3);
        \draw (1) to (4);
        \draw (2) to (6);
        \draw (3) to (5);
        \draw (3) to (6);
        \draw (3) to (7);
        \draw (4) to (5);
        \draw (4) to (6);
        \draw (5) to (8);
        \draw (6) to (9);
    \end{pgfonlayer}
\end{tikzpicture}.
\end{align*}

\end{document}